\documentclass[10pt,journal,twocolumn,final,twoside]{IEEEtran}
\usepackage{ifpdf}
\usepackage{cite}
\usepackage{graphicx}
\usepackage[cmex10]{amsmath}
\usepackage{algorithmic}
\usepackage{array}
\usepackage{mdwmath}
\usepackage{mdwtab}
\usepackage{eqparbox}
\usepackage[tight,footnotesize]{subfigure}
\usepackage{stfloats}
\usepackage{amssymb}
\usepackage{bm}
\usepackage{overpic}
\usepackage{color}
\usepackage{dsfont}
\usepackage{cases}
\usepackage{xr}
\usepackage{threeparttable}
\usepackage{multirow}
\usepackage{cancel}
\usepackage{tabularx}

\newtheorem{theorem}{Theorem}
\newtheorem{lemma}{Lemma}
\newtheorem{corollary}{Corollary}
\newtheorem{definition}{Definition}

\newtheorem{assumption}{Assumption}
\newtheorem{example}{Example}
\newtheorem{remark}{Remark}

\def\nn{\nonumber}

\def\Expt{\mathbb{E}}
\def\E{\mathbb{E}}
\def\mb{\mathbb}
\def\defeq{\triangleq}
\def\mc{\mathcal}
\def\col{\mathrm{col}}

\def\diag{\mathrm{diag}}

\def\Tr{\mathrm{Tr}}
\def\bP{\bar{P}}

\def\one{\mathds{1}}

\def\bs{\bm{s}}

\def\bv{\bm{v}}
\def\bw{\bm{w}}

\def\bz{\bm{z}}

\def\mA{\mc{A}}
\def\mD{\mc{D}}
\def\mF{\mc{F}}
\def\mM{\mc{M}}

\def\mS{\mc{S}}

\def\mU{\mc{U}}
\def\mW{\mc{W}}
\def\mZ{\mc{Z}}

\def\s{\bm{s}}

\def\u{\bm{u}}
\def\w{\bm{w}}

\def\z{\bm{z}}

\allowdisplaybreaks

\title{On the Learning Behavior of Adaptive Networks --- Part I: Transient Analysis}

\begin{document}

\author{Jianshu~Chen,~\IEEEmembership{Member,~IEEE,}%
        ~and~Ali~H.~Sayed,~\IEEEmembership{Fellow,~IEEE}
\thanks{
Manuscript received December 28, 2013; revised November 21, 2014; accepted
April 08, 2015. This work was supported in part by NSF grants CCF-1011918 and ECCS-1407712.
A short and limited version of this work appears in the conference publication%
\cite{chen2012AllertonLimit} without proofs and under more restrictive conditions than considered in this broader and expanded work.}
\thanks{J. Chen was with Department of Electrical Engineering, University of California, Los Angeles, and is currently with Microsoft Research, Redmond, WA 98052. This work was performed while he was a PhD student at UCLA. Email: cjs09@ucla.edu. 
}
\thanks{A. H. Sayed is with Department of Electrical Engineering,
University of California, Los Angeles, CA 90095. Email: sayed@ee.ucla.edu.
}
\thanks{Communicated by Prof. Nicol\`{o} Cesa-Bianchi,  Associate Editor for Pattern Recognition, Statistical Learning, and Inference.}
\thanks{Copyright (c) 2014 IEEE.
Personal use of this material is permitted. However, permission to use this material for any other
purposes must be obtained from the IEEE by sending a request to pubs-permissions@ieee.org.
}
}

\markboth{IEEE Transactions on Information Theory, VOL.~XX, NO.~XX, MONTH~2015}{Chen and Sayed: 
On the Learning Behavior of Adaptive Networks --- Part I: Transient Analysis}

\maketitle

\begin{abstract}
This work carries out a detailed transient analysis of the learning behavior of multi-agent networks, and reveals interesting results about the learning abilities of distributed strategies. Among other results, the analysis reveals how combination policies influence the learning process of networked agents, and how these policies can steer the convergence point towards any of many possible Pareto optimal solutions. The results also establish that the learning process of an adaptive network undergoes three (rather than two) well-defined stages of evolution with distinctive convergence rates during the first two stages, while attaining a finite mean-square-error (MSE) level in the last stage. The analysis reveals what aspects of the network topology influence performance directly and suggests design procedures that can optimize performance by adjusting the relevant topology parameters. Interestingly, it is further shown that, in the adaptation regime, each agent in a sparsely connected network is able to achieve the same performance level as that of a centralized stochastic-gradient strategy even for left-stochastic combination strategies. These results lead to a deeper understanding and useful insights on the convergence behavior of coupled distributed learners. The results also lead to effective design mechanisms to help diffuse information more thoroughly over networks.
\end{abstract}
%

\section{INTRODUCTION}
\label{Sec:Intro}

In multi-agent systems, agents interact with each other to solve a 
problem of common interest, such as an optimization problem in a distributed manner. 
Such networks of interacting agents are useful in solving
distributed estimation, learning and decision making problems
\cite{barbarossa2007bio,li2008new, nedic2001incremental,
tsitsiklis1986distributed,kar2011converegence,kar2008distributed,kar2013distributed,dimakis2010gossip,
nedic2009distributed,nedic2010cooperative,eksin2012tspdistributed,
eksin2013learning,
theodoridis2011adaptive,dini2012cooperative,
lopes2008diffusion,Cattivelli10,chen2011TSPdiffopt,zhao2012performance,chen2013JSTSPpareto,sayed2012diffbookchapter,sayed2014proc,sayed2014adaptation,chen2014tspDictLearn,macua2015distributed,chouvardas2011adaptive,gharenhshiran2013jstsp,
ram2010distributed,srivastava2011distributed,lee2013distributed,tsianos2012consensus,
palomar2006tutorial,saligrama2006distributed,predd2009collaborative,boyd2006randomized,ren2005consensus,sardellitti2010fast,jadbabaie2003coordination,olfati2007consensus}.
They are also useful in modeling biological
networks\cite{di2011biotsp,Cattivelli09bird,Tu10mobile},
collective rational behavior\cite{eksin2012tspdistributed,eksin2013learning},
and in developing biologically-inspired designs\cite{barbarossa2007bio,hong2005scalable}.
Two useful strategies that can be used to guide the interactions of
agents over a network are consensus strategies \cite{nedic2009distributed,nedic2010cooperative,tsitsiklis1986distributed,
kar2011converegence,
kar2008distributed,kar2013distributed,dimakis2010gossip}
 and diffusion strategies\cite{lopes2008diffusion,Cattivelli10,chen2011TSPdiffopt,zhao2012performance,chen2013JSTSPpareto,sayed2012diffbookchapter,sayed2014proc,sayed2014adaptation,chen2014tspDictLearn,macua2015distributed,chouvardas2011adaptive,gharenhshiran2013jstsp}.
Both classes of algorithms involve self-learning and social-learning steps.  
During self-learning, each agent updates its state using its local data. 
During social learning, each agent aggregates information from its neighbors. 
A useful feature that results from these localized interactions is that the 
network ends up exhibiting global patterns of behavior.
For example, in distributed estimation and learning, each agent
is able to attain the performance of centralized solutions by 
relying solely on local cooperation
\cite{kar2011converegence,zhao2012performance,sayed2014proc,sayed2014adaptation}. 

In this article, and the accompanying Part II \cite{chen2013learningPart2}, we consider a general class of distributed
strategies, which includes diffusion and consensus updates as special cases,  
and study the resulting global learning behavior by addressing four important questions: 
(i) where  does the
distributed algorithm converge to? (ii) when does it converge? (iii) how fast does it converge?
and (iv) how close does it converge to the intended point? 
We answer questions (i)--(iii) in Part I and question (iv) in Part II \cite{chen2013learningPart2}. We study these four questions by characterizing the learning dynamics of the
network in some great detail. An interesting conclusion that follows from our 
analysis is that the learning curve of a multi-agent system
will be shown to exhibit \emph{three} different phases.
In the first phase (Transient Phase I), the convergence rate of the network is determined by the second 
largest eigenvalue of the combination matrix in magnitude, which is related to 
the degree of network connectivity. In the second phase (Transient Phase II), the convergence 
rate is determined by the entries of the right-eigenvector of the combination matrix
corresponding to the eigenvalue at one. And, in the third phase (the steady-state phase)
the mean-square performance of the algorithm turns out to depend on this same
right-eigenvector in a revealing way. Even more surprisingly,  we shall discover that the agents have
the same learning behavior starting at Transient Phase II,
and are able to achieve a performance level that matches that of a fully connected network
or a centralized stochastic-gradient strategy. Actually, we shall show that the consensus and
diffusion strategies can be represented as perturbed versions
of a centralized {\em reference} recursion in a certain transform
domain. We quantify the effect of the perturbations and 
establish the aforementioned properties for the various phases
of the learning behavior of the networks.
The results will reveal the manner by which the network topology influences performance in some unexpected ways. 

There have been of course many insightful works in the literature on distributed strategies and their convergence behavior. In Sections \ref{Sec:ProblemFormulation:PriorWork} and \ref{Sec:LearningBehavior:Overview} further ahead, we explain in what ways the current manuscript extends these earlier investigations and what novel contributions this work leads to. In particular, it will be seen that several new insights are discovered that clarify how distributed networks learn. For the time being, in these introductory remarks, we limit ourselves to mentioning one important aspect of our development. Most prior studies on distributed optimization and estimation tend
to focus on the performance and convergence of the algorithms
under \emph{diminishing} step-size conditions\cite{tsitsiklis1986distributed,
kar2011converegence,kar2008distributed,kar2013distributed,dimakis2010gossip,ram2010distributed,
srivastava2011distributed,nedic2009distributed,lee2013distributed, bianchi2012performance}, or on convergence
under deterministic  conditions on the data\cite{nedic2009distributed}. This is perfectly 
fine for applications involving {\em static} optimization problems where the objective is to locate the fixed optimizer of some aggregate cost function of interest. In this paper, however,
we examine the learning behavior of distributed strategies
under \emph{constant} step-size conditions. 
This is because constant step-sizes are necessary to enable continuous
adaptation, learning, and tracking in the presence of streaming data and drifting conditions. These
features would enable
the algorithms to perform well even when the location of the optimizer drifts with time. Nevertheless, the use of constant step-sizes enriches the dynamics of (stochastic-gradient) distributed algorithms 
in that the gradient update term does not die out with time anymore, in clear contrast
to the diminishing step-size case where the influence of the gradient term is annihilated over time due to the decaying value of the step-size parameter. For this reason, 
more care is needed to examine the learning behavior of distributed strategies in the constant step-size regime 
since their updates remain continually active and the effect of gradient noise is always present. 
This work also generalizes and extends in non-trivial ways the studies in \cite{chen2011TSPdiffopt,chen2013JSTSPpareto}. For example, while reference \cite{chen2011TSPdiffopt} assumed that the individual costs of all agents have the {\em same} minimizer, and reference \cite{chen2013JSTSPpareto} assumed that each of these individual costs is strongly convex, these requirements are not needed in the current study: individual costs can have distinct minimizers and they do not even need to be convex (see the discussion after expression \eqref{Equ:ProblemFormulation:J_glob_weigthed}). This fact widens significantly the class of distributed learning problems that are covered by our framework. Moreover, the network behavior is studied under less restrictive assumptions and for broader scenarios, including a close study of the various phases of evolution during the transient phase of the learning process. We also study a larger class of distributed strategies that includes diffusion and consensus strategies as special cases. 

To examine the learning behavior of adaptive networks under broader and more relaxed conditions than usual, we pursue a new analysis route by introducing a {\em reference} centralized recursion and by studying the perturbation of the diffusion and consensus strategies relative to this centralized solution over time. Insightful new results are obtained through this perturbation analysis. For example, we are now able to examine closely {\em both} the transient phase behavior and the steady-state phase behavior of the learning process and to explain how behavior in these two stages relate to the behavior of the centralized solution (see Fig. \ref{Fig:TypicalLearningCurve} further ahead). Among several other results, the mean-square-error expression \eqref{Equ:LearnBehav:MSE} derived later in Part II \cite{chen2013learningPart2} following some careful analysis, which builds on the results of this Part I, is one of the new (compact and powerful) insights; it reveals how the performance of each agent is closely related to that of the centralized stochastic approximation strategy --- see the discussion right after \eqref{Equ:LearnBehav:MSE}. As the reader will ascertain from the derivations in the appendices, arriving at these conclusions for a broad class of distributed strategies and under weaker conditions than usual is demanding and necessitates a careful study of the evolution of the error dynamics over the network and its stability. When all is said and done, Parts I and II \cite{chen2013learningPart2} lead to several novel insights into the learning behavior of adaptive networks.
\\

\noindent
{\bf Notation}. 
All vectors are column vectors. We use boldface letters to denote random quantities (such as $\bm{u}_{k,i}$) and regular font to denote their realizations or deterministic variables (such as $u_{k,i}$). We use $\mathrm{diag}\{x_1,\ldots,x_N\}$ to denote a (block) diagonal matrix consisting of diagonal entries (blocks) $x_1,\ldots,x_N$, and use $\mathrm{col}\{x_1,\ldots,x_N\}$ to denote a column vector formed by stacking $x_1,\ldots,x_N$ on top of each other. The notation $x \preceq y$ means each entry of the vector $x$ is less than or equal to the corresponding entry of the vector $y$, and the notation $X \preceq Y$ means each entry of the matrix $X$ is less than or equal to the corresponding entry of the matrix $Y$. The notation $x=\mathrm{vec}(X)$ denotes the vectorization operation that stacks the columns of a matrix $X$ on top of each other to form a vector $x$, and $X=\mathrm{vec}^{-1}(x)$  is the inverse operation. The operators $\nabla_w$ and $\nabla_{w^T}$ denote the column and row gradient vectors with respect to $w$. When $\nabla_{w^T}$ is applied to a  column vector $s$, it generates a 
matrix. The notation $a(\mu) = O( b(\mu) )$ means that there exists a constant $C>0$ such
that { for all $\mu$,} $a(\mu) \le  C \cdot b(\mu)$.

\section{PROBLEM FORMULATION}
\label{Sec:ProblemFormulation}
\subsection{Distributed Strategies: Consensus and Diffusion}
\label{Sec:ProblemFormulation:Dist}
	
We consider a connected network of $N$ agents that are linked together through a topology --- see Fig. \ref{Fig:Fig_Network}. Each agent $k$ implements a distributed algorithm of the following form to update its state vector from $\bw_{k,i-1}$ to $\bw_{k,i}$:
	\begin{align}
		\label{Equ:ProblemFormulation:DisAlg_Comb1}
		\bm{\phi}_{k,i-1}	&=	\sum_{l=1}^N a_{1,lk} \bm{w}_{l,i-1}		\\
		\label{Equ:ProblemFormulation:DisAlg_Adapt}
		\bm{\psi}_{k,i}		&=	\sum_{l=1}^N a_{0,lk} \bm{\phi}_{l,i-1}
								-
								\mu_k \hat{\bm{s}}_{k,i}(\bm{\phi}_{k,i-1})	\\
		\label{Equ:ProblemFormulation:DisAlg_Comb2}
		\bm{w}_{k,i}		&=	\sum_{l=1}^N a_{2,lk} \bm{\psi}_{l,i}
	\end{align}
where $\bm{w}_{k,i} \in \mb{R}^M$ is the state of agent $k$ at time $i$, usually an estimate for the solution of some optimization problem, $\bm{\phi}_{k,i-1} \in \mb{R}^M$ and $\bm{\psi}_{k,i} \in \mb{R}^M$ are intermediate variables generated at node $k$ before updating to $\bm{w}_{k,i}$, $\mu_k$ is a non-negative constant step-size parameter used by node $k$, and $\hat{\bm{s}}_{k,i}(\cdot)$ is an $M \times 1$ update vector function at node $k$.  In deterministic optimization problems, the update vectors $\hat{\bm{s}}_{k,i}(\cdot)$  can be the gradient or Newton steps associated with the cost functions\cite{nedic2009distributed}. On the other hand, in stocastic approximation problems, such as adaptation, learning and estimation problems \cite{tsitsiklis1986distributed,kar2011converegence,kar2013distributed,kar2008distributed,dimakis2010gossip,
theodoridis2011adaptive,chouvardas2011adaptive,dini2012cooperative,
lopes2008diffusion,Cattivelli10,zhao2012performance,chen2011TSPdiffopt,chen2013JSTSPpareto,sayed2014adaptation,sayed2014proc,macua2015distributed,
ram2010distributed,srivastava2011distributed,sayed2012diffbookchapter,gharenhshiran2013jstsp},
the update vectors are usually computed from realizations of data samples that arrive sequentially at the nodes. In the stochastic setting, the quantities appearing in \eqref{Equ:ProblemFormulation:DisAlg_Comb1}--\eqref{Equ:ProblemFormulation:DisAlg_Comb2}  become random and we use boldface letters to highlight their stochastic nature. In Example \ref{Ex:UpdateVector} below, we illustrate choices for $\hat{\bm{s}}_{k,i}(w)$ in different contexts.

The combination coefficients $a_{1,lk}$, $a_{0,lk}$ and $a_{2,lk}$ in \eqref{Equ:ProblemFormulation:DisAlg_Comb1}--\eqref{Equ:ProblemFormulation:DisAlg_Comb2} are nonnegative weights that each node $k$ assigns to the information arriving from node $l$; these coefficients are required to satisfy:
	\begin{align}
		\label{Equ:ProblemFormulation:a_lk:convex}
		&\sum_{l=1}^N a_{1,lk}=1,	\quad
		\sum_{l=1}^N a_{0,lk}=1,	\quad
		\sum_{l=1}^N a_{2,lk}=1			\\
		\label{Equ:ProblemFormulation:a_lk:nonnegative}
		&a_{1,lk} \ge 0,	\qquad
		a_{0,lk} \ge 0,	\qquad
		a_{2,lk} \ge 0							\\
		\label{Equ:ProblemFormulation:a_lk:localized}
		&a_{1,lk} = a_{2,lk} = a_{0,lk} = 0,\qquad \mathrm{if~}l \notin \mc{N}_k
	\end{align}
Observe from \eqref{Equ:ProblemFormulation:a_lk:localized} that the combination 
coefficients are zero if $l \notin \mc{N}_k$, where $\mc{N}_k$ denotes 
the set of neighbors of node $k$. Therefore, each summation
in \eqref{Equ:ProblemFormulation:DisAlg_Comb1}--\eqref{Equ:ProblemFormulation:DisAlg_Comb2}
is actually confined to the neighborhood of node $k$.
In algorithm
\eqref{Equ:ProblemFormulation:DisAlg_Comb1}--\eqref{Equ:ProblemFormulation:DisAlg_Comb2},
each node $k$ first combines the states $\{\bm{w}_{l,i-1}\}$
from its neighbors and updates $\bm{w}_{k,i-1}$ to 
the intermediate variable $\bm{\phi}_{k,i-1}$. Then,
the $\{\bm{\phi}_{l,i-1}\}$ from the neighbors are aggregated and 
updated to $\bm{\psi}_{k,i}$ along the opposite direction of $\hat{\bm{s}}_{k,i}(\bm{\phi}_{k,i-1})$. 
Finally, the intermediate estimators $\{\bm{\psi}_{l,i}\}$ from the
neighbors are combined to generate the new state $\bm{w}_{k,i}$
at node $k$.

	\begin{figure}[t]
		\centering
		\includegraphics[width=0.28\textwidth]{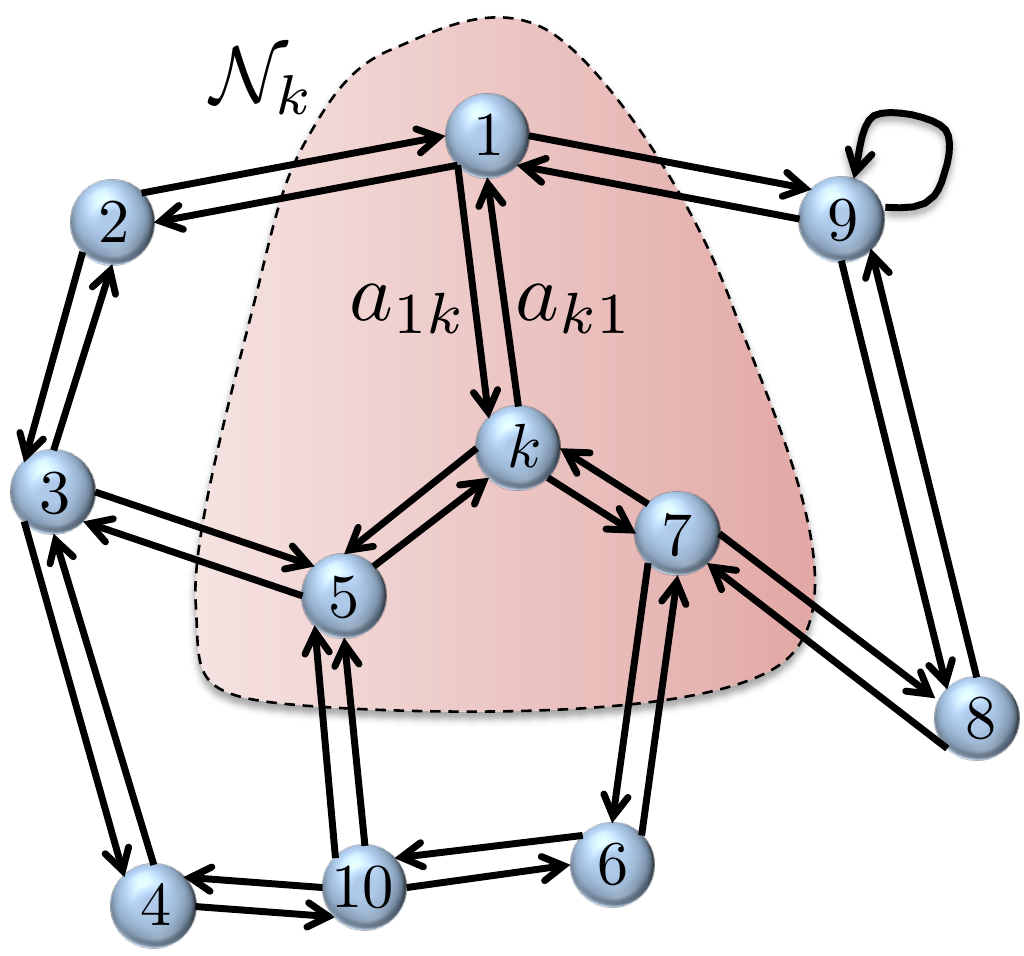}
		\caption{A network representing a multi-agent system. The set of all agents that
		can communicate with node $k$ is 
		denoted by $\mc{N}_k$.The edge linking any two agents is 
		represented by two directed arrows to emphasize that 
		information can flow in both directions.}
		\label{Fig:Fig_Network}
	\end{figure}

		\vspace{0.5em}
	\begin{example}
		\label{Ex:UpdateVector}
		The distributed algorithm
		\eqref{Equ:ProblemFormulation:DisAlg_Comb1}--%
		\eqref{Equ:ProblemFormulation:DisAlg_Comb2}
		can be applied to optimize aggregate costs of the following form:
			\begin{align}
				\label{Equ:ProblemFormulation:Ex:J_glob}
				J^{\mathrm{glob}}(w)	=	\sum_{k=1}^N J_k(w)
			\end{align}
		or to find Pareto-optimal solutions to multi-objective
		optimization problems, such as:
			\begin{align}
				\label{Equ:ProblemFormulation:Multiobjective}
				\min_w\left\{ J_1(w),\ldots,J_N(w) \right\}
			\end{align}
		where $J_k(w)$ is an individual convex cost associated
		with each agent $k$. Optimization problems of the form
		 \eqref{Equ:ProblemFormulation:Ex:J_glob}--\eqref{Equ:ProblemFormulation:Multiobjective}
		arise
		in various applications --- see \cite{li2008new, nedic2001incremental,
tsitsiklis1986distributed,kar2011converegence,kar2008distributed,kar2013distributed,dimakis2010gossip,
nedic2009distributed,nedic2010cooperative,eksin2012tspdistributed,
eksin2013learning,
theodoridis2011adaptive,dini2012cooperative,chouvardas2011adaptive,
sayed2012diffbookchapter,
lopes2008diffusion,Cattivelli10,zhao2012performance,chen2011TSPdiffopt,chen2013JSTSPpareto,
gharenhshiran2013jstsp,sayed2014adaptation,sayed2014proc,chen2014tspDictLearn,macua2015distributed,
ram2010distributed,srivastava2011distributed,lee2013distributed,tsianos2012consensus}.
		Depending on the context, the update vector 
		$\hat{\bm{s}}_{k,i}(\cdot)$ may be chosen in different ways:
			\begin{itemize}
				\item
					In deterministic optimization problems, the expressions
					for $\{J_k(w)\}$ are known and the update vector
					$\hat{\bm{s}}_{k,i}(\cdot)$ at node $k$ is chosen as 
					the deterministic gradient (column) vector $\nabla_w J_k(\cdot)$.
				\item
					In distributed estimation and learning, the individual
					cost function at each node $k$ is usually selected as
					the expected value of some loss function
					$Q_k(\cdot,\cdot)$, i.e., $J_k(w)=\E\{Q_k(w,\bm{x}_{k,i})\}$
					\cite{chen2011TSPdiffopt},
					where the expectation is with respect to the randomness
					in the data samples $\{\bm{x}_{k,i}\}$ collected at 
					node $k$ at time $i$. The exact expression for 
					$\nabla_w J_k(w)$ is usually unknown since 
					the probability distribution of the data is not known beforehand. 
					In these situations, the update vector
					$\hat{\bm{s}}_{k,i}(\cdot)$	is chosen as an instantaneous
					approximation for the true gradient vector, such as, 
					$\hat{\bs}_{k,i}(\cdot)=
					\widehat{\nabla_w J_{k}}(\cdot)=\nabla_w Q_k(\cdot,\bm{x}_{k,i})$, which
					is known as \emph{stochastic gradient}.
					Note that the update vector $\hat{\bm{s}}_{k,i}(w)$
					is now evaluated from the random data sample
					$\bm{x}_{k,i}$.
					Therefore, it is also random and time dependent.
					
			\end{itemize}
		The update vectors $\{\hat{\bs}_{k,i}(\cdot)\}$ may not necessarily 
		be the gradients of cost 
		functions or their stochastic approximations. They may take
		other forms for different reasons. For example,
		in \cite{kar2011converegence}, a certain gain matrix $K$
					is  multiplied to the left of the stochastic gradient
					vector $\widehat{\nabla_w J_{k}}(\cdot)$ 
					to make the estimator asymptotically efficient
					for a linear observation model.
		\hfill\QED
	\end{example}
	\vspace{0.5em}

Returning to the general distributed strategy
\eqref{Equ:ProblemFormulation:DisAlg_Comb1}--\eqref{Equ:ProblemFormulation:DisAlg_Comb2},
we note that it can be specialized into various useful algorithms.
We let $A_1$, $A_0$ and $A_2$ denote the $N \times N$ matrices that
collect the coefficients $\{a_{1,lk}\}$, $\{a_{0,lk}\}$ and $\{a_{2,lk}\}$.
Then, condition \eqref{Equ:ProblemFormulation:a_lk:convex} is equivalent
to
	\begin{align}
		\label{Equ:ProblemFormulation:a_lk:convex_matrixform}
		A_1^T \mathds{1} = \mathds{1},	\quad
		A_0^T \mathds{1} = \mathds{1},	\quad
		A_2^T \mathds{1} = \mathds{1}
	\end{align}
where $\one$ is the $N\times 1$ vector with all its entries equal to one.
Condition \eqref{Equ:ProblemFormulation:a_lk:convex_matrixform} 
means that the matrices $\{A_0,A_1,A_2\}$ are left-stochastic 
(i.e., the entries on each of their columns add up to one).  
Different choices for $A_1$, $A_0$ and $A_2$ correspond to different
distributed strategies, as summarized in Table \ref{Tab:ChoiceOfMatrixA}.
Specifically, the traditional 
consensus\cite{nedic2009distributed,nedic2010cooperative,tsitsiklis1986distributed,
kar2011converegence,kar2008distributed,kar2013distributed,dimakis2010gossip}
 and diffusion (ATC and CTA) \cite{lopes2008diffusion,sayed2012diffbookchapter,
Cattivelli10,zhao2012performance,chen2011TSPdiffopt,chen2013JSTSPpareto,sayed2014adaptation,sayed2014proc}
algorithms with \emph{constant} step-sizes are given by the following 
iterations:
	\begin{align}
		\label{Equ:ProblemFormulation:Consensus}
		\mathrm{Consensus:}\;	
		&\begin{cases}
			\displaystyle
			\bm{\phi}_{k,i-1} = \sum_{l \in \mc{N}_k}\!\! a_{0,lk} \bm{w}_{l,i-1}\\
			\displaystyle
			\bm{w}_{k,i}		=	\bm{\phi}_{k,i-1}
								-
								\mu_k \hat{\bm{s}}_{k,i}(\bm{w}_{k,i-\!1})
		\end{cases}
								\\\nonumber\\
		\label{Equ:ProblemFormulation:CTA}
		\mathrm{CTA~diffusion:}\;
		&\begin{cases}
			\displaystyle
			\bm{\phi}_{k,i-1}	=	\sum_{l \in \mc{N}_k} a_{1,lk} \bm{w}_{l,i-1}		\\
			\bm{w}_{k,i}			=	\bm{\phi}_{k,i-1}
									-
									\mu_k \hat{\bm{s}}_{k,i}(\bm{\phi}_{k,i-1})
		\end{cases}
		\\\nonumber\\		
		\label{Equ:ProblemFormulation:ATC}
		\mathrm{ATC~diffusion:}\;
		&\begin{cases}
			\displaystyle
			\bm{\psi}_{k,i}	=	\bm{w}_{k,i-1}
								\!-\!
								\mu_k \hat{\bm{s}}_{k,i}(\bm{w}_{k,i-1})	\\
			\displaystyle
			\bm{w}_{k,i}		=	\sum_{l \in \mc{N}_k} a_{2,lk} \bm{\psi}_{l,i}
		\end{cases}
	\end{align}	
Therefore, the convex combination steps appear in different locations in the
consensus and diffusion implementations. For instance, observe that the consensus strategy 
\eqref{Equ:ProblemFormulation:Consensus} evaluates the update direction
$\hat{\bs}_{k,i}(\cdot)$ at $\bw_{k,i-1}$, which is the estimator \emph{prior} to the aggregation, 
while the diffusion strategy \eqref{Equ:ProblemFormulation:CTA}
evaluates the update direction at $\bm{\phi}_{k,i-1}$, 
which is the estimator \emph{after} the aggregation. In our analysis, 
we will proceed with the general form
\eqref{Equ:ProblemFormulation:DisAlg_Comb1}--\eqref{Equ:ProblemFormulation:DisAlg_Comb2}
to study all three schemes, and other possibilities, within a unifying framework.

\begin{table}
	\centering
	\caption{Different choices for $A_1$, $A_0$ and $A_2$ correspond to different
	distributed strategies.}
	\label{Tab:ChoiceOfMatrixA}
	\begin{tabular}{c|ccc|c}
		\hline\hline
		Distributed Strategeis	&	$A_1$	&	$A_0$	&	$A_2$	&	$A_1A_0A_2$	\\
		\hline
		Consensus				&	$I$		&	$A$		&	$I$		&	$A$			\\
		ATC diffusion			&	$I$		&	$I$		&	$A$		&	$A$			\\
		CTA	diffusion			&	$A$		&	$I$		&	$I$		&	$A$			\\
		\hline
	\end{tabular}
	\end{table}

We observe that there are two types of learning processes involved
in the dynamics of each agent $k$: (i) self-learning in \eqref{Equ:ProblemFormulation:DisAlg_Adapt} 
from locally
sensed data and (ii) social learning in
\eqref{Equ:ProblemFormulation:DisAlg_Comb1} and \eqref{Equ:ProblemFormulation:DisAlg_Comb2} 
from neighbors.
All nodes implement the same
self- and social learning structure. As a result, the learning dynamics
of all nodes in the network are coupled; knowledge exploited 
from local data at node $k$ will be propagated to its neighbors 
and from there to their neighbors in a diffusive learning process. 
It is expected that some global performance pattern will emerge 
from these localized interactions in the multi-agent system. 
In this work and the accompanying Part II \cite{chen2013learningPart2}, we address the following questions: 
	\begin{itemize}
		\item
			\underline{Limit point}: where does each state $\bm{w}_{k,i}$ converge to?
		\item
			\underline{Stability}: under which condition does convergence occur?
		\item	
			\underline{Learning rate}: how fast does convergence occur?
		\item
			\underline{Performance}: how close is $\bm{w}_{k,i}$ to the limit point?
		\item
			\underline{Generalization}: can $\w_{k,i}$ match the performance of a centralized solution?
	\end{itemize}
We address the first three questions in this part, and examine the last two questions pertaining to performance in Part II \cite{chen2013learningPart2}.  We  address the five questions by characterizing analytically the learning dynamics of the
network to reveal the global behavior that emerges in the small step-size regime.
The answers to these questions will provide useful and novel insights about how 
to tune the algorithm parameters in order to reach desired performance levels
--- see Sec. \ref{P2-Sec:Benefits} in Part II \cite{chen2013learningPart2}.

\subsection{Relation to Prior Work}
\label{Sec:ProblemFormulation:PriorWork}
In comparison with the existing literature\cite{tsitsiklis1986distributed,
kar2011converegence,kar2008distributed,kar2013distributed,dimakis2010gossip,ram2010distributed,
srivastava2011distributed,nedic2009distributed,lee2013distributed,bianchi2012performance,johansson2008subgradient,braca2008running,stankovic2011decentralized}, it is worth noting that most prior studies on distributed optimization algorithms focus on studying their performance and convergence  under {\em diminishing} step-size conditions and for {\em doubly-stochastic} combination policies (i.e., matrices for which the entries on each of their columns {\em and} on each of their rows add up to one). These are of course useful conditions, especially when the emphasis is on solving {\em static} optimization problems. We focus instead on the case of \emph{constant} step-sizes because, as explained earlier, they enable continuous adaptation and lea{rning under drifting conditions; in contrast, diminishing step-sizes turn off learning once they approach zero. By using constant step-sizes, the resulting algorithms are able to track {\em dynamic} solutions that may slowly drift as the underlying problem conditions change. 
{

Moreover, constant step-size implementations have merits even for  stationary environments where the solutions remain {\em static}. This is because, as we are going to show later in this work and its accompanying Part II \cite{chen2013learningPart2},  constant step-size learning converges at a geometric rate, in the order of $O(\gamma^i)$ for some $0<\gamma<1$, towards a small mean-square error in the order of the step-size parameter. This means that these solutions can attain satisfactory performance even after short intervals of time. In comparison, implementations that rely on a diminishing step-size of the form $\mu(i)=\mu_o/i$, for some constant $\mu_o$, converge almost surely to the solution albeit at the slower rate of $O(1/i)$.
%
%
%
Furthermore, the choice of the parameter $\mu_o$ is critical to guarantee the $O(1/i)$ rate \cite[p.54]{poliak1987introduction}; if $\mu_o$ is not large enough, the resulting convergence rate can be considerably slower than $O(1/i)$. To avoid this slowdown in convergence, a large initial value $\mu_o$ is usually chosen in practice, which ends up leading to an overshoot in the learning curve; the curve grows up initially before starting its decay at the asymptotic rate, $O(1/i)$.
}

We remark that we also do not limit the choice of combination policies to being doubly-stochastic; we only require condition 
\eqref{Equ:ProblemFormulation:a_lk:convex_matrixform}. It turns out that left-stochastic matrices lead to superior mean-square error performance (see, e.g., expression \eqref{P2-Equ:Benefits:HastingRule} in Part II \cite{chen2013learningPart2} and also \cite{Cattivelli10}). The use of both constant step-sizes and left-stochastic combination policies enrich the learning dynamics of the network in interesting ways, as we are going to discover. In particular, under 
 these conditions, we will derive an interesting result that reveals how the topology of the network determines the limit point of the distributed strategies. We will show that the combination weights steer the convergence point away from the expected solution and towards any of many possible Pareto optimal solutions. This is in contrast to commonly-used doubly-stochastic combination policies where the limit point of the network is fixed and cannot be changed regardless of the topology. We will show that the limit point is determined by the right eigenvector that is associated with the eigenvalue at one for the matrix product $A_1A_0A_2$. We will also be able to characterize in Part II \cite{chen2013learningPart2} how close each agent in the network gets to this limit point and to explain how the limit point plays the role of a Pareto optimal solution for a suitably defined aggregate cost function. 

{
We note that the concept of a limit point in this work is different from earlier studies on the limit point of consensus implementations that deal exclusively with the problem of evaluating the weighted average of initial state values at the agents (e.g.,  \cite{tahbaz2008necessary}). In these implementations, there are no adaptation steps and no streaming data; the step-size parameters $\{\mu_k\}$ are set to zero in \eqref{Equ:ProblemFormulation:DisAlg_Adapt}, \eqref{Equ:ProblemFormulation:Consensus}, \eqref{Equ:ProblemFormulation:CTA} and \eqref{Equ:ProblemFormulation:ATC}. In contrast, the general distributed strategy \eqref{Equ:ProblemFormulation:DisAlg_Comb1}--\eqref{Equ:ProblemFormulation:DisAlg_Comb2} is meant to solve continuous adaptation and learning problems from streaming data arriving at the agents. In this case, the adaptation term $\hat{\bm{s}}_{k,i}(\cdot)$ (self-learning) is necessary, in addition to the combination step (social-learning). There is a non-trivial coupling between both steps and across the agents. For this reason, identifying the actual limit point of the distributed strategy is rather challenging and requires a close examination of the evolution of the network dynamics, as demonstrated by the technical tools used in this work. In comparison, while the evolution of traditional average-consensus implementations can be described by linear first-order recursions,  the same is not true for adaptive networks where the dynamics evolves according to nonlinear stochastic difference recursions. 
}

\section{MODELING ASSUMPTIONS}
\label{Sec:Assump}
In this section, we collect the assumptions and definitions
that are used in the analysis and explain why they are justified and how they relate to similar assumptions used in several prior studies in the literature. As the discussion will reveal, in most cases, the assumptions that we adopt here are relaxed (i.e., weaker) versions than conditions used before in the literature
such as in
\cite{tsitsiklis1986distributed,kar2011converegence,kar2008distributed,
ram2010distributed,srivastava2011distributed,nedic2009distributed,
chen2011TSPdiffopt,chen2013JSTSPpareto,zhao2012performance,
bertsekas2000gradient,poliak1987introduction,
eksin2012tspdistributed,
lee2013distributed}. We do so in order to analyze the learning behavior of networks under conditions that are similar to what is normally assumed in the prior art, albeit ones that are generally less restrictive. 
\vspace{0.5em}
	\begin{assumption}[Strongly-connected network]
		\label{Assumption:Network}
		The $N \times N$ matrix product $A \defeq A_1A_0A_2$ is assumed to be a primitive 
		left-stochastic matrix, i.e., $A^T\mathds{1}=\mathds{1}$ 
		and there exists a finite integer $j_o$ such that 
		all entries of $A^{j_o}$ are strictly positive. 
		\hfill\QED
	\end{assumption}
\vspace{0.5em}	

This condition is satisfied for most networks and is not restrictive. 
Let $A=[a_{lk}]$ denote the entries of $A$. Assumption \ref{Assumption:Network} 
is automatically satisfied if the product $A$ corresponds to a connected network and there exists at least one 
$a_{kk}>0$ for some node $k$ (i.e., at least one node with a nontrivial self-loop)
\cite{sayed2012diffbookchapter,sayed2014adaptation}. It then follows from the Perron-Frobenius Theorem
\cite{horn1990matrix} that  
the matrix $A_1A_0A_2$ has a single eigenvalue at one of 
multiplicity one and all other eigenvalues are strictly less than one in magnitude,
i.e.,
	\begin{align}
		1=\lambda_1(A) > |\lambda_2(A)| \ge \cdots \ge |\lambda_N(A)|
		\label{Equ:Modeling:lambda_A_order}
	\end{align}
Obviously, $\mathds{1}^T$ is a left eigenvector for $A_1A_0A_2$ corresponding to the eigenvalue at one. Let $\theta$ denote the 
right eigenvector corresponding to the eigenvalue at one (the Perron vector) and whose entries are normalized to add up to one, i.e., 
	\begin{align}
		A \theta = \theta, \qquad \one^T \theta = 1
	\end{align}
Then, the Perron-Frobenius Theorem further ensures that all entries of $\theta$ satisfy $0 < \theta_k <1$. Note that, unlike 
\cite{kar2013distributed,nedic2010cooperative,tsitsiklis1986distributed,kar2011converegence,kar2008distributed,dimakis2010gossip,ram2010distributed,srivastava2011distributed,nedic2009distributed}, we do not require the matrix  $A_1A_0A_2$ to be doubly-stochastic (in which case $\theta$ would be $\mathds{1}/N$ and, therefore, all its entries will be identical to each other). Instead, we will study the performance of the algorithms in the context of 
general left-stochastic matrices $\{A_1, A_0, A_2\}$ and we will examine the influence of (the generally non-equal entries of) $\theta$ 
on both the limit point and performance of the network.
 	\vspace{0.5em}
	\begin{definition}[Step-sizes]
		\label{Def:StepSize}
		Without loss of generality, we express the step-size at each node $k$ as
		$\mu_k=\mu_{\max}\beta_k$, where
		$\mu_{\max}\defeq \max \{\mu_k\}$ is the largest step-size, and $0 \le \beta_k \le 1$.
		We assume $\beta_k>0$ for at least one $k$. Thus, observe 
		that we are allowing the possibility of zero step-sizes by some of the agents. 
		
		\hfill\QED
	\end{definition}
	\vspace{0.5em}
	\begin{definition}[Useful vectors]
		\label{Def:UsefulVectors}
		Let $\pi$ and $p$ be the following $N \times 1$ vectors:
			\begin{align}
				\label{Def:UsefulVectors:pi}
				\pi		&\defeq		A_2 \theta				\\
				\label{Def:UsefulVectors:p}
				p		&\defeq		\col\{\pi_1\beta_1,\ldots,\pi_N\beta_N\}
			\end{align}
		where $\pi_k$ is the $k$th entry of the vector $\pi$.
		\hfill\QED
	\end{definition}
	\vspace{0.5em}

The vector $p$ will play a critical role in the 
performance of the distributed strategy
\eqref{Equ:ProblemFormulation:DisAlg_Comb1}--\eqref{Equ:ProblemFormulation:DisAlg_Comb2}.
Furthermore, we introduce the following assumptions on 
the update vectors $\hat{\bm{s}}_{k,i}(\cdot)$ in \eqref{Equ:ProblemFormulation:DisAlg_Comb1}--%
\eqref{Equ:ProblemFormulation:DisAlg_Comb2}.
	\vspace{0.5em}
	\begin{assumption}[Update vector: Randomness]
		\label{Assumption:UpdateVectorRandomness}
		There exists an $M \times 1$ deterministic vector function 
		$s_k(w)$ such that, for all $M \times 1$ vectors $\bm{w}$ in the filtration $\mc{F}_{i-1}$
		generated by  the past history of iterates $\{\bm{w}_{k,j}\}$ for $j \le i-1$
		and all $k$, it holds that
			\begin{align}
				\label{Equ:Assumption:Randomness:MDS}
				&\E\left\{
					\hat{\bm{s}}_{k,i}(\bm{w}) | \mc{F}_{i-1}
				\right\}
						=	s_k(\bm{w})						
			\end{align}
		for all $i,k$. Furthermore, there exist
		$\alpha \ge 0$ and $\sigma_{v}^2 \ge 0$ such that
		for all $i,k$ and $\bm{w} \in \mc{F}_{i-1}$:
			\begin{align}
				\label{Equ:Assumption:Randomness:RelAbsNoise}
				&\E\left\{\!
					\left\| 
						\hat{\bm{s}}_{k,i}(\bm{w})\!-\!s_{k}(\bm{w})
					\right\|^2
					\big|
					\mF_{i-1}
				\right\}
						\le \alpha
							\!\cdot\!
							\| \w \|^2 \!+\! \sigma_{v}^2							
			\end{align}
		\hfill\QED		
	\end{assumption}
	\vspace{0.5em}

Condition \eqref{Equ:Assumption:Randomness:RelAbsNoise} requires the conditional variance of the random update
direction $\hat{\bm{s}}_{k,i}(\bm{w})$ to be bounded by the square-norm
of $\bm{w}$. Condition \eqref{Equ:Assumption:Randomness:RelAbsNoise}
is a generalized version
of Assumption 2 from \cite{chen2011TSPdiffopt,chen2013JSTSPpareto}; it is also a generalization of the assumptions from \cite{ram2010distributed,bertsekas2000gradient,poliak1987introduction},
where $\hat{\bm{s}}_{k,i}(\bm{w})$ was instead modeled as the following 
perturbed version of the true gradient vector:
	\begin{align}
		\hat{\bs}_{k,i}(\bw)
		=\widehat{\nabla_w J_{k}}(\bm{w})	=	\nabla_w J_k (\bm{w}) + \bm{v}_{k,i}(\bm{w})
		\label{Equ:Modeling:StochasticGradient}
	\end{align}
with $s_k(\bm{w})=\nabla_w J_k(\bm{w})$, in which case conditions 
\eqref{Equ:Assumption:Randomness:MDS}--\eqref{Equ:Assumption:Randomness:RelAbsNoise}
translate into the following requirements on the 
gradient noise $\bm{v}_{k,i}(\bm{w})$:
	\begin{align}
				\label{Equ:EarlyAssumption:RandomnessGradient:MDS}
				&\E\left\{
					{\bm{v}}_{k,i}(\bm{w}) | \mc{F}_{i-1}
				\right\}
						=	0	
				\qquad \mathrm{(zero~mean)}
							\\
				\label{Equ:EarlyAssumption:RandomnessGradient:RelAbsNoise}
				&\E\left\{\!
					\left\| 
						{\bm{v}}_{k,i}(\bm{w})
					\right\|^2
					\big|
					\mF_{i-1}
				\right\}
						\le \alpha 
							\!\cdot\!
							\| \w \|^2 \!+\! \sigma_{v}^2					
			\end{align}
In Example 2 of \cite{chen2011TSPdiffopt}, we explained how these conditions are satisfied automatically in the context of mean-square-error adaptation over networks. Assumption \ref{Assumption:UpdateVectorRandomness} given by \eqref{Equ:Assumption:Randomness:MDS}--\eqref{Equ:Assumption:Randomness:RelAbsNoise}  is more general than \eqref{Equ:EarlyAssumption:RandomnessGradient:MDS}--\eqref{Equ:EarlyAssumption:RandomnessGradient:RelAbsNoise} because we are allowing the update vector $\hat{\bm{s}}_{k,i}(\cdot)$ to be constructed in forms other than \eqref{Equ:Modeling:StochasticGradient}. Furthermore, Assumption \eqref{Equ:EarlyAssumption:RandomnessGradient:RelAbsNoise} is also more relaxed than the following variant used in \cite{bertsekas2000gradient,poliak1987introduction}:
	\begin{align}
		\E\left\{\!
			\left\| 
				{\bm{v}}_{k,i}(\bm{w})
			\right\|^2
			\big|
			\mF_{i-1}
		\right\}
				\le 
					\alpha 
					\!\cdot\!
					\| \nabla_w J_k(\w) \|^2 \!+\! \sigma_{v}^2	
		\label{Equ:EarlyAssumption:RandomnessGradient:RelAbsNoise_variant}
	\end{align}
This is because \eqref{Equ:EarlyAssumption:RandomnessGradient:RelAbsNoise_variant} implies a condition of the form \eqref{Equ:EarlyAssumption:RandomnessGradient:RelAbsNoise}. Indeed, note that 
	\begin{align}
		\E&\left\{\!
			\left\| 
				{\bm{v}}_{k,i}(\bm{w})
			\right\|^2
			\big|
			\mF_{i-1}
		\right\}
						\nn\\
				&= 
						\alpha 
						\!\cdot\!
						\| \nabla_w J_k(\w) - \nabla_w J_k(0) + \nabla_w J_k(0) \|^2 \!+\! \sigma_{v}^2	
						\nn\\
				&\overset{(a)}{\le }
						2\alpha 
						\!\cdot\!
						\| \nabla_w J_k(\w) - \nabla_w J_k(0) \|^2
						+
						2\alpha
						\| \nabla_w J_k(0) \|^2 
						\!+\! \sigma_{v}^2
						\nn\\
				&\overset{(b)}{\le}
						2\alpha \lambda_U^2 
						\cdot
						\| \w \|^2
						+
						2\alpha
						\| \nabla_w J_k(0) \|^2 
						\!+\! \sigma_{v}^2
						\nn\\
				&\defeq
						\alpha' \cdot \| \w \|^2 + \sigma_{v'}^2
	\end{align}
where step (a) uses the relation $\| x+y\|^2 \le 2 \| x\|^2 + 2\| y\|^2$, and step (b) used \eqref{Equ:Assumption:Lipschitz} to be assumed next.
	
	\vspace{0.5em}
	\begin{assumption}[Update vector: Lipschitz]
		\label{Assumption:UpdateVectorLipschitz}
		There exists a nonnegative $\lambda_{U}$ such that for 
		all $x,y \in \mb{R}^M$ and all $k$:
			\begin{align}
				\label{Equ:Assumption:Lipschitz}
				\|s_k(x)-s_k(y)\|	\le 	\lambda_{U} \cdot \|x-y\|
			\end{align}
		where the subscript ``$U$'' in $\lambda_{U}$ means
		``upper bound''.
		\hfill\QED
	\end{assumption}	
	\vspace{0.5em}
A similar assumption to \eqref{Equ:Assumption:Lipschitz}  
was used before in the literature for the model 
\eqref{Equ:Modeling:StochasticGradient} by requiring the gradient vector 
of the individual cost functions $J_k(w)$ to be Lipschitz
\cite{tsitsiklis1986distributed,eksin2012tspdistributed,
bertsekas2000gradient,poliak1987introduction,lee2013distributed}.
Again, condition \eqref{Equ:Assumption:Lipschitz} is more general 
because we are not limiting the construction of the update direction to 
\eqref{Equ:Modeling:StochasticGradient}. 		
	\vspace{0.5em}
	\begin{assumption}[Update vector: Strong monotonicity]
		\label{Assumption:UpdateVectorMonot}
		Let $p_k$ denote the $k$th entry of the vector $p$ 
		defined in \eqref{Def:UsefulVectors:p}.
		There exists $\lambda_L > 0$ such that for
		all $x,y \in \mb{R}^M$:
			\begin{align}
				\label{Equ:Assumption:StrongMonotone}
				(x-y)^T \cdot
						\sum_{k=1}^N
						p_k 
						\Big[
								s_k(x)-s_k(y)
						\Big]
						\ge 	\lambda_L \cdot \|x-y\|^2
			\end{align}
		where the subscript ``$L$'' in $\lambda_L$ means ``lower bound''.
		\hfill\QED
	\end{assumption}
	\vspace{0.5em}
	\begin{remark}
		Applying the Cauchy-Schwartz inequality \cite[p.15]{horn1990matrix} 
		to the left-hand side of
		\eqref{Equ:Assumption:StrongMonotone} and using 
		\eqref{Equ:Assumption:Lipschitz}, we deduce the following
		relation between $\lambda_L$ and $\lambda_U$:
			\begin{align}
				\label{Equ:Remark:lambdaU_lambda_L_relation}
				\lambda_U \cdot \|p\|_1  \ge \lambda_L
			\end{align}
		where $\|\cdot\|_1$ denotes the $1-$norm of the vector argument.
		\hfill\QED
	\end{remark}

The following lemma gives the equivalent forms of Assumptions \ref{Assumption:UpdateVectorLipschitz}--\ref{Assumption:UpdateVectorMonot} when the $\{s_k(w)\}$ happen to be differentiable.
	\begin{lemma}[Equivalent conditions on update vectors]
		\label{Lemma:EquivCond_updateVec}
		Suppose $\{s_k(w)\}$ are differentiable in an open set $\mS \subseteq \mb{R}^M$. Then,
		having conditions \eqref{Equ:Assumption:Lipschitz} and \eqref{Equ:Assumption:StrongMonotone} hold on $\mS$ is 
		equivalent to the following conditions, respectively,
			\begin{align}
				\| \nabla_{w^T} s_k(w) \|		&\le 		\lambda_U				
				\label{Equ:Lemma:EquivCondUpdateVec:LipschitzUpdate_HessianUB}
															\\
				\frac{1}{2} [ H_c(w) + H_c^T(w) ] 
												&\ge 		\lambda_L \cdot I_M
				\label{Equ:Lemma:EquivCondUpdateVec:StrongMono_HessianLB}
			\end{align}
		for any $w \in \mS$, where $\|\cdot\|$ denotes the $2$-induced norm (largest singular value) of 
		its matrix argument and 
			\begin{align}
				H_c(w)		\defeq		\sum_{k=1}^n p_k \nabla_{w^T} s_k(w)
				\label{Equ:Lemma:EquivCondUpdateVec:Hc_def}
			\end{align}
	\end{lemma}
	\begin{proof}
		See Appendix \ref{Appendix:Proof_Lemma_EquivCondUpdateVec}.
	\end{proof}

Since in Assumptions \ref{Assumption:UpdateVectorLipschitz}--\ref{Assumption:UpdateVectorMonot} we require conditions \eqref{Equ:Assumption:Lipschitz} and \eqref{Equ:Assumption:StrongMonotone} to hold over the entire $\mb{R}^M$, 
then the equivalent conditions \eqref{Equ:Lemma:EquivCondUpdateVec:LipschitzUpdate_HessianUB}--\eqref{Equ:Lemma:EquivCondUpdateVec:StrongMono_HessianLB} will need to hold over the entire $\mb{R}^M$ when the $\{s_{k}(w)\}$ are differentiable. In the context of distributed optimization problems of the form \eqref{Equ:ProblemFormulation:Ex:J_glob}--%
\eqref{Equ:ProblemFormulation:Multiobjective} with twice-differentiable
$J_k(w)$, where the stochastic gradient vectors are constructed as in
\eqref{Equ:Modeling:StochasticGradient}, Lemma \ref{Lemma:EquivCond_updateVec} implies that the above Assumptions \ref{Assumption:UpdateVectorLipschitz}--\ref{Assumption:UpdateVectorMonot} are equivalent to the following conditions on the Hessian matrix of each $J_k(w)$\cite[p.10]{poliak1987introduction}:
	\begin{align}
		\label{Equ:ProblemFormulation:Hessian_Upperbound}
		&\left\|\nabla_w^2 J_k(w)\right\|	\le 	\lambda_U		\\
		\label{Equ:ProblemFormulation:Hessian_Lowerbound}
		&\sum_{k=1}^N p_k \nabla_w^2 J_k(w)
							\ge 	\lambda_L I_M	>	0
	\end{align}
Condition \eqref{Equ:ProblemFormulation:Hessian_Lowerbound} is in turn equivalent to requiring the following weighted sum of the individual cost functions $\{J_k(w)\}$ to be strongly convex:
	\begin{align}
		\label{Equ:ProblemFormulation:J_glob_weigthed}
		J^{\mathrm{glob}, \star}(w)	\defeq	\sum_{k=1}^N p_k J_k(w)
	\end{align}
We note that strong convexity conditions are prevalent in many studies on optimization techniques in the literature. For example, each of the individual costs $J_k(w)$ is assumed to be stronlgy convex in \cite{srivastava2011distributed} in order
to derive upper bounds on the limit superior (``$\limsup$'') of the mean-square-error of the estimates $\bm{w}_{k,i}$ or the expected value of the cost function at $\bm{w}_{k,i}$. In comparison, the framework in this work does not require the individual costs to be strongly convex or even convex. Actually, some of the costs $\{J_k(w)\}$  can be non-convex as long as the aggregate cost \eqref{Equ:ProblemFormulation:J_glob_weigthed} remains strongly convex. 
Such relaxed assumptions on the individual costs introduce challenges into the analysis, and we need to develop a systematic approach to characterize the limiting behavior of adaptive networks under such less restrictive conditions. 
{
	\begin{example}
		\label{Ex:LMS_GlobalObservability}
		
		The strong-convexity condition \eqref{Equ:ProblemFormulation:Hessian_Lowerbound} on the
		aggregate cost \eqref{Equ:ProblemFormulation:J_glob_weigthed} can be related to a global
		observability condition similar to \cite{kar2011converegence,kar2008distributed,kar2013distributed}. 
		To illustrate this point, we consider an example dealing with quadratic costs. Thus, consider 
		a network of $N$ agents that are connected according to a certain topology. 
		The data samples received at each agent $k$ at 
		time $i$ consist of the observation signal $\bm{d}_{k}(i) \in \mathbb{R}$ and the regressor vector 
		$\bm{u}_{k,i} \in \mathbb{R}^{1 \times M}$, which are assumed to be related 
		according to the following linear model:
			\begin{align}
				\bm{d}_{k}(i)	=	\bm{u}_{k,i} w^o + \bm{v}_{k}(i)
				\label{Equ:Ex:LinearModel_LMS}
			\end{align}
		where $\bm{v}_k(i) \in \mathbb{R}$ is a zero-mean additive white noise that is uncorrelated with the
		regressor 
		vector $\bm{u}_{\ell,j}$ for all $k, \ell,i,j$. Each agent in the network would like to estimate 
		$w^o \in \mathbb{R}^M$ by 
		learning from the local data stream $\{ \bm{d}_{k}(i), \bm{u}_{k,i}\}$ and by collaborating with its 
		intermediate neighbors. The problem can be formulated as minimizing the aggregate cost 
		\eqref{Equ:ProblemFormulation:Ex:J_glob} with $J_k(w)$ chosen to be
			\begin{align}
				J_k(w)	=	\frac{1}{2}\Expt | \bm{d}_k(i) - \bm{u}_{k,i}w |^2
			\end{align}
		i.e.,
			\begin{align}
				J^{\mathrm{glob}}(w)		=	\sum_{k=1}^N \frac{1}{2} 
										\Expt | \bm{d}_k(i) - \bm{u}_{k,i}w |^2
				\label{Equ:Ex:J_glob_LMS_def1}
			\end{align}
		This is a distributed least-mean-squares (LMS) estimation problem studied in
		\cite{Cattivelli10,lopes2008diffusion,zhao2012performance}. 
		We would like to explain that 
		condition \eqref{Equ:ProblemFormulation:Hessian_Lowerbound}
		amounts to a global observability condition. First, note that the Hessian matrix of $J_k(w)$ 
		in this case is the covariance matrix of the regressor $\u_{k,i}$:
			\begin{align}
				R_{u,k} \defeq \Expt\{ \u_{k,i}^T \u_{k,i} \}
			\end{align}
		Therefore, condition \eqref{Equ:ProblemFormulation:Hessian_Lowerbound} becomes that there
		exists a $\lambda_L>0$ such that
			\begin{align}
				\sum_{k=1}^N p_k R_{u,k}	\ge 	\lambda_L I_M 	>	0
				\label{Equ:Example:GlobalObservabilityLMSweighted}
			\end{align}
		Furthermore, it can be verified that the above inequality holds
		for any positive $\{p_k\}$ as long as the following \emph{global observability} condition
		holds:
			\begin{align}
				\sum_{k=1}^N R_{u,k}	>	0
				\label{Equ:Example:GlobalObservabilityLMS}
			\end{align}
		To see this, let $p_{\min} = \min_k p_k$, and write the left-hand side of 
		\eqref{Equ:Example:GlobalObservabilityLMSweighted} as
			\begin{align}
				\sum_{k=1}^N p_k R_{u,k}
					&=		p_{\min} \sum_{k=1}^N R_{u,k}
							+
							\sum_{k=1}^N (p_k - p_{\min}) R_{u,k}
							\nn\\
					&\ge		p_{\min} \sum_{k=1}^N R_{u,k}
					>		\underbrace{
								p_{\min} 
								\lambda_{u,\min}
							}_{ \defeq \lambda_L}
							\cdot
							I_M
			\end{align}
		where $\lambda_{u,\min}$ denotes the minimum eigenvalue of $\sum_{k=1}^N R_{u,k}$.
		Note that the left-hand side of \eqref{Equ:Example:GlobalObservabilityLMS} is the Hessian
		of $J^{\mathrm{glob}}(w)$ in \eqref{Equ:Ex:J_glob_LMS_def1}. Therefore, condition 
		\eqref{Equ:Example:GlobalObservabilityLMS} means that the aggregate
		cost function \eqref{Equ:Ex:J_glob_LMS_def1} is strongly convex so that the information
		provided by the linear observation model \eqref{Equ:Ex:LinearModel_LMS} over the 
		entire network is sufficient to uniquely identify the minimizer of 
		\eqref{Equ:Ex:J_glob_LMS_def1}. Similar global observability conditions were
		used in \cite{kar2011converegence,kar2008distributed,kar2013distributed} to study the performance
		of distributed parameter estimation problems. Such conditions are useful because it implies
		that even if $w^o$ is not locally observable to any agent in the network  but is globally observable,
		i.e., $R_{u,k}>0$ does not hold for any $k=1,\ldots,N$ but
		\eqref{Equ:Example:GlobalObservabilityLMS} holds, the distributed strategies 
		\eqref{Equ:ProblemFormulation:Consensus}--\eqref{Equ:ProblemFormulation:ATC} will still
		enable each agent to estimate the correct $w^o$ through local cooperation. In Part II
		\cite{chen2013learningPart2}, we provide more insights into how cooperation
		benefits the learning at each agent.
		\hfill\QED
	\end{example}
}

	\begin{assumption}[Jacobian matrix: Lipschitz]
		\label{Assumption:JacobianUpdatVectorLipschitz}
		Let $w^o$ denote the limit point of the distributed strategy
		\eqref{Equ:ProblemFormulation:DisAlg_Comb1}--\eqref{Equ:ProblemFormulation:DisAlg_Comb2},
		which is defined further ahead as 
		the unique solution to \eqref{Equ:LearnBehav:FixedPointEqu}. Then, in a small 
		neighborhood around $w^o$, we assume that $s_k(w)$ is differentiable with 
		respect to $w$ and satisfies
			\begin{align}
				\label{Equ:Assumption:LipschitzJacobian}
				\|\nabla_{w^T} s_k(w^o+\delta w)-\nabla_{w^T} s_k(w^o)\|\leq\;\lambda_H\cdot \|\delta w\|	
			\end{align}
		for all $\|\delta w\|\leq r_H$ for some small $r_H$, and where $\lambda_H$ is a 
		nonnegative number independent of $\delta w$. 
		
		\hfill\QED
	\end{assumption}
	\vspace{0.5em}
In the context of distributed optimization problems of the form \eqref{Equ:ProblemFormulation:Ex:J_glob}--%
\eqref{Equ:ProblemFormulation:Multiobjective} with twice-differentiable
$J_k(w)$, where the stochastic gradient vectors are constructed as in
\eqref{Equ:Modeling:StochasticGradient},
the above Assumption translates into the following Lipschitz Hessian condition:
	\begin{align}
		\label{Equ:Assumption:LipschitzHessian}
		\|\nabla_{w}^2 J_k(w^o + \delta w) - \nabla_{w}^2 J_k(w^o)\| 
				\le \lambda_{H} \cdot \| \delta w \|
	\end{align}
Condition \eqref{Equ:Assumption:LipschitzJacobian} is useful when we examine the convergence rate of the algorithm later in this article. It is also useful in deriving the steady-state mean-square-error expression \eqref{Equ:LearnBehav:MSE} in Part II \cite{chen2013learningPart2}.

\section{Learning Behavior}
\label{Sec:LearningBehavior}

\subsection{Overview of Main Results}
\label{Sec:LearningBehavior:Overview}

Before we proceed to the formal analysis, we first give
a brief overview of the main results that we are going
to establish in this part on the learning behavior of 
the distributed strategies
\eqref{Equ:ProblemFormulation:DisAlg_Comb1}--%
\eqref{Equ:ProblemFormulation:DisAlg_Comb2} for sufficiently
small step-sizes.
The first major conclusion is that for general
\emph{left-stochastic} matrices $\{A_1,A_0,A_2\}$, the agents
in the network will have their estimators $\bw_{k,i}$ converge,
in the mean-square-error sense, to the \emph{same} vector $w^o$
that corresponds to the unique solution of the following
algebraic equation:
	\begin{align}
		\sum_{k=1}^N p_k s_k(w)	=	0
		\label{Equ:LearnBehav:FixedPointEqu}
	\end{align}
For example, in the context of distributed optimization problems of the form
\eqref{Equ:ProblemFormulation:Ex:J_glob}, this result implies that
for left-stochastic matrices $\{A_1,A_0,A_2\}$,
the distributed strategies represented by \eqref{Equ:ProblemFormulation:DisAlg_Comb1}--%
\eqref{Equ:ProblemFormulation:DisAlg_Comb2} will
\emph{not} converge to the global minimizer of the original aggregate cost
\eqref{Equ:ProblemFormulation:Ex:J_glob}, which is the unique solution to the
alternative algebraic equation
	\begin{align}
		\sum_{k=1}^N \nabla_w J_k(w) = 0
		\label{Equ:LearnBehav:FixedPointEquGrad_unifWeight}
	\end{align}
Instead, these distributed
solutions will converge to the global minimizer of the \emph{weighted} aggregate
cost $J^{\mathrm{glob},\star}(w)$ defined by \eqref{Equ:ProblemFormulation:J_glob_weigthed}
in terms of the entries $p_k$, i.e., to the unique solution of 
	\begin{align}
		\sum_{k=1}^N p_k \nabla_w J_k(w) = 0		
		\label{Equ:LearnBehav:FixedPointEquGrad_pWeight}
	\end{align}
Result \eqref{Equ:LearnBehav:FixedPointEqu} also means that the distributed strategies \eqref{Equ:ProblemFormulation:DisAlg_Comb1}--\eqref{Equ:ProblemFormulation:DisAlg_Comb2} converge to a 
Pareto optimal solution of the multi-objective problem \eqref{Equ:ProblemFormulation:Multiobjective};  one Pareto solution for each selection of the topology parameters $\{p_k\}$. The distinction between the aggregate costs $J^{\mathrm{glob}}(w)$ and 
$J^{\mathrm{glob},\star}(w)$ does not appear in earlier studies on distributed optimization\cite{lee2013distributed,nedic2010cooperative,dimakis2010gossip,kar2013distributed,tsitsiklis1986distributed,
kar2011converegence,kar2008distributed,ram2010distributed, srivastava2011distributed,nedic2009distributed} mainly because these studies focus on \emph{doubly-stochastic} combination matrices, for which the entries $\{p_k\}$ will all become equal to each other  for uniform step-sizes $\mu_k\equiv \mu$ or $\mu_k(i) \equiv \mu(i)$.  In that case, the minimizations of \eqref{Equ:ProblemFormulation:Ex:J_glob} and  \eqref{Equ:ProblemFormulation:J_glob_weigthed} become equivalent and the solution of \eqref{Equ:LearnBehav:FixedPointEquGrad_unifWeight} and \eqref{Equ:LearnBehav:FixedPointEquGrad_pWeight} would then coincide. In other words,  regardless of the choice of the doubly stochastic combination weights, when the $\{p_k\}$ are identical, the limit point will be unique and correspond to the solution of 
	\begin{align}
		\sum_{k=1}^N s_k(w)=0
		\label{Equ:LearnBehav:FixedPoint_UnifWeight}
	\end{align}
In contrast, result \eqref{Equ:LearnBehav:FixedPointEqu} shows that 
left-stochastic combination policies add more flexibility into the 
behavior of the network. By selecting different combination weights, 
or even different topologies, the entries $\{p_k\}$ can be made to 
change and the limit point can be steered towards other desired Pareto optimal solutions. 
{
Even in the traditional case of consensus-type implementations for computing averages, as opposed to learning from streaming data, it also holds that it is beneficial to relax the requirement of a doubly-stochastic combination policy in order to enable broadcast algorithms without feedback \cite{Iutzeler2013analysis}.
}

The second major conclusion of the paper is that we will show in \eqref{Equ:Thm_NonAsympBound:StepSize} further ahead that there always exist sufficiently small step-sizes such that the learning process over the network is mean-square stable. This means that the weight error vectors relative to $w^o$ will satisfy
	\begin{align}
		\limsup_{ i \rightarrow \infty} \Expt \| \tilde{\w}_{k,i} \|^2 
				\le 		O(\mu_{\max})
	\end{align}
so that the steady-state mean-square-error at each agent will be of the order of $O(\mu_{\max})$.

The third major conclusion of our analysis is that we will show that,
during the convergence process towards the limit point $w^o$, the learning
curve at each agent exhibits \emph{three} distinct
phases: Transient Phase I, Transient Phase II, and Steady-State
Phase. These phases are illustrated in Fig. \ref{Fig:TypicalLearningCurve}
and they are interpreted as follows. Let us first introduce a
\emph{reference} (centralized) procedure that is described by the following centralized-type
recursion:
	\begin{align}
		\label{Equ:LearnBehav:RefRec}
		\bar{w}_{c,i}	=	\bar{w}_{c,i-1} - \mu_{\max}\sum_{k=1}^N p_k s_k(\bar{w}_{c,i-1})
	\end{align}
which is initialized at 
	\begin{align}
		\bar{w}_{c,0}	=	\sum_{k=1}^N \theta_k w_{k,0}
		\label{Equ:LearnBehav:wc0bar_initRefRecur}
	\end{align}
where $\theta_k$ is the $k$th entry of the eigenvector $\theta$, 
$\mu_{\max}$, and $\{p_k\}$ are defined in Definitions \ref{Def:StepSize}--%
\ref{Def:UsefulVectors}, $w_{k,0}$ is the initial value of the distributed
strategy at agent $k$, and $\bar{w}_{c,i}$ is an $M\times 1$
vector generated by the reference recursion \eqref{Equ:LearnBehav:RefRec}. 
The three phases
of the learning curve will be shown to have the following
features:
	\begin{figure}[t]
		\centering
		\includegraphics[width=0.45\textwidth]{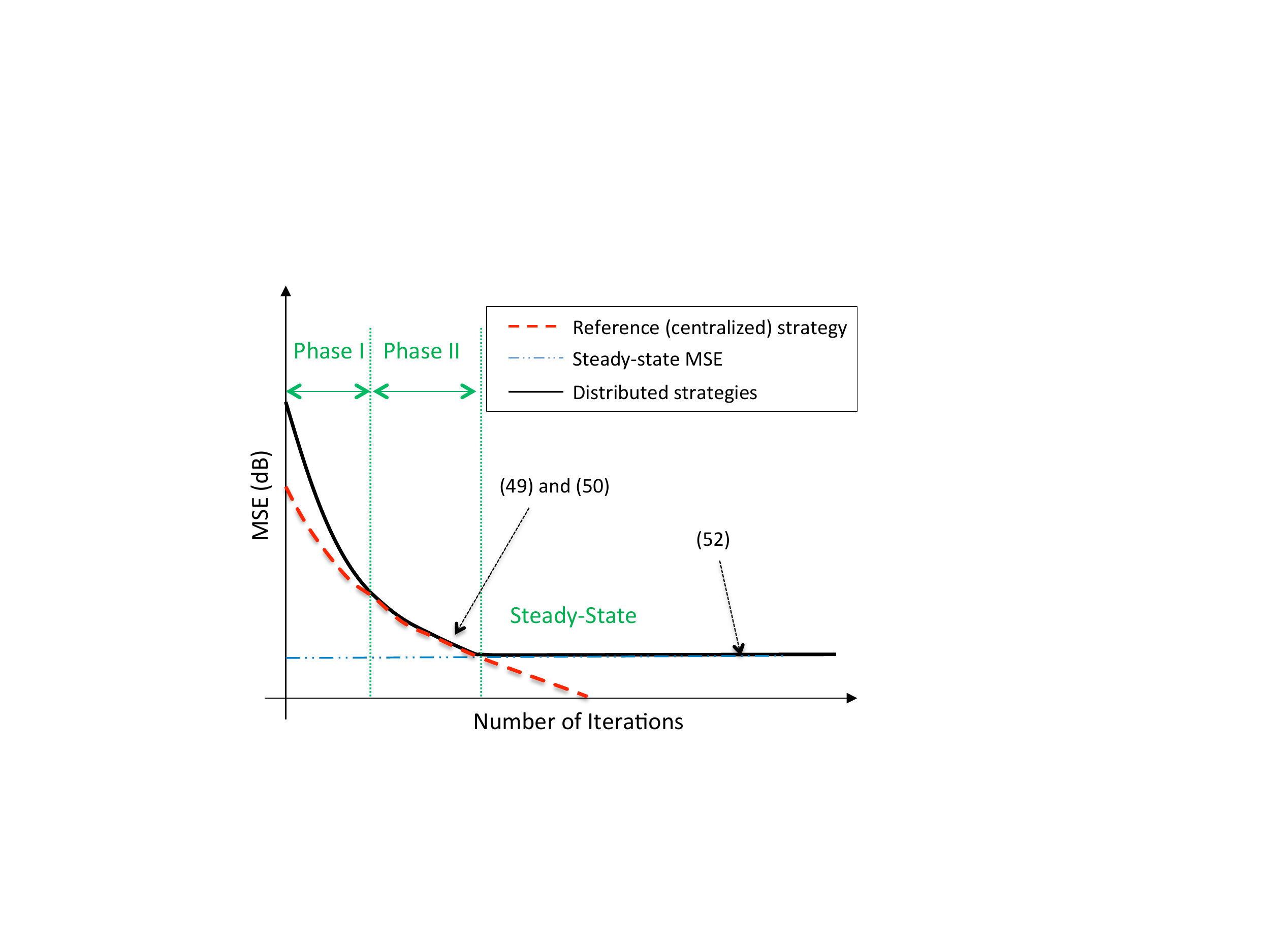}
		\caption{A typical mean-square-error (MSE) 
		learning curve includes a transient stage that consists of
		two phases and a steady-state phase. The plot shows 
		how the learning curve of a network of agents compares 
		to the learning curve of a centralized reference solution. 
		The analysis in this work, and in the accompanying Part II \cite{chen2013learningPart2} characterizes in detail the 
		parameters that determine the behavior of the network 
		(rate, stability, and performance) during each phase of the 
		learning process.}
		\label{Fig:TypicalLearningCurve}
		\vspace{-1\baselineskip}
	\end{figure}
	\begin{itemize}
		\item
		{\bf Transient Phase I:}\\
		If agents are initialized at different values, then the estimates
		of the various agents will initially evolve in such a way to make
		each $\bw_{k,i}$ get closer to the
		reference recursion $\bar{w}_{c,i}$. The rate at which
		the agents approach
		$\bar{w}_{c,i}$ will be determined by $|\lambda_2(A)|$, 
		the second largest eigenvalue of $A$ in magnitude. 
		If the agents are initialized at the
		same value, say, e.g., $\bw_{k,0}=0$, then the learning curves
		start at Transient Phase II directly.

		\item
		{\bf Transient Phase II:}\\
		In this phase, the trajectories of all agents are uniformly 
		close to the trajectory of the reference recursion; they
		converge in a coordinated manner to steady-state. The
		learning curves at this phase are well modeled by the 
		same reference recursion \eqref{Equ:LearnBehav:RefRec} 
		since we will show in \eqref{Equ:MSStability:MSE_wki_PhaseII} that:
			\begin{align}
				\E\|\tilde{\w}_{k,i}\|^2	=	\|\tilde{w}_{c,i}\|^2 
									+ 
									O(\mu_{\max}^{1/2}) \cdot \gamma_c^i 
									+ 
									O(\mu_{\max})
			\end{align}
		Furthermore, for small step-sizes and during the later stages of this phase, $\bar{w}_{c,i}$ will
		be close enough to $w^o$ and the convergence rate $r$ will be shown to satisfy:
			\begin{align}
				r		&=		\big[
							\rho(I_M - \mu_{\max} H_c)
						\big]^2
						+ 
						O\big( (\mu_{\max} \epsilon )^{\frac{1}{2(M-1)}} \big)
				\label{Equ:LearnBehav:r}
			\end{align}
		where $\rho(\cdot)$ denotes the spectral radius of its matrix argument, $\epsilon$ is
		an arbitrarily small positive number, and $H_c$ is the same matrix that results from evaluating \eqref{Equ:Lemma:EquivCondUpdateVec:Hc_def} 
		at $w=w^o$, i.e., 
			\begin{align}
				H_c	\defeq	\sum_{k=1}^N p_k H_k
					=		H_c(w^o)
				\label{Equ:LearnBehav:Hc_def}
			\end{align}
		where $H_k \defeq \nabla_{w^T} s_k(w^o)$.
		
		\item
		{\bf Steady-State Phase:}\\
		The reference recursion \eqref{Equ:LearnBehav:RefRec}
		continues converging towards $w^o$ so that $\|\tilde{w}_{c,i}\|^2
		=\|w^o-\bar{w}_{c,i}\|^2$ will converge to zero ($-\infty$ dB 
		in Fig. \ref{Fig:TypicalLearningCurve}).	
		However, for the distributed strategy \eqref{Equ:ProblemFormulation:DisAlg_Comb1}--%
		\eqref{Equ:ProblemFormulation:DisAlg_Comb2}, the mean-square-error
		$\E\|\tilde{\bw}_{k,i}\|^2=\E\|w^o-\bw_{k,i}\|^2$ at each
		agent $k$ will converge to a \emph{finite} steady-state value. 
		We will be able to characterize this value in terms of the vector $p$ in Part II \cite{chen2013learningPart2}
		as follows:\footnote{The interpretation of the limit in \eqref{Equ:LearnBehav:MSE}
		is explained in more detail in Sec. \ref{P2-Sec:LearnBehav:SteadyStateAnal} of Part II\cite{chen2013learningPart2}.}
			\begin{align}
				\lim_{i\rightarrow\infty}\!
					\E\|\tilde{\bw}_{k,i}\|^2	
						&=	
							\mu_{\max}\!\cdot\!
							\Tr\left\{
								X
								(p^T \!\! \otimes\! I_M) \mc{R}_v (p\!\otimes\! I_M)
							\right\} 
							\nn\\
							&\quad
							+ o(\mu_{\max})
				\label{Equ:LearnBehav:MSE}
			\end{align}
		where $X$ is the solution to the Lyapunov equation described
		later in
		\eqref{P2-Equ:SteadyState:ContinuousLyapunovEqu_final} of Part II \cite{chen2013learningPart2}
		(when $\Sigma=I$), and $o(\mu_{\max})$ denotes a strictly higher order term of $\mu_{\max}$.
		Expression \eqref{Equ:LearnBehav:MSE} is a 
		revealing result. It is a non-trivial extension of a classical 
		result pertaining to the mean-square-error performance 
		of stand-alone adaptive filters \cite{widrow1976stationary,jones1982analysis,
		gardner1984learning,feuer1985convergence} to the more demanding 
		context when a multitude of adaptive agents are coupled together in a 
		cooperative manner through a network topology. This result has an important ramification, 
		which we pursue in Part II \cite{chen2013learningPart2}. We will show there that no matter 
		how the agents are connected to each other, there is always a way to select the combination 
		weights such that the performance of the network is invariant to the topology. 
		This will also imply that, for any connected topology, there is always a way to select 
		the combination weights such that the performance of the network matches that of the centralized solution.
				
	\end{itemize}

{

Note that the above results are obtained for the general distributed strategy \eqref{Equ:ProblemFormulation:DisAlg_Comb1}--\eqref{Equ:ProblemFormulation:DisAlg_Comb2}. Therefore, the results can be specialized to the consensus, CTA diffusion, and ATC diffusion strategies in \eqref{Equ:ProblemFormulation:Consensus}--\eqref{Equ:ProblemFormulation:ATC} by choosing the matrices $A_1$, $A_2$, and $A_0$ according to Tab. \ref{Tab:ChoiceOfMatrixA}. The results in this paper and its accompanying Part II\cite{chen2013learningPart2} not only generalize the analysis from earlier works \cite{lopes2008diffusion,Cattivelli10,zhao2012performance,chen2011TSPdiffopt,chen2013JSTSPpareto} but, more importantly, they also provide deeper insights into the learning behavior of these adaptation and learning strategies.

}

\section{Study of Error Dynamics}

\subsection{Error Quantities}
\label{Sec:Performance:ErrorRecursion}
We shall examine the learning behavior of the distributed strategy 
\eqref{Equ:ProblemFormulation:DisAlg_Comb1}--\eqref{Equ:ProblemFormulation:DisAlg_Comb2} by examining
how the perturbation between the distributed solution 
\eqref{Equ:ProblemFormulation:DisAlg_Comb1}--\eqref{Equ:ProblemFormulation:DisAlg_Comb2} 
and the reference solution \eqref{Equ:LearnBehav:RefRec} evolves over time --- see Fig. \ref{Fig:Fig_NetworkTransformation}. Specifically, let $\check{\bw}_{k,i}$ denote the discrepancy between $\bw_{k,i}$ and $\bar{w}_{c,i}$, i.e.,
	\begin{align}
		\check{\bw}_{k,i}	\defeq	\bw_{k,i} - \bar{w}_{c,i}
		\label{Equ:DistProc:wki_check_def}
	\end{align}
and let $\bw_{i}$ and $\check{\bw}_{i}$ denote the global vectors
that collect the $\bw_{k,i}$ and $\check{\bw}_{k,i}$ from across the network, respectively:
	\begin{align}
		\bw_{i}			&\defeq		\col\{\bw_{1,i},\ldots,\bw_{N,i}\}	
		\label{Equ:DistProc:w_i_def}
									\\
		\check{\bw}_{i}	&\defeq		\col\{\check{\bw}_{1,i},\ldots,\check{\bw}_{N,i}\}
						=			\bw_{i} - \one \otimes \bar{w}_{c,i}
		\label{Equ:DistProc:w_check_i_def}
	\end{align}
It turns out that it is insightful to study the evolution of $\check{\bw}_{i}$ 
in a \emph{transformed} domain where it is possible to express the distributed
recursion \eqref{Equ:ProblemFormulation:DisAlg_Comb1}--\eqref{Equ:ProblemFormulation:DisAlg_Comb2} 
as a perturbed version of the reference recursion
\eqref{Equ:LearnBehav:RefRec}. 

\begin{definition}[Network basis transformation]
	We define the transformation by introducing the 
	Jordan canonical decomposition of the matrix $A=A_1A_0A_2$. Let 
		\begin{align}
			A^T = U D U^{-1}
			\label{Equ:distProc:A_JCF}
		\end{align}
	where $U$ is an invertible matrix whose columns correspond to the
	right-eigenvectors of $A^T$, and $D$ is a block Jordan matrix 
	with a single eigenvalue at one with multiplicity one while all other eigenvalues 
	are strictly less than one.  The Kronecker form of $A$ 
	then admits the decomposition: 
		\begin{align}
			\mc{A}^T		\defeq 	A^T \otimes I_M
						=		\mc{U} \mc{D} \mc{U}^{-1}
		\end{align}
	where 
		\begin{align}
			\mc{U}		&\defeq		U \otimes I_M,
			\qquad
			\mc{D}		\defeq		D \otimes I_M
		\end{align}
	We use $\mU$ to define the following basis transformation: 
		\begin{align}
			\label{Equ:DistProc:wi_prime_def}
			\bw_i'				&\defeq		\mU^{-1} \bw_i	
								=			(U^{-1}\otimes I_M) \bw_i
											\\
			\label{Equ:DistProc:wicheck_prime_def}
			\check{\bm{w}}_i'	&\defeq		\mU^{-1} \check{\bw}_{i}
								=			(U^{-1}\otimes I_M) \check{\bw}_i
		\end{align}
	The relations between the quantities in transformations \eqref{Equ:DistProc:wi_prime_def}--\eqref{Equ:DistProc:wicheck_prime_def} are illustrated in Fig. \ref{Fig:Fig_NetworkTransformation:Transformation}.
	\hfill \QED
\end{definition}

\begin{figure*}
	\centerline{
		\subfigure[]{
					\includegraphics[width=0.5\textwidth]{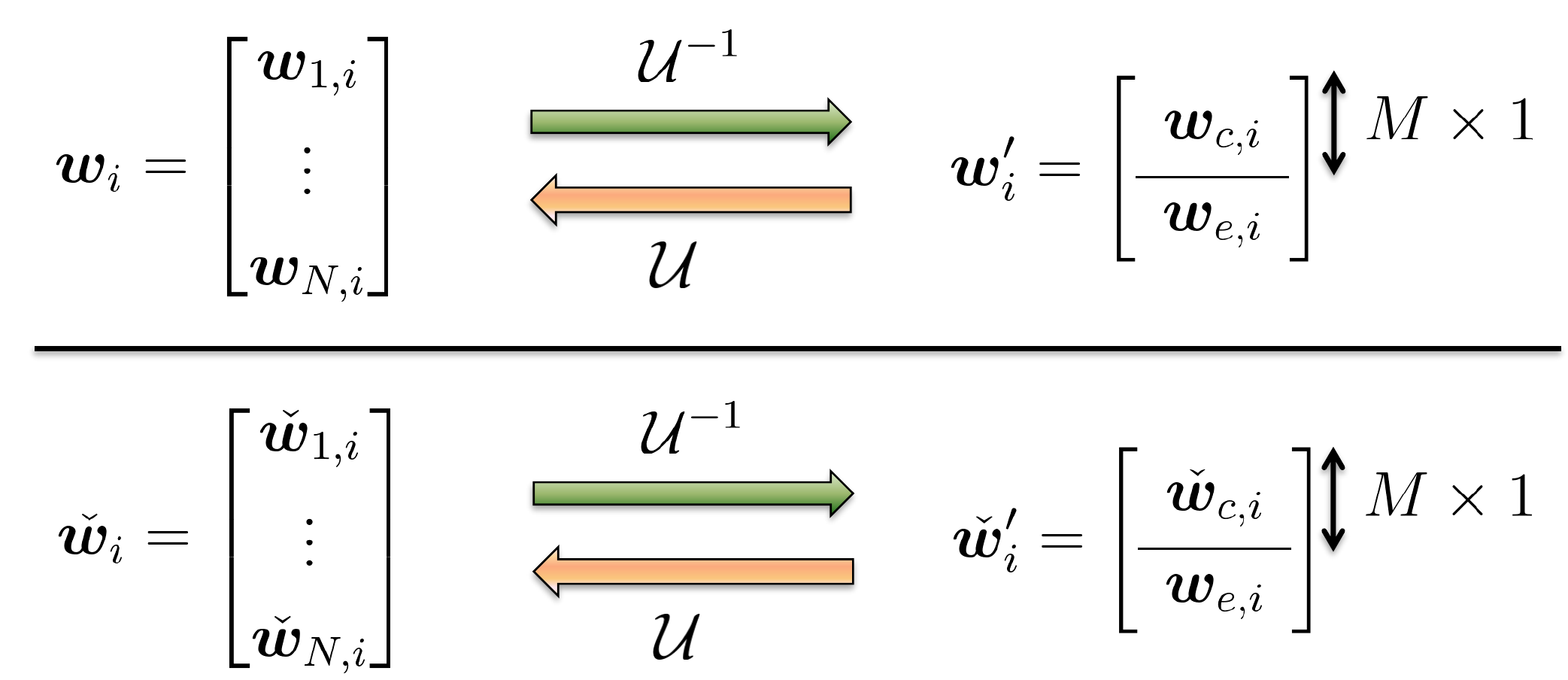}
					\label{Fig:Fig_NetworkTransformation:Transformation}
					}
	}
	\centerline{
		\subfigure[]{
					\includegraphics[width=0.6\textwidth]{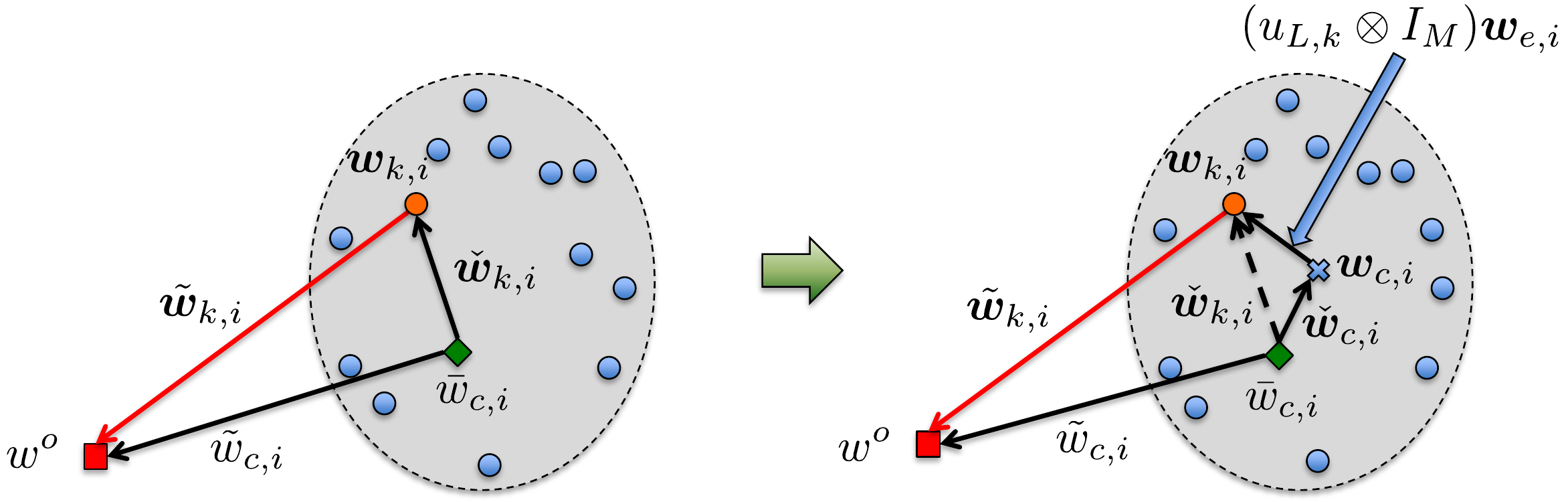}
					}
	}
	\caption{(a) Network basis transformation. (b) The diagrams show how the iterate $\w_{k,i}$ is decomposed relative to the reference $\bar{w}_{c,i}$ and relative to the centroid, $\w_{c,i}$, of the $N$ iterates across the network.
	}
	\label{Fig:Fig_NetworkTransformation}
\end{figure*}

We can gain useful insight into the nature of this transformation by exploiting more directly the structure of the matrices $\mc{U}$, $\mc{D}$, and $\mc{U}^{-1}$.	By Assumption \ref{Assumption:Network}, the matrix $A^T$ has an eigenvalue one of multiplicity one, with the corresponding left- and right-eigenvectors being $\theta^T$ and $\one$, respectively. All other eigenvalues of $D$ are strictly less than one in magnitude. Therefore, the matrices $D$, $U$, and $U^{-1}$ can be partitioned as
	\begin{align}
		\label{Equ:DistProc:D_U_Uinv}
		D	&=	
				\left[
					\begin{array}{c|c}
						1	&	\\
						\hline
							&	D_{N-1}
					\end{array}
				\right]
				\quad
		U	=	
				\left[
					\begin{array}{c|c}
						\mathds{1} 	&	U_L	
					\end{array}
				\right]
				\quad
		U^{-1}	
			=	
				\left[
					\begin{array}{c}
						\theta^T	\\
						\hline
						U_R		
					\end{array}
				\right]
	\end{align}
where $D_{N-1}$ is an $(N-1) \times (N-1)$ Jordan matrix with all
diagonal entries strictly less than one in magnitude, 
$U_L$ is an $N \times (N-1)$ matrix, and $U_R$ is an $(N-1) \times N$ matrix.
Then, the Kronecker forms $\mD$, $\mU$, and $\mU^{-1}$ can be expressed as
	\begin{align}
		\label{Equ:DistProc:mD_block}
		\mD		&=		
						\left[
							\begin{array}{c|c}
								\!\! I_M \!\!	&		\\
								\hline
									&	\!\! \mD_{N\!-\!1} \!\!
							\end{array}
						\right],
						\;
		\mU		=		
						\left[
							\begin{array}{c|c}
								\!\! \one \!\otimes\! I_M \!\!	&	\!\! \mU_L \!\!
							\end{array}
						\right],
						\;
		\mU^{-1}=		
						\left[
							\begin{array}{c}
								\!\! \theta^T	 \!\otimes\! I_M \!\! \\
								\hline
								\!\! \mU_R	\!\!	
							\end{array}
						\right]
	\end{align}
where
	\begin{align}
		\mU_L 			&\defeq 		U_L \otimes I_M
										\\
		\mU_R 			&\defeq 		U_R \otimes I_M
										\\
		\mc{D}_{N-1} 	&\defeq		D_{N-1} \otimes I_M
	\end{align}
It is important to note that $U^{-1} U=I_N$ and that
	\begin{align}
		\label{Equ:DistProc:UUinv_blockOrthonormality}
		\theta^T \one = 1, 
		\quad 
		\theta^T U_L = 0,
		\quad
		U_R \one = 0,
		\quad
		U_R U_L = I_{N-1}
	\end{align}
We first study the structure of $\bw_i'$ defined in
\eqref{Equ:DistProc:wi_prime_def} using  \eqref{Equ:DistProc:D_U_Uinv}:
	\begin{align}
		\label{Equ:DistProc:w_i_prime_w_ci_w_ei}
		\bw_{i}'
					&=	\col\{	
								\underbrace{(\theta^T \otimes I_M) \bm{w}_i}
								_{\defeq \bm{w}_{c,i}}, 
								\; 
								\underbrace{
								(U_R \otimes I_M) \bm{w}_i}_{\defeq \bm{w}_{e,i}}
						\}
	\end{align}
The two components $\bw_{c,i}$ and $\bw_{e,i}$ have useful
interpretations. Recalling that $\theta_k$ denotes the $k$th entry of the vector
$\theta$, then $\bw_{c,i}$ can be expressed as
	\begin{align}
		\bw_{c,i}	=	\sum_{k=1}^N \theta_k \bw_{k,i}
		\label{Equ:DistProc:wci_centroid}
	\end{align}
As we indicated after Assumption \ref{Assumption:Network}, 
the entries $\{\theta_k\}$ are
positive and add up to one. Therefore, $\bw_{c,i}$
is a weighted average (i.e., the centroid) of the estimates $\{\bw_{k,i}\}$ across all
agents. To interpret $\bw_{e,i}$, we examine the 
inverse mapping of \eqref{Equ:DistProc:wi_prime_def} 
from $\bm{w}_i'$ to $\bm{w}_i$ 
using the block structure of $\mU$ in \eqref{Equ:DistProc:D_U_Uinv}:
		\begin{align}
			\bm{w}_i		&=	(U\otimes I_M) \bm{w}_i'
							\nonumber\\
						&=	(\mathds{1} \otimes I_M)\bm{w}_{c,i} 
							+ (U_L \otimes I_M) \bm{w}_{e,i}
							\nonumber\\
						&=	\mathds{1} \otimes \bm{w}_{c,i} 
							+ (U_L \otimes I_M) \bm{w}_{e,i}					
			\label{Equ:DistProc:invTF_wi_wi_prime}
		\end{align}
which implies that the individual estimates at the various agents satisfy:
	\begin{align}
		\label{Equ:DistProc:w_ki_decomp}
		\bw_{k,i}	=	\bw_{c,i} + (u_{L,k}\otimes I_M) \bw_{e,i}
	\end{align}
where $u_{L,k}$ denotes the $k$th row of the matrix $U_L$. The network basis transformation defined by \eqref{Equ:DistProc:wi_prime_def}  represents the cluster of iterates $\{\w_{k,i} \}$ by its centroid $\w_{c,i}$ and their positions $\{ u_{L,k} \otimes I_M ) \w_{e,i}\}$ relative to the centroid as shown in Fig. \ref{Fig:Fig_NetworkTransformation}. The two parts, $\w_{c,i}$ and $\w_{e,i}$, of $\w_i'$ in \eqref{Equ:DistProc:w_i_prime_w_ci_w_ei} are the coordinates in this new transformed representation. Then, the actual error quantity $\tilde{\w}_{k,i}$ relative to $w^o$ can be represented as
	\begin{align}
		\tilde{\w}_{k,i}		&=	w^o -\bar{w}_{c,i} 
							-
							( \w_{k,i} - \bar{w}_{c,i} )
							\nn\\
						&=	w^o -\bar{w}_{c,i}  - ( \bw_{c,i} + (u_{L,k}\otimes I_M) \bw_{e,i} - \bar{w}_{c,i})
		\label{Equ:DistProc:wki_tilde_decomposition_interm1}
	\end{align}
Introduce
	\begin{align}
		\tilde{w}_{c,i}		&\defeq	w^o - \bar{w}_{c,i}
								\\
		\check{\w}_{c,i}		&\defeq	\w_{c,i} - \bar{w}_{c,i}
		\label{Equ:DistProc:wcicheck_def}
	\end{align}
Then, from \eqref{Equ:DistProc:wki_tilde_decomposition_interm1} we arrive at the following critical relation for our analysis in the sequel:
	\begin{align}
		\tilde{\w}_{k,i}		&=
							\tilde{w}_{c,i} - \check{\w}_{c,i} - (u_{L,k} \otimes I_M) \w_{e,i}
		\label{Equ:DistProc:wki_tilde_decomposition_final}
	\end{align}	
This relation is also illustrated in Fig. \ref{Fig:Fig_NetworkTransformation}. Then, the behavior of the error quantities $\{\tilde{\w}_{k,i}\}$ can be studied by examining $\tilde{w}_{c,i}$, $\check{\w}_{c,i}$ and $\w_{e,i}$, respectively, which is pursued in Sec. \ref{Sec:LearnBehav:Transient} further ahead. The first term is the error between the reference recursion and $w^o$, which is studied in Theorems \ref{Thm:LimitPoint}--\ref{Thm:ConvergenceRateRefRec}.  The second quantity is the difference between the weighted centroid $\bw_{c,i}$ of the cluster and the reference vector $\bar{w}_{c,i}$, and the third quantity characterizes the positions of the individual iterates $\{\bw_{k,i}\}$ relative to the centroid $\bw_{c,i}$. As long as the second and the third terms in \eqref{Equ:DistProc:wki_tilde_decomposition_final}, or equivalently, $\check{\w}_{c,i}$ and $\w_{e,i}$, are small (which will be shown in Theorem \ref{Thm:NonAsymptoticBound}), the behavior of each $\bw_{k,i}$ can be well approximated by the behavior of the reference vector $\bar{w}_{c,i}$. Indeed, $\check{\w}_{c,i}$ and $\w_{e,i}$ are the coordinates of the transformed vector $\check{\w}_i'$ defined by \eqref{Equ:DistProc:wicheck_prime_def}. To see this, we substitute  \eqref{Equ:DistProc:D_U_Uinv} and \eqref{Equ:DistProc:w_check_i_def} into
\eqref{Equ:DistProc:wicheck_prime_def} to get
	\begin{align}
		\check{\bw}_i'	&=	(U^{-1}\otimes I_M) (\bw_i - \one \otimes \bar{w}_{c,i})
							\nonumber\\
						&=	\bw_i' - (U^{-1}\one)\otimes \bar{w}_{c,i}
		\label{Equ:DistProc:wiprimecheck_block_temp1}
	\end{align}
Recalling \eqref{Equ:DistProc:UUinv_blockOrthonormality} and 
the expression for $U^{-1}$ in \eqref{Equ:DistProc:D_U_Uinv}, we obtain
	\begin{align}
		U^{-1}\one		&=	\col\{\theta^T \one, \; U_R \one\}	
							\nonumber\\
						&=	\col\{1, \; 0_{N-1}\}
		\label{Equ:DistProc:Uinv_times_one}
	\end{align}
where $0_{N-1}$ denotes an $(N-1)\times 1$ vector with all zero entries.
Substituting \eqref{Equ:DistProc:Uinv_times_one} and 
\eqref{Equ:DistProc:w_i_prime_w_ci_w_ei} 
into \eqref{Equ:DistProc:wiprimecheck_block_temp1}, we get
	\begin{align}
		\check{\bw}_i'	=	\col\{\bw_{c,i}-\bar{w}_{c,i}, 	\; \bw_{e,i}\}
					=	\col\{\check{\w}_{c,i}, 	\; \bw_{e,i}\}
		\label{Equ:DistProc:wicheck_blockstructure}
	\end{align}
Therefore, it suffices to study the dynamics of $\check{\w}_i'$ and its mean-square performance. We will establish joint recursions for $\w_{c,i}$ and $\w_{e,i}$ in Sec. \ref{Sec:StudyOfErrorDynamics:SignalRec}, and joint recursions for $\check{\w}_{c,i}$ and $\w_{e,i}$ in Sec. \ref{Sec:StudyOfErrorDynamics:ErrorRec}.  Table \ref{Tab:Quantities} summarizes the definitions of the various quantities, the recursions that they follow, and their relations.

\begin{table*}

\centering
\scriptsize

\renewcommand{\arraystretch}{1.3}

\caption{Summary of various iterates, error quantities, and their relations.}

\label{Tab:Quantities}

\begin{threeparttable}

\begin{tabular}{ |c||c|c|c||c|c|c||c|c|}

\hline

 & \multicolumn{3}{c||}{ \bfseries{Original system} } & \multicolumn{3}{c||}{ \bfseries{Transformed system}\tnote{a} } & \multicolumn{2}{c|}{ \bfseries{Reference system} } \\

\hline

\bfseries {Quantity} & $\w_{k,i}$ & $\tilde{\w}_{k,i}$  & $\check{\w}_{k,i}$ & $\w_{c,i}$ & $\check{\w}_{c,i}$ & $\w_{e,i}$ & $\bar{w}_{c,i}$ & $\tilde{w}_{c,i}$ \\

\hline

\bfseries Definition & Iterate at agent $k$ & $w^o \!-\! \w_{k,i}$ & $\w_{k,i} \!-\! \bar{w}_{c,i}$ & $\displaystyle \sum_{k=1}^N \theta_k \w_{k,i}$ & $\w_{c,i} \!-\! \bar{w}_{c,i}$ & $\mU_R\w_i $ & Ref. Iterate & $w^o \!-\! \bar{w}_{c,i}$ \\

\hline

\bfseries Recursion  & Eqs. \eqref{Equ:ProblemFormulation:DisAlg_Comb1}--\eqref{Equ:ProblemFormulation:DisAlg_Comb2} & \cancel{\phantom{Nothing}} & \cancel{\phantom{Nothing}} & Eq. \eqref{Equ:Lemma:SignalDynamics:w_ci_recursion_tmp1} & Eq. \eqref{Equ:Lemma:ErrorDynamics:JointRec_wc_check} & Eq. \eqref{Equ:Lemma:ErrorDynamics:JointRec_TF2_we} & Eq. \eqref{Equ:LearnBehav:RefRec} & \cancel{\phantom{Nothing}} \\

\hline

%

%
%
%
%
%
%

\end{tabular}

\smallskip

\begin{tablenotes}

\item[a] The transformation is defined by \eqref{Equ:DistProc:wi_prime_def}--\eqref{Equ:DistProc:wicheck_prime_def}.



\end{tablenotes}

\end{threeparttable}

\vspace{-1\baselineskip}

\end{table*}

\subsection{Signal Recursions}
\label{Sec:StudyOfErrorDynamics:SignalRec}

We now derive the joint recursion that describes the evolution of the quantities 
$\check{\w}_{c,i} = \bw_{c,i}-\bar{w}_{c,i}$ and $\bw_{e,i}$.
Since $\bar{w}_{c,i}$ follows the reference recursion
\eqref{Equ:LearnBehav:RefRec}, it suffices to derive the joint recursion
for $\bw_{c,i}$ and $\bw_{e,i}$. To begin with,  we introduce the following global quantities:
	\begin{align}
		\mc{A}		&=	A \otimes I_M		\\
		\mc{A}_0		&=	A_0 \otimes I_M		\\
		\mc{A}_1		&=	A_1 \otimes I_M		\\
		\mc{A}_2		&=	A_2 \otimes I_M		\\
		\mc{M}		&=	\Omega \otimes I_M	\\
		\Omega		&=	\diag\{\mu_1,\ldots,\mu_N\}
	\end{align}
We also let the notation $x=\col\{x_1,\ldots,x_N\}$ denote an arbitrary 
$N\times 1$ block column vector that is formed by stacking 
$M\times 1$ sub-vectors $x_1,\ldots,x_N$ on top of each other.
We further define the following global update vectors:
			\begin{align}				
				\hat{\bm{s}}_i(x)
						&\defeq	\col\{
									\hat{\bm{s}}_{1,i}(x_1),
									\ldots,
									\hat{\bm{s}}_{N,i}(x_N)
								\}\\
				\label{Equ:Def:s}
				s(x)		&\defeq	\col\{s_1(x_1),\ldots,s_N(x_N)\}	
			\end{align}
Then, the general recursion  for the distributed strategy
\eqref{Equ:ProblemFormulation:DisAlg_Comb1}--\eqref{Equ:ProblemFormulation:DisAlg_Comb2}
can be rewritten in terms of these extended quantities as follows:
	\begin{align}
		\label{Equ:LearnBehav:GlobRec_orig}
		\bm{w}_i		=	\mc{A}^T \bm{w}_{i-1}
						-
						\mc{A}_2^T \mc{M} \hat{\bm{s}}_i(\bm{\phi}_{i-1})
	\end{align}
where
	\begin{align}
		\bm{\phi}_i	\defeq	\col\{\bm{\phi}_{1,i},\ldots,\bm{\phi}_{N,i}\}
		\label{Equ:DistProc:phi_i_def}
	\end{align}
and is related to $\bw_i$ and $\bw_i'$ via the following relation	
	\begin{align}
		\bm{\phi}_{i}	=	\mc{A}_1^T \bm{w}_{i}
						=	\mc{A}_1^T \mc{U} \bm{w}_{i}'
		\label{Equ:DistProc:relation_phi_w_wprime}
	\end{align}
Applying the transformation \eqref{Equ:DistProc:wi_prime_def} 
to both sides of \eqref{Equ:LearnBehav:GlobRec_orig},
we obtain the transformed global recursion:	
	\begin{align}
		\label{Equ:DistProc:GlobRec_TF}
		\bm{w}_i'	=	\mc{D} \bm{w}_{i-1}'
						-
						\mc{U}^{-1}\mc{A}_2^T \mc{M} 
						\hat{\bm{s}}_i\left(\bm{\phi}_{i-1}\right)
	\end{align}
We can now use
the block structures in \eqref{Equ:DistProc:mD_block} and
\eqref{Equ:DistProc:w_i_prime_w_ci_w_ei} to derive 
recursions for $\bw_{c,i}$ and $\bw_{e,i}$ from 
\eqref{Equ:DistProc:GlobRec_TF}. 
Substituting  \eqref{Equ:DistProc:mD_block} and
\eqref{Equ:DistProc:w_i_prime_w_ci_w_ei} into
\eqref{Equ:DistProc:GlobRec_TF}, and using properties of Kronecker 
products\cite[p.147]{laub2005matrix}, 
we obtain
	\begin{align}
		\bw_{c,i}	&=	\bw_{c,i-1} 
						- 
						(\theta^T \otimes I_M) \mA_2^T \mM 
						\hat{\bs}_{i}
						\left(
							\bm{\phi}_{i-1}
						\right)
						\nonumber\\
					&=	\bw_{c,i-1} 
						- 
						(\theta^T A_2^T \Omega \otimes I_M)
						\hat{\bs}_{i}
						\left(
							\bm{\phi}_{i-1}
						\right)
						\nonumber\\
					&=	\bw_{c,i-1} 
						- 
						\mu_{\max} \cdot (p^T \otimes I_M)
						\hat{\bs}_{i}
						\left(
							\bm{\phi}_{i-1}
						\right)
		\label{Equ:DistProc:w_ci_recursion_tmp1}
	\end{align}
and
	\begin{align}
		\bw_{e,i}	&=	\mD_{N-1} \bw_{e,i-1} 
						-
						\mU_{R} \mA_2^T \mM
						\hat{\bs}_i
						\left(
							\bm{\phi}_{i-1}
						\right)
		\label{Equ:DistProc:w_ei_recursion_tmp2}
	\end{align}
where in the last step of \eqref{Equ:DistProc:w_ci_recursion_tmp1}
we used the relation 
	\begin{align}
		\mu_{\max} \cdot p 	=	\Omega A_2 \theta 
		\label{Equ:DistProc:p_Omega_theta}
	\end{align}
which follows from Definitions \ref{Def:StepSize} and \ref{Def:UsefulVectors}.
Furthermore, by adding and subtracting identical factors, the term $\hat{\bs}_{i}
						\left(
							\bm{\phi}_{i-1}
						\right)$
that appears in \eqref{Equ:DistProc:w_ci_recursion_tmp1} and
\eqref{Equ:DistProc:w_ei_recursion_tmp2}
can be expressed as
	\begin{align}
		\hat{\bs}_{i}
		\left(
			\bm{\phi}_{i-1}
		\right)			&=			s(\one \otimes \bw_{c,i-1})
									+
									\underbrace{
									\hat{\bs}_{i}
									\left(
										 \bm{\phi}_{i-1}
									\right)
									\!-\!
									s\left(
										\bm{\phi}_{i-1}
									\right)
									}_{\defeq\bv_i(\bm{\phi}_{i-1})}
									\nn\\
									&\quad
									+
									\underbrace{
									s\left(
										\bm{\phi}_{i-1}
									\right)
									\!-\!
									s(\one \otimes \bw_{c,i-1})
									}_{\defeq \bz_{i-1}}
		\label{Equ:DistProc:s_hat_phi_decomp_zdef_vdef}
	\end{align}
where the first perturbation term $\bv_i(\bm{\phi}_{i-1})$ consists of the difference between
the true update vectors $\{s_k(\bm{\phi}_{k,i-1})\}$ and their stochastic approximations
$\{\hat{\bs}_{k,i}(\bm{\phi}_{k,i-1})\}$,
while the second perturbation term $\bz_{i-1}$ represents the difference between
the same $\{s_k(\bm{\phi}_{k,i-1})\}$ and $\{s_k(\bw_{c,i-1})\}$.
The subscript 
$i-1$ in $\bm{z}_{i-1}$ implies that this variable depends on data up to time
$i-1$ and the subscript $i$ in $\bm{v}_i(\bm{\phi}_{i-1})$ implies that its value depends on
data up to time $i$ (since, in general, $\hat{\s}_i(\cdot)$ can depend on data from time $i$ --- see Eq. \eqref{P2-Equ:Example:ski_LMS} in Part II for an example).
Then, $\hat{\bs}_{i}
		\left(
			\bm{\phi}_{i-1}
		\right)$
can be expressed as
	\begin{align}
		\hat{\bs}_{i}
		\left(
			\bm{\phi}_{i-1}
		\right)
						&=			s(\one \otimes \bw_{c,i-1})
									+
									\bv_{i}
									+
									\bz_{i-1}
		\label{Equ:DistProc:si_hat_expr}
	\end{align}

\begin{lemma}[Signal dynamics]
	\label{Lemma:SignalDynamics}
	In summary, the previous derivation shows that the weight iterates at each agent evolve according to the following
	dynamics: 
		\begin{align}
			\label{Equ:Lemma:SignalDynamics:w_ki_decomp}
			\bw_{k,i}	&=	\bw_{c,i} + (u_{L,k}\otimes I_M) \bw_{e,i}		
						\\			
			\bw_{c,i}	&=	\bw_{c,i-1} 
						- 
						\mu_{\max} \cdot (p^T \otimes I_M)
						\hat{\bs}_{i}
						\left(
							\bm{\phi}_{i-1}
						\right)
			\label{Equ:Lemma:SignalDynamics:w_ci_recursion_tmp1}
						\\
			\bw_{e,i}	&=	\mD_{N-1} \bw_{e,i-1} 
						-
						\mU_{R} \mA_2^T \mM
						\hat{\bs}_i
						\left(
							\bm{\phi}_{i-1}
						\right)
			\label{Equ:Lemma:SignalDynamics:w_ei_recursion_tmp2}
						\\
			\hat{\bs}_{i}
			\left(
				\bm{\phi}_{i-1}
			\right)
							&=			s(\one \otimes \bw_{c,i-1})
										+
										\bv_{i}
										+
										\bz_{i-1}
			\label{Equ:Lemma:SignalDynamics:si_hat_expr}
		\end{align}
	\hfill\QED
\end{lemma}

\subsection{Error Dynamics}
\label{Sec:StudyOfErrorDynamics:ErrorRec}
To simplify the notation, we introduce the centralized operator $T_c: \mathbb{R}^{M} \rightarrow \mathbb{R}^M$ as the following mapping for any $x\in \mathbb{R}^M$:
	\begin{align}
		T_c(x)	&\defeq	x 
						- 
						\mu_{\max}\cdot 
						(p^T \otimes I_M) \; s(\mathds{1}\otimes x)
						\nn\\
				&=		x - \mu_{\max}\sum_{k=1}^N p_k s_k(x)
		\label{Equ:Def:Tc2}
	\end{align}
Substituting \eqref{Equ:DistProc:si_hat_expr} into
\eqref{Equ:Lemma:SignalDynamics:w_ci_recursion_tmp1}--\eqref{Equ:Lemma:SignalDynamics:w_ei_recursion_tmp2}
and using \eqref{Equ:Def:Tc2}, we find that we can rewrite
\eqref{Equ:Lemma:SignalDynamics:w_ci_recursion_tmp1} and
\eqref{Equ:Lemma:SignalDynamics:w_ei_recursion_tmp2} in the alternative form:
	\begin{align}
		\label{Equ:DistProc:JointRec_TF2_wc}
			\bm{w}_{c,i}	&=	T_c(\bm{w}_{c,i-1}) - 
							\mu_{\max}\cdot (p^T \otimes I_M) 
							\left[
									\bm{z}_{i-1}
									+
									\bm{v}_i
							\right]
							\\
			\bm{w}_{e,i}	&=	\mD_{N-1} \bm{w}_{e,i-1}
							-
							\mU_R\mA_2^T \mM
							\left[
									s(\mathds{1} \otimes \bm{w}_{c,i-1})
									+
									\bm{z}_{i-1}
									\!+\!
									\bm{v}_i
							\right]
		\label{Equ:DistProc:JointRec_TF2_we}
	\end{align}
Likewise, we can write the reference recursion
\eqref{Equ:LearnBehav:RefRec} in the following compact form:
	\begin{align}
		\label{Equ:DistProc:RefRec_Tc_form}
		\bar{w}_{c,i}	=	T_c(\bar{w}_{c,i-1})
	\end{align}
Comparing \eqref{Equ:DistProc:JointRec_TF2_wc} with
\eqref{Equ:DistProc:RefRec_Tc_form}, we notice that
the recursion for the centroid vector, $\bm{w}_{c,i}$, follows
the same update rule as the reference recursion
except for the two driving perturbation terms
$\bm{z}_{i-1}$ and $\bm{v}_i$. Therefore, we would expect
the trajectory of $\bw_{c,i}$ to be a perturbed version
of that of $\bar{w}_{c,i}$.
Recall from \eqref{Equ:DistProc:wcicheck_def} that
	\begin{align}
		\check{\bw}_{c,i}	\defeq	\bw_{c,i}-\bar{w}_{c,i}
								\nn
	\end{align}
To obtain the dynamics of $\check{\bw}_{c,i}$, we subtract \eqref{Equ:DistProc:RefRec_Tc_form} 
from \eqref{Equ:DistProc:JointRec_TF2_wc}.
	
\begin{lemma}[Error dynamics]
	\label{Lemma:ErrorDynamics}
	The error quantities that appear on the right-hand side of \eqref{Equ:DistProc:wicheck_blockstructure} 
	evolve according to the following dynamics: 
		\begin{align}
			\check{\bm{w}}_{c,i}	
							&=	T_c(\bm{w}_{c,i-1}) 
								-
								T_c(\bar{w}_{c,i-1})
								\nn\\
								&\quad
								-
								\mu_{\max} \!\cdot\! (p^T \!\otimes\! I_M) 
								\left[
										\bm{z}_{i-1}
										+
										\bm{v}_i
								\right]
			\label{Equ:Lemma:ErrorDynamics:JointRec_wc_check}
							\\
				\bm{w}_{e,i}	&=	\mD_{N-1} \bm{w}_{e,i-1}
								\nn\\
								&\quad
								-
								\mU_R\mA_2^T \mM
								\left[
										s(\mathds{1} \otimes \bm{w}_{c,i-1})
										\!+\!
										\bm{z}_{i-1}
										\!+\!
										\bm{v}_i
								\right]
			\label{Equ:Lemma:ErrorDynamics:JointRec_TF2_we}
		\end{align}
	\hfill\QED
\end{lemma}

\noindent
The analysis in sequel will study the dynamics of the variances of the error quantities $\check{\bw}_{c,i}$ and ${\bw}_{e,i}$ based on \eqref{Equ:Lemma:ErrorDynamics:JointRec_wc_check}--\eqref{Equ:Lemma:ErrorDynamics:JointRec_TF2_we}. The main challenge is that these two recursions are coupled with each other through $\bz_{i-1}$ and $\bv_{i}$. To address the difficulty, we will extend the energy operator approach developed in \cite{chen2013JSTSPpareto} to the general scenario under consideration.

\subsection{Energy Operators}
\label{Sec:LearnBehav:Operator}

To carry out the analysis, 
we need to introduce the following
operators.

\begin{definition}[Energy vector operator]
    \label{Def:PowerOperator}
    Suppose $x=\col\{x_1,\ldots,x_N\}$ is an arbitrary 
    $N \times 1$ block column vector that is formed
    by stacking $M_0 \times 1$ vectors $x_1,\ldots, x_N$ 
    on top of each other. The energy vector operator 
    $P_{M_0}: \mb{C}^{M_0 N} \rightarrow \mb{R}^{N}$ is defined as the mapping:
        \begin{align}
            \label{Equ:Def:PowerVectorOperator}
            P_{M_0}[x]    \defeq  \col\{ \|x_1\|^2, \ldots, \|x_N\|^2\}
        \end{align}
    where $\|\cdot\|$ denotes the Euclidean norm of a vector.
    \hfill\QED
\end{definition}

\begin{definition}[Norm matrix operator]
    \label{Def:MatrixPowerOperator}
    Suppose $X$ is an arbitrary $K \times N$ block matrix consisting of blocks
    $\{X_{kn}\}$ of size $M_0 \times M_0$:
    		\begin{align}
    			\label{Equ:Def:X}
    			X	&=	\begin{bmatrix}
    						X_{11}	&	\cdots	&	X_{1N}	\\
    						\vdots	&			&	\vdots	\\
    						X_{K1}	&	\cdots	&	X_{KN}	
		    			\end{bmatrix}
    		\end{align}
    The norm matrix operator 
    $\bP_{M_0}: \mb{C}^{M_0K \times M_0N} 
    \rightarrow \mb{R}^{K \times N}$ is defined as the mapping:
        \begin{align}
            \label{Equ:Def:NormMatrixOperator}
            \bP_{M_0}[x]    \defeq  
            			\begin{bmatrix}
    						\|X_{11}\|	&	\cdots	&	\|X_{1N}\|	\\
    						\vdots	&			&	\vdots	\\
    						\|X_{K1}\|	&	\cdots	&	\|X_{KN}	\|
		    			\end{bmatrix}
        \end{align}
    where $\|\cdot\|$ denotes the $2-$induced norm of a matrix.
    \hfill\QED
\end{definition}

By default, we choose $M_0$ to be $M$, the size of the vector $\bw_{k,i}$. In this case, we will drop the subscript in$P_{M_0}[\cdot]$ and use $P[\cdot]$ for convenience. However, in other cases, we will keep the subscript to avoid confusion. Likewise, $\bP_{M_0}[\cdot]$ characterizes the norms of different parts of a matrix it operates on. We will also drop the subscript if $M_0=M$.
{
In Appendix \ref{Appendix:PropertyEnergyOperatorStatement}, we collect several properties of the above energy operators, and which will be used in the sequel to characterize how the energy of the error quantities propagates through the dynamics \eqref{Equ:Lemma:ErrorDynamics:JointRec_wc_check}--\eqref{Equ:Lemma:ErrorDynamics:JointRec_TF2_we}.
}


\section{Transient Analysis}
\label{Sec:LearnBehav:Transient}

Using the energy operators and the various properties, we can now examine the transient behavior of the learning curve more closely. Recall from \eqref{Equ:DistProc:wki_tilde_decomposition_final} that $\tilde{\w}_{k,i}$ consists of three parts: the error of the reference recursion, $\tilde{w}_{c,i}$, the difference between the centroid and the reference, $\check{\w}_{c,i}$, and the position of individual iterates relative to the centroid, $(u_{L,k} \otimes I_M) \w_{e,i}$. The main objective in the sequel is to study the convergence of the reference error, $\tilde{w}_{c,i}$, and establish non-asymptotic bounds for the mean-square values of $\check{\w}_{c,i}$ and $\w_{e,i}$, which will allow us to understand how fast and how close the iterates at the individual agents, $\{\w_{k,i}\}$, get to  the reference recursion. Recalling from \eqref{Equ:DistProc:wicheck_blockstructure} that $\check{\w}_{c,i}$ and $\w_{e,i}$ are the two blocks of the transformed vector $\check{\w}_{i}'$ defined by \eqref{Equ:DistProc:wicheck_prime_def}, we can examine instead the evolution of
	\begin{align}
		\check{\mW}_{i}'	&\defeq	\E P[\check{\bw}_{i}']
						=		\col\left\{
										\E P[\check{\bw}_{c,i}], \E P[\bw_{e,i}]
									\right\}
								\nn\\
						&=		
								\col\left\{
										\E \|\check{\bw}_{c,i}\|^2, \E P[\bw_{e,i}]
								\right\}	
		\label{Equ:DistProc:mW_check_prime_def}
	\end{align}
Specifically, we will study the convergence of $\tilde{w}_{c,i}$ in Sec. \ref{Sec:Transient:LimitPoint}, the stability of $\check{\mW}_{i}'$ in Sec. \ref{Sec:Transient:MSStability}, and the two transient phases of $\tilde{\w}_{k,i}$ in Sec. \ref{Sec:Transient:Interpret}.

\subsection{Limit Point}
\label{Sec:Transient:LimitPoint}

Before we proceed to study $\check{\mW}_{i}'$, we state
the following theorems on the existence of a limit point and on the convergence of the reference
recursion \eqref{Equ:DistProc:RefRec_Tc_form}.
	\begin{theorem}[Limit point]
		\label{Thm:LimitPoint}
		Given Assumptions \ref{Assumption:UpdateVectorLipschitz}--\ref{Assumption:UpdateVectorMonot}, 
		there exists a unique $M \times 1$ vector $w^o$ that solves
			\begin{align}
				\label{Equ:Lemma:LimitPoint_def}
				\sum_{k=1}^N p_k s_k(w^o)	=	0
			\end{align}
		where $p_k$ is the $k$th entry of the vector $p$ defined in 
		\eqref{Def:UsefulVectors:p}.
	\end{theorem}
	\begin{proof}
		See Appendix \ref{Appendix:Proof_Thm_LimitPoint}.
	\end{proof}
	\begin{theorem}[Convergence of the reference recursion]
		\label{Thm:ConvergenceRefRec:DeterministcCent}
		Let $\tilde{w}_{c,i}\defeq w^o-\bar{w}_{c,i}$ denote the 
		error vector of the reference recursion
		\eqref{Equ:DistProc:RefRec_Tc_form}.
		Then, the following non-asymptotic bound on the squared
		error holds for all $i\ge 0$:
			\begin{align}
				\label{Equ:Thm:ConvergenceRefRec:NonAsympBound}
						(1\!-\!2\mu_{\max} \|p\|_1 \lambda_U)^i 
						\!\cdot\! \|\tilde{w}_{c,0}\|^2	
						\le 
						\|\tilde{w}_{c,i}\|^2	
						\le	\gamma_c^{2i} 
							\!\cdot\!
							\|\tilde{w}_{c,0}\|^2
			\end{align}
		where
			\begin{align}
				\gamma_c		
						\defeq
								1 - \mu_{\max} \lambda_L 
								+ \frac{1}{2} \mu_{\max}^2 \|p\|_1^2 \lambda_U^2
				\label{Equ:thm:ConvergenceRefRec:gamma_c_def}
			\end{align}
		Furthermore, if the following condition on the step-size holds
			\begin{align}
				\label{Equ:Thm:ConvergenceRefRec:StepSize}
				0	<	\mu_{\max}	<	\frac{2\lambda_L}{\|p\|_1^2 \lambda_U^2}
			\end{align}
		then, the iterate $\tilde{w}_{c,i}$ converges
		to zero.
	\end{theorem}
	\begin{IEEEproof}
		See Appendix \ref{Appendix:Proof_Thm_ConvergenceRefRec}.
	\end{IEEEproof}
	
Note from \eqref{Equ:Thm:ConvergenceRefRec:NonAsympBound} that, when the
step-size is sufficiently small, the reference recursion \eqref{Equ:LearnBehav:RefRec} 
converges at a geometric
rate between $1-2\mu_{\max} \|p\|_1 \lambda_U$ and 
$\gamma_c^2 = 1-2\mu_{\max} \lambda_L + o(\mu_{\max})$. 
{
Note that this is a \emph{non-asymptotic} result. That is, the convergence rate $r_{\mathrm{Ref}}$ of the reference recursion \eqref{Equ:DistProc:RefRec_Tc_form} is always lower and upper bounded by these two rates:
	\begin{align}
		1-2\mu_{\max} \|p\|_1 \lambda_U	
		\le
		r_{\mathrm{Ref}}
		\le
		\gamma_c^2,
		\quad
		\forall i \ge 0
		\label{Equ:DistProc:r_RefRec_NonAsympUBLB}
	\end{align}
We can obtain a more precise characterization
of the convergence rate of the reference recursion in the asymptotic regime (for large enough $i$), as follows.
}%

\begin{theorem}[Convergence rate of the reference recursion]
\label{Thm:ConvergenceRateRefRec}
Specifically, for any small $\epsilon>0$, there exists a time instant $i_0$ such that, for $i \ge i_0$, the error vector $\tilde{w}_{c,i}$ converges to zero at the following rate:
	\begin{align}
		r_{\mathrm{Ref}}
				&=		\big[
							\rho(I_M - \mu_{\max} H_c)
						\big]^2
						+ 
						O\big( (\mu_{\max} \epsilon )^{\frac{1}{2(M-1)}} \big)
		\label{Equ:DistProc:r_RefRec}
	\end{align}
\end{theorem}
\begin{proof}
	See Appendix \ref{Appendix:Proof_Thm_ConvergenceRefRec_approx}.
\end{proof}
Note that since \eqref{Equ:DistProc:r_RefRec} holds for arbitrary $\epsilon>0$, we can choose $\epsilon$ to be an arbitrarily small positive number. Therefore, the convergence rate of the reference recursion is arbitrarily close to $[\rho(I_M - \mu_{\max} H_c)]^2$.

\subsection{Mean-Square Stability}
\label{Sec:Transient:MSStability}

Now we apply the properties from Lemmas
\ref{Lemma:BasicPropertiesOperator}--\ref{Lemma:VarianceRelations}
to derive an inequality recursion for the transformed energy
vector $\check{\mW}_{i}' = \E P[\check{\bw}_{i}']$.
The results are summarized in the following lemma.
	\begin{lemma}[Inequality recursion for $\check{\mW}_{i}'$]
		\label{Lemma:IneqRecur_W_check_prime}
		The $N\times 1$ vector $\check{\mW}_{i}'$ defined by
		\eqref{Equ:DistProc:mW_check_prime_def} satisfies the 
		following relation for all time instants:
			\begin{align}
				\label{Equ:FirstOrderAnal:W_i_prime_ineq_Rec1}
				\check{\mc{W}}_i'	\preceq		\Gamma \check{\mc{W}}_{i-1}' + \mu_{\max}^2 b_v
			\end{align}
		where
			\begin{align}
				\label{Equ:FirstOrderAnal:Gamma_def}
				\Gamma		&\defeq		\Gamma_0
										+
										\mu_{\max}^2 \psi_0 \cdot \mathds{1}\mathds{1}^T
										\in \mathbb{R}^{N \times N}
										\\
				\label{Equ:FirstOrderAnal:Gamma0_def}
				\Gamma_0		&\defeq		\begin{bmatrix}
											\gamma_c		&	\mu_{\max} h_c(\mu_{\max}) \cdot \mathds{1}^T
															\\
											0			&	\Gamma_e
										\end{bmatrix}
										\in
										\mathbb{R}^{N \times N}
										\\
				b_v			&\defeq		\col\{b_{v,c},\; b_{v,e} \cdot \one\} \in \mathbb{R}^{N}
				\label{Equ:FirstOrderAnal:bv_def}
										\\
           			\label{Equ:FirstOrderAnal:Gamma_e_def}
           			\Gamma_e		&\defeq	\begin{bmatrix}
	           									|\lambda_2(A)| & \frac{2}{1-|\lambda_2(A)|} & &
	           														  	\\
	           									    & \ddots & \ddots &	\\
	           									    &\!\! & \ddots & 
	           									    	\frac{2}{1-|\lambda_2(A)|}	\\
	           									    & & & |\lambda_2(A)|
	           								\end{bmatrix}
										\!
										\in
										\mathbb{R}^{(N\!-\!1)\times (N\!-\!1)}
			\end{align}
		The scalars $\psi_0$, $h_c(\mu)$,
		$b_{v,c}$ and $b_{v,e}$ are defined as
			\begin{align}
				\psi_0		&\defeq		\max\bigg\{
											4 \alpha \|p\|_1^2,
											\;
											\;
											4 \alpha\|p\|_1^2 \cdot
											\left\|
													\bP[\mc{A}_1^T \mc{U}_L]
											\right\|_{\infty}^2, 
											\;\;
											\nn\\
											&\qquad
											4N
											\cdot
											\left\|
												\bP[
													\mc{U}_R \mc{A}_2^T
												]
											\right\|_{\infty}^2
											\lambda_U^2
											\left(
												\frac{3}{1-|\lambda_2(A)|}
												+
												\frac{\alpha}{\lambda_U^2}
											\right),
											\nonumber\\
											&\qquad
											4N
											\cdot
											\left\|
												\bP[
													\mc{U}_R \mc{A}_2^T
												]
											\right\|_{\infty}^2
											\cdot
											\left\|
												\bP[
													\mc{A}_1^T \mc{U}_L
												]
											\right\|_{\infty}^2
											\lambda_U^2
											\nn\\
											&\qquad	
											\cdot				
											\left(
												\frac{1}{1-|\lambda_2(A)|} + \frac{\alpha}{\lambda_U^2}
											\right)
										\bigg\}
				\label{Equ:Lemma:IneqRecur:psi0_def}
										\\
				h_c(\mu_{\max})		&\defeq		\|p\|_1^2 \!\cdot\!
										\left\|
												\bP[\mc{A}_1^T \mc{U}_L]
										\right\|_{\infty}^2
										\lambda_U^2
										\!\cdot\!
										\Big[
											\frac{1}
											{
												\lambda_L 
												\!-\! 
												\frac{1}{2}\mu_{\max}\|p\|_1^2 \lambda_U^2
											}												
										\Big]
				\label{Equ:Lemma:IneqRecur:hc_def}
										\\
				b_{v,c}		&\defeq		\|p\|_1^2 \cdot
											\left[
												4\alpha
												( \|\tilde{w}_{c,0}\|^2 + \|w^o\|^2 )
												+
												\sigma_v^2
											\right]
				\label{Equ:Lemma:IneqRecur:bvc_def}
										\\
				b_{v,e}		&\defeq		N
										\left\|
											\bP[
												\mc{U}_R \mc{A}_2^T
											]
										\right\|_{\infty}^2
										\bigg(
												12
												\frac{\lambda_U^2 \|\tilde{w}_{c,0}\|^2 
												+ 
												\|g^o\|_{\infty}}
												{1-|\lambda_2(A)|}
												\nn\\
												&\qquad
												+
												4\alpha
												( \|\tilde{w}_{c,0}\|^2 + \|w^o\|^2 )
												+
												\sigma_v^2
										\bigg)
				\label{Equ:Lemma:IneqRecur:bve_def}
			\end{align}
		where $g^o		\defeq	P[s(\mathds{1} \otimes w^o)]$.
	\end{lemma}
	\begin{IEEEproof}
		See Appendix \ref{Appendix:Proof_Lemma_W_check_prime_recursion}.
	\end{IEEEproof}

From \eqref{Equ:FirstOrderAnal:Gamma_def}--\eqref{Equ:FirstOrderAnal:Gamma0_def}, we see that as the step-size $\mu_{\max}$ becomes small, we have $\Gamma \approx \Gamma_0$, since the second term in the expression for $\Gamma$ depends on the square of the step-size. Moreover, note that $\Gamma_0$ is an upper triangular matrix. Therefore, $\check{\bw}_{c,i}$ and $\bw_{e,i}$ are weakly coupled for small step-sizes; $\E P[\bw_{e,i}]$ evolves on its own, but it will seep into the evolution of $\E P[\check{\bw}_{c,i}]$ via the off-diagonal term in $\Gamma_0$, which is $O(\mu_{\max})$. This insight is exploited to establish a non-asymptotic bound  on $\check{\mW}_{i}' = \col\{ \Expt\| \check{\w}_{c,i}\|^2, \Expt P[\w_{e,i}]\}$ in the following theorem.
	\begin{theorem}[Non-asymptotic bound for $\check{\mW}_{i}'$]
		\label{Thm:NonAsymptoticBound}
		Suppose the matrix $\Gamma$ defined in \eqref{Equ:FirstOrderAnal:Gamma_def}
		is stable, i.e., $\rho(\Gamma)<1$. Then,
		the following non-asymptotic bound holds for all $i \ge 0$:
			\begin{align}
				\E P[\check{\bw}_{c,i}]	&\preceq	\mu_{\max} h_c(\mu_{\max})
												\!\cdot\!
												\one^T 
												(\gamma_c I \!-\! \Gamma_e)^{-1}
												\left(
													\gamma_c^i I
													\!-\!
													\Gamma_e^i
												\right)
												\mW_{e,0}
												\nn\\
												&\quad
												+
												\check{\mW}_{c,\infty}^{\mathrm{ub}'}
				\label{Equ:DistProc:EPwcicheck_bound}
												\\
				\E P[\bw_{e,i}]			&\preceq 
												\Gamma_e^i \mW_{e,0}
												+
												\check{\mW}_{e,\infty}^{\mathrm{ub}'}
				\label{Equ:DistProc:EPwei_bound}
			\end{align}
		where $\mW_{e,0} \defeq \E P[\bw_{e,0}]$, $\check{\mW}_{c,\infty}^{\mathrm{ub}'}$ and 
		$\check{\mW}_{e,\infty}^{\mathrm{ub}'}$ are the $\limsup$ bounds of $\E P[\check{\bw}_{c,i}]$
		and $\E P[\bw_{e,i}]$, respectively:
			\begin{align}
				\check{\mW}_{c,\infty}^{\mathrm{ub}'}	
						&=				\mu_{\max}\cdot
										\frac{
												\psi_0 \big(\lambda_L+h_c(0)\big) 
												\one^T (I-\Gamma_e)^{-1} \mW_{e,0} 
												+ b_{v,c} \lambda_L
											}
										    {\lambda_L^2}
										    \nn\\
										    &\quad
										 + o(\mu_{\max})
				\label{Equ:DistProc:EPwci_Omu}
										\\
				\check{\mW}_{e,\infty}^{\mathrm{ub}'}
						&=				\mu_{\max}^2 
										\!\cdot\! 
										\frac{	
												\psi_0 \big( \lambda_L + h_c(0) \big) 
												\one^T (I-\Gamma_e)^{-1} \mW_{e,0} + b_{v,e}\lambda_L
											}
											{\lambda_L}
											\nn\\
											&\quad
											\times
										(I\!-\!\Gamma_e)^{-\!1} \one
										+
										o(\mu_{\max}^2)
				\label{Equ:DistProc:EPwei_Omu}
			\end{align}	
		where $o(\cdot)$ denotes strictly higher order terms, and $h_c(0)$ is the value of $h_c(\mu_{\max})$ (see
		\eqref{Equ:Lemma:IneqRecur:hc_def}) evaluated at 
		$\mu_{\max}=0$. 
		An important implication of \eqref{Equ:DistProc:EPwcicheck_bound} and
			\eqref{Equ:DistProc:EPwci_Omu} is that 
				\begin{align}
					\E P[\check{\w}_{c,i} ] \le O(\mu_{\max}),
					\qquad\quad
					\forall i \ge 0
					\label{Equ:Thm:NonasymptoticBound:EPwci_check_Omu_bound}
				\end{align}
		Furthermore, a sufficient condition that guarantees the stability of the matrix $\Gamma$ is that
			\begin{align}
				0	<	\mu_{\max}		
					<	\min\bigg\{\!\!&
							\frac{\lambda_L}
							{\frac{1}{2}\|p\|_1^2 \lambda_U^2 \!+\! \frac{1}{3}\psi_0
							\left( \! \frac{1 - |\lambda_2(A)|}{2} \! \right)^{ \! -2N}}, 
							\;
							\nn\\
							&
							\sqrt{\frac{3(1 \!-\! |\lambda_2(A)|)^{2N \!+\! 1}}{2^{2N + 2}\psi_0}},
							\;
							\nn\\
							&
							\frac{\lambda_L}
								{\|p\|_1^2\lambda_U^2 \!
								\left(\|\bP_1[A^T U_L]\|_{\infty}^2 \!+\! \frac{1}{2}\right)}
							\!\!
						\bigg\}
				\label{Equ:Thm_NonAsympBound:StepSize}
			\end{align}
	\end{theorem}
	\begin{proof}
		See Appendix \ref{Appendix:Proof_Thm_NonAsymptotiBound}.
	\end{proof}
	\begin{corollary}[Asymptotic bounds]
		\label{Cor:AsymptoticBounds_wcicheck_we}
		It holds that
		\begin{align}
			\limsup_{i\rightarrow \infty} \E \|\check{\bw}_{c,i}\|^2
								&\le		O(\mu_{\max})
			\label{Equ:Cor:EPwcicheck_asymptotic_bound}
											\\
			\limsup_{i\rightarrow \infty} \E \|\bw_{e,i}\|^2
								&\le 
											O(\mu_{\max}^2)
			\label{Equ:Cor:EPwei_asymptotic_bound}
		\end{align}
	\end{corollary}
	\begin{proof}
		The bound \eqref{Equ:Cor:EPwcicheck_asymptotic_bound} holds since 
		$\E \|\check{\w}_{c,i}\|^2 = \E P[\check{\w}_{c,i}] \le O(\mu_{\max})$ for all $i \ge 0$
		according to \eqref{Equ:Thm:NonasymptoticBound:EPwci_check_Omu_bound}.
		Furthermore, inequality \eqref{Equ:Cor:EPwei_asymptotic_bound} holds because
			\begin{align}
				\limsup_{i\rightarrow \infty} \E \|\bw_{e,i}\|^2
									&\overset{(a)}{=}												
												\limsup_{i\rightarrow \infty} 
												\one^T\E P[\bw_{e,i}]
												\nn\\
									&\overset{(b)}{\preceq}
												\one^T \check{\mW}_{e,\infty}^{\mathrm{ub}'}
												\nn\\
									&\overset{(c)}{=}
												O(\mu_{\max}^2)
				\label{Equ:DistProc:EPwei_asymptotic_bound}
			\end{align}
		where step (a) uses property \eqref{Equ:Properties:PX_EuclNorm} of the energy operator 
		$P[\cdot]$, step (b)
		uses \eqref{Equ:DistProc:EPwei_bound}, and step (c) uses \eqref{Equ:DistProc:EPwei_Omu}.
	\end{proof}
	
Finally, we present following main theorem that characterizes the difference between the learning curve of $\tilde{\w}_{k,i}$ at each agent $k$ and that of $\tilde{w}_{c,i}$ generated by the reference recursion \eqref{Equ:DistProc:RefRec_Tc_form}.
	\begin{theorem}[Learning behavior of $\E\|\tilde{\w}_{k,i}\|^2$]
		\label{Thm:LearnBehav_wki}
		Suppose the stability condition \eqref{Equ:Thm_NonAsympBound:StepSize} holds. Then, 
		the difference between the learning curve of the mean-square-error $\E\|\tilde{\w}_{k,i}\|^2$ 
		at each agent $k$ and the learning curve of $\|\tilde{w}_{c,i}\|^2$ is bounded
		{
		\emph{non-asymptotically}
		}
		as
			\begin{align}
				\big|
					\Expt & \| \tilde{\w}_{k,i}\|^2
					-
					\| \tilde{w}_{c,i} \|^2
				\big|
							\nn\\
						&\le 
							2 \| u_{L,k} \otimes I_M \|^2
							\cdot	
							\one^T\Gamma_e^i \mW_{e,0}
							\nn\\
							&\quad
							+
							2
							\|\tilde{w}_{c,0}\|
							\cdot
							\| u_{L,k} \otimes I_M \|
							\cdot
							\sqrt{
									\one^T\Gamma_e^i \mW_{e,0}
							}
							\nn\\
							&\quad
							+
							\gamma_c^i
							\cdot
							O(\mu_{\max}^{\frac{1}{2}})
							+
							O(\mu_{\max})
							\qquad\quad
							\mathrm{for~all~} i \ge 0
				\label{Equ:Thm:LearnBehav_wki:Bound}
			\end{align}
		where $\gamma_c$ was defined earlier in \eqref{Equ:thm:ConvergenceRefRec:gamma_c_def}.
	\end{theorem}
	\begin{proof}
		See Appendix \ref{Appendix:Proof_LearnBehav_wki}.
	\end{proof}
	
\subsection{Interpretation of Results}
\label{Sec:Transient:Interpret}

\begin{table*}[t]
	\scriptsize
	\centering
	\caption{Behavior of error quantities in different phases.}
	\label{Tab:BehavErrorPhases}
	\begin{threeparttable}
		\begin{tabular}{|c||c|c||c|c||c|}
			\hline
			\multirow{2}{*}{\bfseries{Error quantity}}
			& 	\multicolumn{2}{c||}{ \bf Transient Phase I} 
			& \multicolumn{2}{c||}{\bf Transient Phase II} & \bfseries Steady-State \tnote{c}	\\
			\cline{2-6}
			 & \bfseries{Convergence rate $r$} \hspace{0.5em}\tnote{a} &   \bfseries{Value}	
			 &  \bfseries{Convergence rate $r$} \hspace{0.5em} \tnote{b} &   \bfseries{Value} &     \bfseries{Value} 
			 \\
			 \hline
			$\| \tilde{w}_{c,i} \|^2$ & $1-2\mu_{\max} \|p\|_1 \lambda_U \le r \le \gamma_c^2$  
			& $\gg O(\mu_{\max})$  & $r = [ \rho(I_M - \mu_{\max} H_c )]^2$ 
			& $\gg O(\mu_{\max})$ & 0 \\
			\hline
			$\E\|\check{\w}_{c,i}\|^2$ & converged & $O(\mu_{\max})$ & converged 
			& $O(\mu_{\max})$ & $O(\mu_{\max})$ \\
			\hline
			$\E P[\check{\w}_{e,i}]$ & $r \le |\lambda_2(A)|$  & $ \gg O(\mu_{\max})$ 
			& converged & $O(\mu_{\max}^2)$ & $O(\mu_{\max}^2)$ \\
			\hline
			$\Expt \|\tilde{\w}_{k,i} \|^2$ & Multiple modes & $\gg O(\mu_{\max})$ 
			&  $r = [ \rho(I_M - \mu_{\max} H_c )]^2$ & $ \gg O(\mu_{\max})$ & $O(\mu_{\max})$ \\
			\hline
		\end{tabular}
		\smallskip
		\begin{tablenotes}
			\item[a] 
				$\gamma_c$ is defined in \eqref{Equ:thm:ConvergenceRefRec:gamma_c_def}, 
				and $\gamma_c^2 = 1 - 2\mu_{\max} \lambda_L + o(\mu_{\max})$.
		
			\item[b]
				We only show the leading term of the convergence rate for $r$. More precise 
				expression can be found in \eqref{Equ:DistProc:r_RefRec}.
				
			\item[c]
				Closer studies of the steady-state performance can be found in 
				Part II \cite{chen2013learningPart2}.
		\end{tablenotes}
	\end{threeparttable}
\vspace{-1\baselineskip}
\end{table*}

\begin{figure*}[t]
	\centerline{
		\subfigure[]{
			\includegraphics[width=0.48\textwidth]{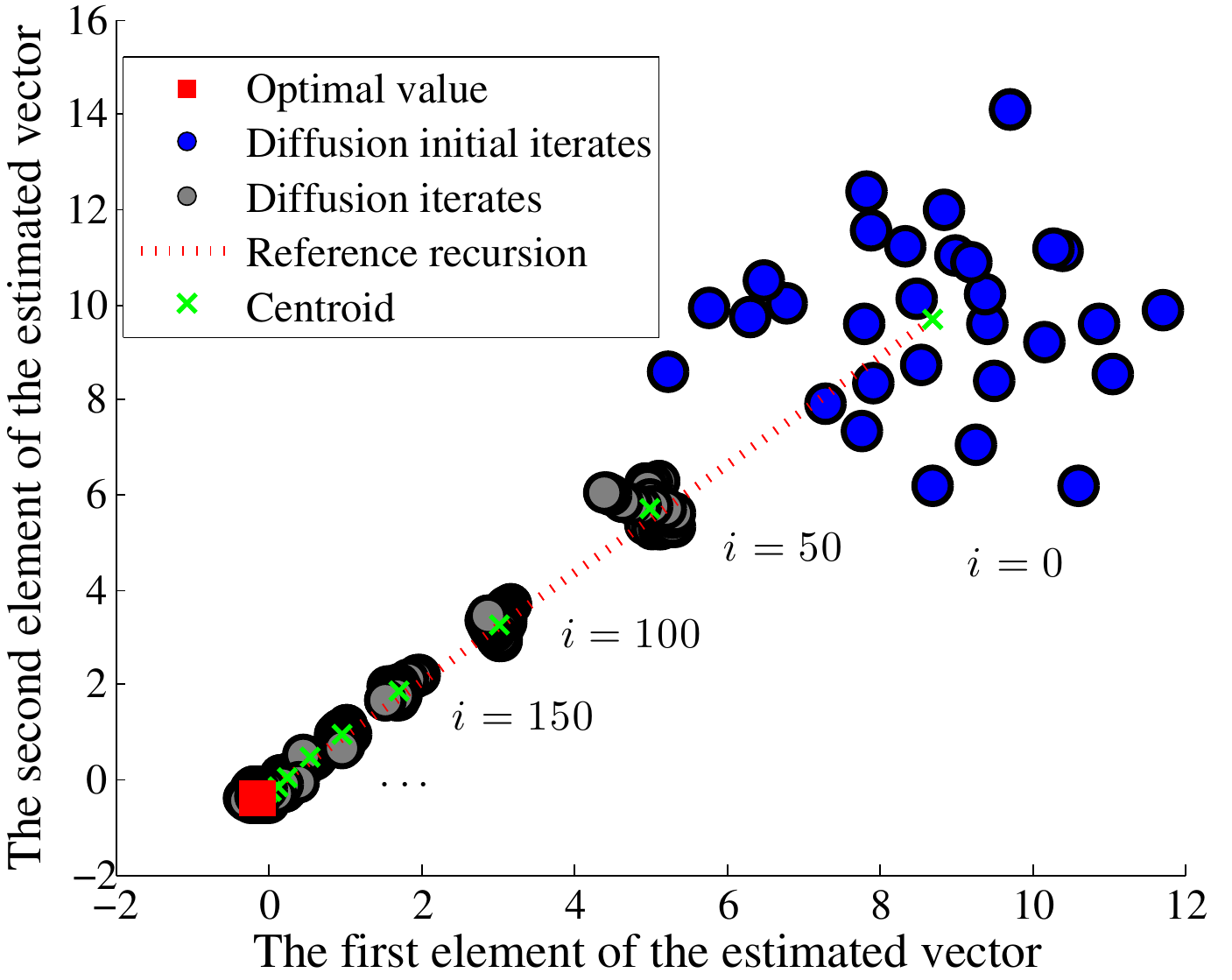}
		}
		\hfill
		\subfigure[]{
			\includegraphics[width=0.48\textwidth]{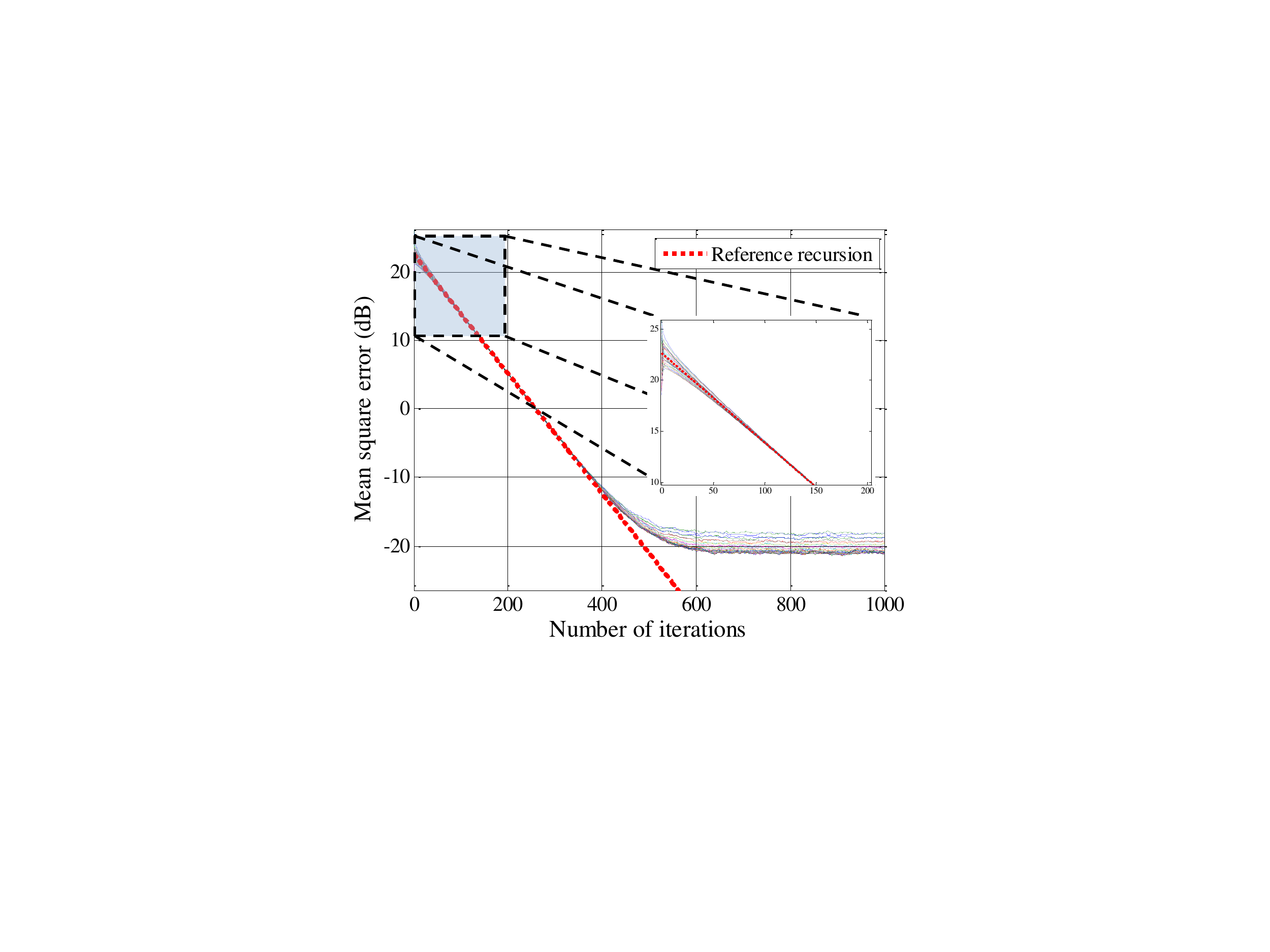}
		}
	}
	\caption{The evolution and learning curves of various quantities in a diffusion LMS adaptive network (best viewed in color), where $M=2$, and the regressors are spatially and temporally white, and isotropic across agents. (a) The evolution of the iterates $\{\w_{k,i}\}$ at all agents, the centroid $\w_{c,i}$, and the reference recursion $\bar{w}_{c,i}$ on the two-dimensional solution space; the horizontal axis and vertical axis are the first and second elements of $\w_{k,i}$, respectively. The clusters of $\{\w_{k,i}\}$ are plotted every $50$ iterations. (b) The MSE learning curves, averaged over $1000$ trials, for the iterates $\{\w_{k,i} \}$ at all agents, and the reference recursion $\bar{w}_{c,i}$. The zoom-in region shows the learning curves for different agents, which quick shrink together in Phase I.}
	\label{Fig:LearningCurveTrajectory_Simu}
\vspace{-1\baselineskip}
\end{figure*}

The result established in Theorem \ref{Thm:LearnBehav_wki} is significant because it allows us to examine the learning behavior of $\Expt \|\tilde{\w}_{k,i} \|^2$. 
{
First, note that the bound \eqref{Equ:Thm:LearnBehav_wki:Bound} established in Theorem \ref{Thm:LearnBehav_wki} is non-asymptotic; that is, it holds for any $i \ge 0$. Let $e_1(i)$, $e_2(i)$, $e_3(i)$ and $e_4$ denote the four terms in \eqref{Equ:Thm:LearnBehav_wki:Bound}:
	\begin{align}
		e_1(i)			&\defeq
							2 \| u_{L,k} \otimes I_M \|^2
							\cdot	
							\one^T\Gamma_e^i \mW_{e,0}
		\label{Equ:e1_def}
							\\
		e_2(i)			&\defeq
							2
							\|\tilde{w}_{c,0}\|
							\cdot
							\| u_{L,k} \otimes I_M \|
							\cdot
							\sqrt{
									\one^T\Gamma_e^i \mW_{e,0}
							}
		\label{Equ:e2_def}
							\\
		e_3(i)			&\defeq
							\gamma_c^i
							\cdot
							O(\mu_{\max}^{\frac{1}{2}})
		\label{Equ:e3_def}
							\\
		e_4				&\defeq
							O(\mu_{\max})
		\label{Equ:e4_def}
	\end{align}
Then, inequality \eqref{Equ:Thm:LearnBehav_wki:Bound} can be rewritten as
	\begin{align}
		\|& \tilde{w}_{c,i} \|^2 \!-\! [e_1(i) + e_2(i) + e_3(i) + e_4]	
		\nn\\
		&\le	
		\Expt \| \tilde{\w}_{k,i}\|^2 	
		\le 	
		\| \tilde{w}_{c,i} \|^2 \!+\! [e_1(i) \!+\! e_2(i) \!+\! e_3(i) \!+\! e_4]
		\label{Equ:MSE_Sandwich}
	\end{align}
for all $i \ge 0$. That is, the learning curve of $\Expt \| \tilde{\w}_{k,i}\|^2 $ is a perturbed version of the learning curve of reference recursion $\|\tilde{w}_{c,i}\|^2$, where the perturbation consists of four parts: $e_1(i)$, $e_2(i)$, $e_3(i)$ and $e_4$. We now examine the non-asymptotic convergence rates of these four perturbation terms relative to that of the reference term $\|\tilde{w}_{c,i}\|^2$ to show how the learning behavior of $\Expt \| \tilde{\w}_{k,i}\|^2 $ can be approximated by $\|\tilde{w}_{c,i}\|^2$. From their definitions in \eqref{Equ:e1_def}--\eqref{Equ:e2_def}, we note that $e_1(i)$ and $e_2(i)$ converge to zero at the \emph{non-asymptotic} rates of $\rho(\Gamma_e)$ and $[\rho(\Gamma_e)]^{\frac{1}{2}}$ over $i \ge 0$. According to \eqref{Equ:FirstOrderAnal:Gamma_e_def},  we have $\rho(\Gamma_e) = |\lambda_2(A)|$, the magnitude of the second largest eigenvalue of $A$. Let $r_{e_1}$ and $r_{e_2}$ denote the convergence rates of $e_1(i)$ and $e_2(i)$. We have
	\begin{align}
		r_{e_1}	=	|\lambda_2(A)|, \quad r_{e_2} 	=	|\lambda_2(A)|^{\frac{1}{2}}
		\label{Equ:Interpretation:r_e1_r_e2}
	\end{align}
By Assumption \ref{Assumption:Network} and \eqref{Equ:Modeling:lambda_A_order}, the value of $|\lambda_2(A)|$ is strictly less than one, and is determined by the network connectivity and the design of the combination matrix; it is also independent of the step-size parameter $\mu_{\max}$. For this reason, the terms $e_1(i)$ and $e_2(i)$ converge to zero at rates that are determined by the network and are independent of $\mu_{\max}$. Furthermore, the term $e_3(i)$ is always small, i.e., $O(\mu_{\max}^{\frac{1}{2}})$, for all $i \ge 0$, and converges to zero as $i \rightarrow \infty$. The last term $e_4$ is also small, namely, $O(\mu_{\max})$. On the other hand, as revealed by Theorem \ref{Thm:ConvergenceRefRec:DeterministcCent} and \eqref{Equ:DistProc:r_RefRec_NonAsympUBLB}, the non-asymptotic convergence rate of the term $\|\tilde{w}_{c,i}\|^2$ in \eqref{Equ:MSE_Sandwich} is bounded by
	\begin{align}	
		1-2\mu_{\max} \|p\|_1 \lambda_U	
		\le
		r_{\mathrm{Ref}}
		\le
		\gamma_c^2
		=
		1-2\mu_{\max} \lambda_L + o(\mu_{\max})
		\label{Equ:Interpretation:r_Ref_bound}
	\end{align}
Therefore, as long as $\mu_{\max}$ is small enough so that\footnote{$|\lambda_2(A)| < |\lambda_2(A)|^{\frac{1}{2}}$ always holds since $|\lambda_2(A)|$ is strictly less than one.}
	\begin{align}
		&|\lambda_2(A)| < |\lambda_2(A)|^{\frac{1}{2}} < 1 - 2\mu_{\max} \|p\|_1 \lambda_U			
	\end{align}
which is equivalent to requiring
	\begin{align}
		\mu_{\max}	<	\frac{1-|\lambda_2(A)|^{\frac{1}{2}}}{2 \|p\|_1 \lambda_U}
		\label{Equ:Interpretation:FasterTransient1_stepsizeCondition}
	\end{align}
we have the following relation regarding the non-asymptotic rates of $e_1(i)$, $e_2(i)$ and $\|\tilde{w}_{c,i}\|^2$:
	\begin{align}
		r_{e_1}	<	r_{e_2} 	< 	r_{\mathrm{Ref}}
	\end{align}
This means that, for sufficiently small $\mu_{\max}$ satisfying \eqref{Equ:Interpretation:FasterTransient1_stepsizeCondition}, $e_1(i)$ and $e_2(i)$ converge faster than $\|\tilde{w}_{c,i}\|^2$. For this reason, the perturbation terms $e_1(i)$ and $e_2(i)$ in \eqref{Equ:MSE_Sandwich} die out earlier than $\|\tilde{w}_{c,i}\|^2$. When they are negligible compared to $\|\tilde{w}_{c,i}\|^2$, we reach the end of Transient Phase I. Then, in Transient Phase II, we have
	\begin{align}
		\| \tilde{w}_{c,i} \|^2 \!-\! [e_3(i) + e_4]	
		\le	
		\Expt \| \tilde{\w}_{k,i}\|^2 	
		\le 	
		\| \tilde{w}_{c,i} \|^2 + [e_3(i) + e_4]
		\label{Equ:MSE_perturbation_TransientPhaseII}
	\end{align}
By \eqref{Equ:e3_def} and \eqref{Equ:e4_def}, the above inequality 
\eqref{Equ:MSE_perturbation_TransientPhaseII} is equivalent to the following relation:
	\begin{align}
		\Expt \| \tilde{\w}_{k,i}\|^2 	
				=		
						\|\tilde{w}_{c,i}\|^2 + O(\mu_{\max}^{\frac{1}{2}}) \cdot \gamma_c^i
						+
						O(\mu_{\max})
		\label{Equ:MSStability:MSE_wki_PhaseII}
	\end{align}
This means that $\E\|\tilde{\w}_{k,i}\|^2$ is close to $\|\tilde{w}_{c,i}\|^2$ in Transient Phase II, and the convergence rate of $\Expt \| \tilde{\w}_{k,i}\|^2$ is the same as that of $\|\tilde{w}_{c,i}\|^2$. Furthermore, in the later stage of Transient Phase II, the convergence rate of $\E\|\tilde{\w}_{k,i}\|^2$ would be close to the asymptotic rate of $\|\tilde{w}_{c,i}\|^2$ given by \eqref{Equ:DistProc:r_RefRec}. Afterwards, as $i \rightarrow \infty$, both $\|\tilde{w}_{c,i}\|^2$ and $O(\mu_{\max}^{\frac{1}{2}}) \cdot \gamma_c^i$ converge to zero and the only term remaining will be $e_4= O(\mu_{\max})$, which contributes to the steady-state MSE. More specifically, taking the $\limsup$ of both sides of \eqref{Equ:Thm:LearnBehav_wki:Bound} leads to
	\begin{align}
		\limsup_{i \rightarrow \infty}
		\Expt\|\tilde{\w}_{k,i}\|^2
					&=		\limsup_{i \rightarrow \infty}
							\big|
								\Expt \| \tilde{\w}_{k,i}\|^2
								-
								\| \tilde{w}_{c,i} \|^2
							\big|
							\nn\\
					&\le 	O(\mu_{\max})
		\label{Equ:MSStability:limsup_MSE_UB}
	\end{align}
We will go a step further and evaluate this steady-state MSE in closed-form for small step-sizes in Part II \cite{chen2013learningPart2}. Therefore, $\bw_{k,i}$ converges to $w^o$ with a small steady-state MSE that is on the order of $O(\mu_{\max})$. And the steady-state MSE can be made arbitrarily small for small step-sizes.

}

In view of this discussion, we now recognize that the results established in Theorems \ref{Thm:LimitPoint}--\ref{Thm:NonAsymptoticBound} reveal the evolution of the three components, $\tilde{w}_{c,i}$, $\check{\w}_{c,i}$ and $\w_{e,i}$ in \eqref{Equ:DistProc:wki_tilde_decomposition_final} during three distinct phases of the learning process. From \eqref{Equ:DistProc:EPwcicheck_bound}, the centroid $\bw_{c,i}$ of the distributed algorithm \eqref{Equ:ProblemFormulation:DisAlg_Comb1}--\eqref{Equ:ProblemFormulation:DisAlg_Comb2} stays close to $\bar{w}_{c,i}$ over the entire time for sufficiently small step-sizes since the mean-square error $\Expt\| \w_{c,i} - \bar{w}_{c,i} \|^2 = \Expt P[\check{\w}_{c,i}]$ is always of the order of $O(\mu_{\max})$. However, $\mW_{e,0} = \Expt P[\bw_{e,i}]$ in \eqref{Equ:DistProc:EPwei_bound} is not necessarily small at the beginning. This is because, as we pointed out in \eqref{Equ:DistProc:w_ki_decomp} and Fig. \ref{Fig:Fig_NetworkTransformation}, $\bw_{e,i}$ characterizes the deviation of the agents from their centroid. If the agents are initialized at different values, then $\E P[\bw_{e,0}] \neq 0$, and it takes some time for $\E P[\bw_{e,i}]$ to decay to a small value of $O(\mu_{\max}^2)$. By \eqref{Equ:DistProc:EPwei_bound}, the rate at which $\E P[\bw_{e,i}]$ decays is $\rho(\Gamma_e) = |\lambda_2(A)|$. On the other hand, recall from Theorems  \ref{Thm:ConvergenceRefRec:DeterministcCent}--\ref{Thm:ConvergenceRateRefRec} that the error of the reference recursion, $\tilde{w}_{c,i}$ converges at a rate between $1-2\mu_{\max} \|p\|_1 \lambda_U$ and $r_c^2 = 1-2\mu_{\max}\lambda_L + o(\lambda_{\max})$ at beginning and then $[ \rho( I_M - \mu_{\max} H_c )]^2$ later on, which is slower than the convergence rate of $\Expt P[\w_{e,i}]$ for small step-size $\mu_{\max}$. Now, returning to relation \eqref{Equ:DistProc:wki_tilde_decomposition_final}:
	\begin{align}
		\tilde{\w}_{k,i}		&=
							\tilde{w}_{c,i} - \check{\w}_{c,i} - (u_{L,k} \otimes I_M) \w_{e,i}
		\label{Equ:DistProc:wki_tilde_decomposition_final2}
	\end{align}
this means that during the initial stage of adaptation, the third term in \eqref{Equ:DistProc:wki_tilde_decomposition_final2} decays to $O(\mu_{\max}^2)$ at a faster rate than the first term, although $\tilde{w}_{c,i}$ will eventually converge to zero. Recalling from \eqref{Equ:DistProc:w_ki_decomp} and Fig. \ref{Fig:Fig_NetworkTransformation} that $\w_{e,i}$ characterizes the deviation of the agents from their centroid, the decay of $\w_{e,i}$ implies that the agents are coordinating with each other
so that their estimates $\bw_{k,i}$ are close to the same $\w_{c,i}$ --- we call this stage Transient Phase I.  Moreover, as we just pointed out, the term $\E P[\check{\bw}_{c,i}]$ is $O(\mu_{\max})$ over the entire time domain so that the second term in \eqref{Equ:DistProc:wki_tilde_decomposition_final2} is always small. This also means that the centroid of the cluster in Fig. \ref{Fig:Fig_NetworkTransformation}, i.e., $\w_{c,i}$, is always close to the reference recursion $\bar{w}_{c,i}$ since $\check{\w}_{c,i} = \w_{c,i} - \bar{w}_{c,i}$ is always small. Now that $\Expt \|\check{\w}_{c,i}\|^2$ is $O(\mu_{\max})$ and $\Expt P[\w_{e,i}]$ is $O(\mu_{\max}^2)$, the error $\tilde{\w}_{k,i}$ at each agent $k$ is mainly dominated by the first term, $\tilde{w}_{c,i}$, in \eqref{Equ:DistProc:wki_tilde_decomposition_final2}, and the estimates $\{\bw_{k,i}\}$ at different agents converge together at the same rate as the reference recursion, given by \eqref{Equ:DistProc:r_RefRec}, to steady-state --- we call this stage Transient Phase II. Furthermore, if $\mW_{e,0}=0$, i.e., the iterates $\w_{k,i}$ are initialized at the same value (e.g., zero vector), then \eqref{Equ:DistProc:EPwei_bound} shows that $\Expt P[ \w_{e,i} ]$ is $O(\mu_{\max}^2)$ over the entire time domain so that the learning dynamics start at Transient Phase II directly. Finally, all agents reach the third phase, steady-state, where $\tilde{w}_{c,i} \rightarrow 0$ and $\tilde{\w}_{k,i}$ is dominated by the second and third terms in \eqref{Equ:DistProc:wki_tilde_decomposition_final2} so that $\Expt \|\tilde{\w}_{k,i}\|^2$ becomes $O(\mu_{\max})$. We summarize the above results in Table \ref{Tab:BehavErrorPhases} and illustrate the evolution of the quantities in the simulated example in Fig. \ref{Fig:LearningCurveTrajectory_Simu}. We observe from Fig. \ref{Fig:LearningCurveTrajectory_Simu} that the radius of the cluster shrinks quickly at the early stage of the transient phase, and then converges towards the optimal solution.

\section{Conclusion}
In this work, we studied the learning behavior of adaptive networks under fairly general conditions. We showed that, in the constant and small step-size regime, a typical learning curve of each agent exhibits three phases: Transient Phase I, Transient Phase II, and Steady-state Phase. A key observation is that, the second and third phases  approach the performance of a centralized strategy. Furthermore,  we showed that the right eigenvector of the combination matrix corresponding to the eigenvalue at one influences the limit point, the convergence rate, and the steady-state mean-square-error (MSE) performance of the distributed optimization and learning strategies in a critical way. Analytical expressions that illustrate these effects were derived. Various implications were discussed and illustrative examples were also considered.

\appendices

\section{Properties of the Energy Operators}
\label{Appendix:PropertyEnergyOperatorStatement}

In this appendix, we state lemmas on properties of the operators $P_{M_0}[\cdot]$ and $\bP_{M_0}[\cdot]$. We begin with some basic properties.

\begin{lemma}[Basic properties]
        \label{Lemma:BasicPropertiesOperator}
        Consider $N\times 1$ block vectors
        $x=\col\{x_1,\ldots,x_N \}$
        and $y=\col\{y_1,\ldots,y_N\}$
        with $M\times 1$ entries $\{x_k,y_k\}$. Consider 
        also the $K\times N$ block matrix $X$  with blocks of size $M\times M$.
        Then, the operators $P[\cdot]$ and $\bP[\cdot]$ 
        satisfy the following properties:
            \begin{enumerate}
                \item
                    {\bf (\emph{Nonnegativity}):} $P[x] \succeq 0$, 
                    $\bP[X] \succeq 0$.
                \item
                    {\bf (\emph{Scaling}):} 
                    For any scalar $a \in \mb{C}$, we have
                    $P[ax]=|a|^2 P[x]$ and $\bP[aX]=|a| \cdot \bP[X]$.
                \item
                    {\bf (\emph{Convexity}):} 
                    suppose $x^{(1)},\ldots,x^{(K)}$ 
                    are $N\times 1$ block vectors formed in the same manner as 
                    $x$, $X^{(1)},\ldots,X^{(K)}$ are $K \times N$ block
                    matrices formed in the same manner as $X$, 
                    and let $a_1,\ldots,a_K$ be non-negative real scalars 
                    that add up to one. Then,
                        \begin{align}
                            P[a_1 x^{(1)}&+\cdots+ a_K x^{(K)}] 
                            					\nn\\
                            		&\preceq 		a_1 P[x^{(1)}] 
                            					+ \cdots + a_K P[x^{(K)}]
                            \label{Equ:Properties:PX_CvxComb}
                            					\\
                            \bP[a_1 X^{(1)}&+\cdots+ a_K X^{(K)}] 
                            					\nn\\
                            		&\preceq 		a_1 \bP[X^{(1)}] 
                            					+ \cdots + a_K \bP[X^{(K)}]	
                            	\label{Equ:Properties:PX_bar_CvxComb}		
                        \end{align}
                \item
                    {\bf (\emph{Additivity}):} Suppose 
                    $\bm{x}=\col\{\bm{x}_1,\ldots,\bm{x}_N\}$ and
                    $\bm{y}=\col\{\bm{y}_1,\ldots,\bm{y}_N\}$ are 
                    $N\times 1$ block random vectors that satisfy
                    $\E \bm{x}_k^* \bm{y}_k=0$ for $k=1,\ldots,N$,
                    where $*$ denotes complex conjugate transposition. 
                    Then,
                        \begin{align}
                            \label{Equ:Properties:PX_Expt}
                            \E P[\bm{x}+\bm{y}] =   \E P[\bm{x}] + \E P[\bm{y}]
                        \end{align}
                \item
                		{\bf (\emph{Triangular inequality}):}
                		Suppose $X$ and $Y$ are two $K \times N$ block matrices
                		of same block size $M$. Then,
                			\begin{align}
                				\label{Equ:Properties:PX_bar_TriIneq}
                				\bP[X + Y]	\preceq	\bP[X] + \bP[Y]
                			\end{align}
                \item
                		{\bf (\emph{Submultiplicity}):}
                		Suppose $X$ and $Z$ are $K \times N$ and $N\times L$ 
                		block matrices of the same block size $M$, respectively. Then,
                			\begin{align}
                				\label{Equ:Properties:PX_bar_SubMult}
                				\bP[XZ]		\preceq	\bP[X]  \bP[Z]
                			\end{align}    
                \item
                		{\bf (\emph{Kronecker structure}):}
                		Suppose $X \in \mathbb{C}^{K \times N}$, 
                		$a \in \mathbb{C}^N$ and $b \in \mathbb{C}^M$. Then,
                			\begin{align}
                				\label{Equ:Properties:PX_bar_Kron}
                				\bP[X \otimes I_M]	&=	\bP_1[X]		\\
                				\label{Equ:Properties:PX_Kron}
                				P[a \otimes b]		&=	\|b\|^2 \cdot P_1[a]
                			\end{align}
                   	where by definition, $\bP_1[\cdot]$ and $P_1[\cdot]$
                   	denote the operators that work on the scalar entries of
                   	their arguments. When
                   	$X$ consists of nonnegative entries, relation
                   	\eqref{Equ:Properties:PX_bar_Kron} becomes
                   		\begin{align}
                   			\label{Equ:Properties:PX_bar_Kron1}
                				\bP[X \otimes I_M]	&=	X
                   		\end{align}
                \item
                    {\bf (\emph{Relation to norms}):}
                    The $\infty-$norm of $P[x]$ is the squared block 
                    maximum norm of $x$:
                        \begin{align}
                        	\label{Equ:Properties:PX_BMNorm}
                            \|P[x]\|_{\infty}   
                            		=		\|x\|_{b,\infty}^2
                                	\defeq  \big(\max_{1 \le k \le N}\|x_k\|\big)^2
                        \end{align}
                    Moreover, the sum of the entries in $P[x]$ is the squared
                    Euclidean norm of $x$:
                   		\begin{align}
                   			\one^T P[x]	=	\|x\|^2
                   						=	\sum_{k=1}^N \|x_{k}\|^2
                   			\label{Equ:Properties:PX_EuclNorm}
                   		\end{align}
                \item
                    {\bf (\emph{Inequality preservation}):}
                    Suppose
                    vectors $x$, $y$ and matrices $F$, $G$  
                    have nonnegative entries, then
                    $x \preceq y$ implies $Fx \preceq  Fy$, and 
                    $F \preceq G$ implies $Fx \preceq Gx$.
               	\item
               		{\bf (\emph{Upper bounds}):}
               		It holds that
               			\begin{align}
               				\label{Equ:Properties:UB_P_bar}
               				\bP[X]	&\preceq		\|\bP[X]\|_{1} 
               									\cdot \mathds{1}\mathds{1}^T
               									\\
               				\bP[X]	&\preceq		\|\bP[X]\|_{\infty}
               									\cdot \one\one^T
               				\label{Equ:Properties:UB_P_bar_inf}
               			\end{align}
               		where $\|\cdot\|_{\infty}$ denotes the 
                    	$\infty-$induced norm of a matrix (maximum
                    	absolute row sum). 
            \end{enumerate}
    \end{lemma}
    \begin{IEEEproof}
    		See Appendix \ref{Appendix:Proof_BasicProperties}.
    \end{IEEEproof}
    
More importantly, the following variance relations hold for the energy and norm operators. 
These relations show how error variances propagate after 
a certain operator is applied to a random vector.
    
    \begin{lemma}[Variance relations]
    		\label{Lemma:VarianceRelations}
		Consider $N\times 1$ block vectors
        $x=\col\{x_1,\ldots,x_N \}$
        and $y=\col\{y_1,\ldots,y_N\}$
        with $M\times 1$ entries $\{x_k,y_k\}$. 
        The following variance relations are satisfied by the energy vector operator
        $P[\cdot]$:
    		
    		\begin{enumerate}    			
                \item
                    {\bf (\emph{Linear transformation}):}
                    Given a $K \times N$ 
                    block matrix $Q$ with the size of 
                    each block being $M \times M$,
                    $Qx$ defines a linear operator on $x$
                    and its energy satisfies
                    		\begin{align}
                    			P[Q x]	&\preceq	\| \bP[Q]\|_{\infty} \cdot
                    							\bP[Q] \; P[x]
                    			\label{Equ:VarPropt:Linear}
                    							\\
                    					&\preceq	\| \bP[Q]\|_{\infty}^2 \cdot
                    							\one\one^T \cdot P[x]
                    			\label{Equ:VarPropt:Linear_ub}
                    		\end{align}
                    As a special case,
                    	for a left-stochastic $N \times N$ matrix
                    	$A$, we have
                        \begin{align}
                            \label{Equ:VarPropt:A}
                            &P[(A^T \otimes I_M) x]   \preceq     A^T P[x]
                        \end{align}
           		\item
           			{\bf (\emph{Update operation}):}
           			The global update vector defined by  \eqref{Equ:Def:s} satisfies the following variance relation:
           				\begin{align}
           					\label{Equ:VarPropt:s}
           					P[s(x)-s(y)]		\preceq	\lambda_U^2 P[x-y]
           				\end{align}
           		\item
           			{\bf (\emph{Centralized operation}):}
           			The centralized operator $T_c(x)$ defined by \eqref{Equ:Def:Tc2} satisfies the following
           			variance relations:
           				\begin{align}
           					\label{Equ:VarPropt:Tc}
           					P[T_c(x)-T_c(y)]	&\preceq	\gamma_c^2 \cdot P[x-y]	\\
           					\label{Equ:VarPropt:Tc_LB}
           					P[T_c(x)-T_c(y)]	&\succeq	(1-2\mu_{\max} \|p\|_1 \lambda_U) \cdot P[x-y]
           				\end{align}
           			where	
           				\begin{align}
           					\label{Equ:VarPropt:gamma_c}
           					\gamma_c		\defeq	1 - \mu_{\max}\lambda_L 
           								+ \frac{1}{2}\mu_{\max}^2\|p\|_1^2 \lambda_U^2
           				\end{align}
           			Moreover, it follows from \eqref{Equ:Remark:lambdaU_lambda_L_relation} that
           				\begin{align}
           					\gamma_c		&\ge		1 - \mu_{\max}\lambda_L 
           										+ \frac{1}{2}\mu_{\max}^2\lambda_L^2
           										\nonumber\\
           								&= 	(1-\frac{1}{2}\mu_{\max}\lambda_L)^2 
           									+ \frac{1}{4}\mu_{\max}^2\lambda_L^2 > 0
           					\label{Equ:VarPropt:gamma_c_positive}
           				\end{align}
           		\item
           			{\bf (\emph{Stable Jordan operation}):}
           			Suppose $D_L$ is an $L \times L$
           			Jordan matrix of the following block form:
           				\begin{align}
           					\label{Equ:Propert:DL_def}
           					D_L			&\defeq
           								\diag\{D_{L,2},\ldots,D_{L,n_0}\}
           				\end{align}
           			where the $n$th $L_n \times L_n$ 
           			Jordan block is defined as (note that $L=L_2+\cdots+L_{n_0}$)
           				\begin{align}
           					\label{Equ:Propert:DLn_def}
           					D_{L,n}		&\defeq	
           								\begin{bmatrix}
           									d_n & 1 & &  	\\
           									    & \ddots & \ddots &	\\
           									    & & \ddots & 1	\\
           									    & & & d_n
           								\end{bmatrix}
           				\end{align}
           			We further assume $D_L$ to be stable with
           			$0 \le |d_{n_0}| \le \cdots \le |d_{2}| < 1$.
           			Then, for any $L \times 1$ vectors $x'$  and $y'$,
           			we have 
           				\begin{align}
           					\label{Equ:Propert:StableJordanOperator}
           					P_1[ D_L x' + y']	\preceq 
           										\Gamma_e \cdot P_1[x'] 
           										+ 
           										\frac{2}{1-|d_2|}
           										\cdot
           										P_1[y']
           				\end{align}
           			where $\Gamma_e$ is the $L \times L$ matrix defined as
           				\begin{align}
           					\label{Equ:Propert:Gamma_e}
           					\Gamma_e		&\defeq	\begin{bmatrix}
		           									|d_2| & \frac{2}{1-|d_2|} & &
		           														  	\\
		           									    & \ddots & \ddots &	\\
		           									    & & \ddots & 
		           									    	\frac{2}{1-|d_2|}	\\
		           									    & & & |d_2|
		           								\end{bmatrix}
           				\end{align}
           		\item
           			{\bf (\emph{Stable Kronecker Jordan operator}):}
           			Suppose $\mc{D}_L = D_L \otimes I_M$, where $D_L$ is 
           			the $L \times L$ Jordan matrix defined
           			in \eqref{Equ:Propert:DL_def}--%
           			\eqref{Equ:Propert:DLn_def}. Then, for any
           			$LM \times 1$ vectors $x_e$ and $y_e$, we have
           				\begin{align}
           					\label{Equ:Propert:StableKronJordanOperator}
           					P[\mc{D}_L x_e + y_e]
           						\preceq	\Gamma_e \cdot P[x_e] + 
           						\frac{2}{1-|d_2|} \cdot P[y_e]
           				\end{align}
    		\end{enumerate}
    \end{lemma}
    \begin{IEEEproof}
    		See Appendix \ref{Appendix:Proof_VarianceRelations}.
    \end{IEEEproof}

\section{Proof of Lemma \ref{Lemma:EquivCond_updateVec}}

\label{Appendix:Proof_Lemma_EquivCondUpdateVec}

First, we establish that conditions \eqref{Equ:Lemma:EquivCondUpdateVec:LipschitzUpdate_HessianUB} and \eqref{Equ:Lemma:EquivCondUpdateVec:StrongMono_HessianLB} imply \eqref{Equ:Assumption:Lipschitz} and \eqref{Equ:Assumption:StrongMonotone}, respectively. Using the mean-value theorem\cite[p.6]{poliak1987introduction}, we have for any $x,y \in \mS$:
			\begin{align}
				\left\|
					s_k(x) - s_k(y)
				\right\|
						&=		\left\|
									\int_{0}^{1}
									\nabla_{w^T} 
									s_k(y + t(x-y)) dt
									\cdot
									(x-y)
								\right\|
								\nn\\
						&\le 	\int_{0}^{1}
								\left\|
								\nabla_{w^T} 
								s_k(y + t(x-y))
								\right\| 
								dt
								\cdot
								\|x-y\|
								\nn\\
						&\le 	\lambda_U \cdot \|x-y\|
			\end{align}
where we used the fact that $y+t(x-y) = tx+(1-t)y \in \mS$ given $x,y \in \mS$ and $0 \le t \le 1$. Likewise, we have
			\begin{align}
				&(x-y)^T \cdot
				\sum_{k=1}^N
				p_k 
				[ s_k(x)-s_k(y) ]
								\nn\\
						&=		(x-y)^T \cdot
								\sum_{k=1}^N
								p_k 
								\int_{0}^{1} \nabla_{w^T} s_k(y + t(x-y)) dt
								\cdot
								(x-y)
								\nn\\
						&=		(x-y)^T \cdot								
								\int_{0}^{1} 
								\sum_{k=1}^N
								p_k 
								\nabla_{w^T} s_k(y + t(x-y)) dt
								\cdot
								(x-y)
								\nn\\
						&\overset{\eqref{Equ:Lemma:EquivCondUpdateVec:Hc_def}}{=}
								(x-y)^T \cdot	
								H_c(y+t(x-y))
								\cdot
								(x-y)
								\nn\\
						&=		(x \! - \! y)^T 
								\!\cdot\!
								\frac{ H_c(y+t(x-y)) + H_c^T(y+t(x-y)) }{2}
								\!\cdot\!
								(x \!- \! y)
								\nn\\
						&\ge 	\lambda_L \cdot \|x-y\|^2
			\end{align}
			
		Next, we establish the reverse direction that conditions \eqref{Equ:Assumption:Lipschitz} and \eqref{Equ:Assumption:StrongMonotone} imply  \eqref{Equ:Lemma:EquivCondUpdateVec:LipschitzUpdate_HessianUB} and \eqref{Equ:Lemma:EquivCondUpdateVec:StrongMono_HessianLB}. Choosing $x=w+ t \cdot \delta w$ and $y = w$ in \eqref{Equ:Assumption:Lipschitz} for any $\delta w \neq 0$ and any small positive $t$, we get
			\begin{align}
				\| &s_k( w + t \cdot \delta w ) - s_k(w) \|
						\le 	t \cdot \lambda_U \cdot \| \delta w \|
								\nn\\
				&\Rightarrow \quad
				\lim_{t \rightarrow 0^{+}}
				\left\|
					\frac{s_k( w + t \cdot \delta w ) - s_k(w)}{t}
				\right\|
						\le 	\lambda_U \cdot \| \delta w \|
								\nn\\
				&\Rightarrow \quad
				\left\|
					\lim_{t \rightarrow 0^{+}}
					\frac{s_k( w + t \cdot \delta w ) - s_k(w)}{t}
				\right\|
						\le 	\lambda_U \cdot \| \delta w \|
								\nn\\
				&\Rightarrow \quad
				\left\|
					\nabla_{w^T} s_k(w) \delta w
				\right\|
						\le 	\lambda_U \cdot \| \delta w \|
								\nn\\
				&\Rightarrow \quad
				\| \nabla_{w^T} s_k(w) \|
						\defeq	\sup_{\delta w \neq 0}
								\frac{\left\|\nabla_{w^T} s_k(w) \delta w\right\|}{\| \delta w \|}
						\le 	\lambda_U
			\end{align}
		Likewise, choosing $x=w+ t \cdot \delta w$ and $y = w$ in \eqref{Equ:Assumption:StrongMonotone} for any $\delta w \neq 0$ and any small positive $t$, we obtain
			\begin{align}
				t \cdot \delta w^T & \cdot
				\sum_{k=1}^N
				p_k 
				[s_k( w + t\cdot \delta w) - s_k(w)]
						\ge 	t^2 \cdot \lambda_L \cdot \|\delta w\|^2
								\nn\\
				&\Rightarrow \quad
				\delta w^T \cdot
				\sum_{k=1}^N
				p_k 
				\left(
					\lim_{ t \rightarrow 0^{+}}
					\frac{s_k( w + t\cdot \delta w) - s_k(w)}{t}
				\right)
								\nn\\
				&\qquad\qquad
							\ge 	\lambda_L \cdot \|\delta w\|^2
								\nn\\
				&\Rightarrow \quad
				\delta w^T \cdot
				\sum_{k=1}^N
				p_k 
				\nabla_{w^T} s_k(w)
				\cdot \delta w
						\ge 	\lambda_L \cdot \|\delta w\|^2
								\nn\\
				&\Rightarrow \quad
				\delta w^T
				H_c(w)
				\delta w
						\ge 	\lambda_L \cdot \|\delta w\|^2
								\nn\\
				&\Rightarrow \quad
				\delta w^T \frac{H_c(w) + H_c^T(w) }{2} \delta w
						\ge 	\lambda_L \cdot \| \delta w \|^2
								\nn\\
				&\Rightarrow \quad
				\frac{H_c(w) + H_c^T(w) }{2}
						\ge 	\lambda_L \cdot I_M
			\end{align}

\section{Proof of Lemma \ref{Lemma:BasicPropertiesOperator}}
\label{Appendix:Proof_BasicProperties}
Properties 1-2 are straightforward from the definitions of $P[\cdot]$
and $\bP[\cdot]$. Property 4 was proved in
\cite{chen2013JSTSPpareto}.
We establish the remaining properties.

{\bf (\emph{Property 3: Convexity})}
The convexity of $P[\cdot]$ has already been proven in
\cite{chen2013JSTSPpareto}. We now establish the convexity of
the operator $\bP[\cdot]$. Let $X_{qn}^{(k)}$ denote the $(q,n)-$th
$M\times M$
block of the matrix $X^{(k)}$, where $q=1,\ldots,K$ and $n=1,\ldots,N$.
Then,
	\begin{align}
		\bP&\left[
			\sum_{k=1}^K a_k X^{(k)}
		\right]
				\nn\\
			&\preceq
				\begin{bmatrix}
					\sum_{k=1}^K a_k\left\| X_{11}^{(k)}\right\| & \cdots & 
						\sum_{k=1}^K a_k\left\| X_{1N}^{(k)}\right\|		\\
					\vdots	&	&	\vdots					\\
					\sum_{k=1}^K a_k\left\| X_{K1}^{(k)}\right\| & \cdots &
						\sum_{k=1}^K a_k\left\| X_{KN}^{(k)}\right\|
				\end{bmatrix}
				\nonumber\\
			&=	\sum_{k=1}^K a_k \bP[X^{(k)}]
	\end{align}

{\bf (\emph{Property 5: Triangular inequality})}
Let $X_{qn}$ and $Y_{qn}$ denote the $(q,n)-$th $M \times M$ 
blocks of the matrices
$X$ and $Y$, respectively, where $q=1,\ldots,K$ and $n=1,\ldots,N$.
Then, by the triangular inequality of the matrix norm $\|\cdot\|$, we  have
	\begin{align}
		\bP\left[
			X+Y
		\right]
			&\preceq
				\begin{bmatrix}
					\left\|X_{11}\right\| + \left\|Y_{11}\right\| & \cdots & 
						\left\|X_{1N}\right\| + 
						\left\|Y_{1N}\right\|		\\
					\vdots	&	&	\vdots					\\
					\left\|X_{K1}\right\| + 
					\left\|Y_{K1}\right\| & \cdots &
						\left\|X_{KN}\right\| + \left\|Y_{KN}\right\|
				\end{bmatrix}
				\nonumber\\
			&=	\bP[X] + \bP[Y]
	\end{align}
	
{\bf (\emph{Property 6: Submultiplicity})}
Let $X_{kn}$ and $Z_{nl}$ be the $(k,n)-$th and $(n,l)-$th 
$M \times M$ blocks of $X$ and $Z$, respectively. Then, the 
$(k,l)-$th $M \times M$ block of the matrix product $XZ$,
denoted by $[XZ]_{k,l}$, is 
	\begin{align}
		[X Z]_{k,l}	=	\sum_{n=1}^N X_{kn} Z_{nl}
	\end{align}
Therefore, the $(k,l)-$th entry of the matrix 
$\bP[XZ]$ can be bounded as
	\begin{align}
		\big[\bP[XZ]\big]_{k,l}
					&=		\left\|
								\sum_{n=1}^N X_{kn} Z_{nl}
							\right\|
					\le 	\sum_{n=1}^N
							\left\|
								X_{kn}
							\right\|
							\cdot
							\left\| 
								Z_{nl}
							\right\|						
	\end{align}
Note that $\|X_{kn}\|$ and $\|Z_{nl}\|$ are the $(k,n)-$th
and $(n,l)-$th entries of the matrices $\bP[X]$ and $\bP[Z]$, respectively.
The right-hand side of the above inequality is therefore the $(k,l)-$th
entry of the matrix product $\bP[X]\bP[Z]$. Therefore,
we obtain
	\begin{align}
		\bP[XZ]		\preceq	\bP[X]\bP[Z]
	\end{align}

{\bf (\emph{Property 7: Kronecker structure})}
For \eqref{Equ:Properties:PX_bar_Kron}, we note that
the $(k,n)-$th $M\times M$ block of $X\otimes I_M$ is $x_{kn} I_M$.
Therefore, by the definition of $\bP[\cdot]$, we have
	\begin{align}
		\bP[X \otimes I_M]	&=	\begin{bmatrix}
									|x_{11}|	& \cdots & |x_{1N}|	\\
									\vdots	&		&			\\
									|x_{K1}|	& \cdots	&	|x_{KN}|
								\end{bmatrix}
							=	\bP_1[X]
	\end{align}
In the special case when $X$ consists of nonnegative entries, 
$\bP_1[X]=X$, and we recover \eqref{Equ:Properties:PX_bar_Kron1}.
To prove \eqref{Equ:Properties:PX_Kron}, we let
$a=\col\{a_1,\ldots,a_N\}$ and $b=\col\{b_1,\ldots,b_M\}$.
Then, by the definitin of $P[\cdot]$, we have
	\begin{align}
		P[a\otimes b]	=	\col\{ |a_1|^2 \cdot \|b\|^2, \ldots, |a_N|^2 \cdot \|b\|^2 \}
						=	\|b\|^2 \cdot P_1[a]
	\end{align}

{\bf (\emph{Property 8: Relation to norms})}
Relations \eqref{Equ:Properties:PX_BMNorm} and \eqref{Equ:Properties:PX_EuclNorm}
are straightforward and follow from the definition.

{\bf (\emph{Property 9: Inequality preservation})}
The proof that $x\preceq y$ implies $Fx \preceq Fy$ can be 
found in \cite{chen2013JSTSPpareto}. We now prove
that $F \preceq G$ implies $Fx \preceq Gx$. This can
be proved by showig that $(G-F)x \succeq 0$, which is true
because all entries of $G-F$ and $x$ are nonnegative due
to $F \preceq G$ and $x \succeq 0$.

{\bf (\emph{Property 10: Upper bounds})}
By the definition of $\bP[X]$ in \eqref{Equ:Def:NormMatrixOperator}, we get
	\begin{align}
		\bP[X] 	&\preceq 	\left(\max_{l,k} \|X_{lk}\|\right)
							\cdot \one\one^T	
							\nonumber\\
				&\preceq		\max_{l} \left(\sum_{k=1}^N \|X_{lk}\|\right)
							\cdot \one \one^T
							\nonumber\\
				&=			\|\bP[X]\|_{\infty} \cdot \one \one^T
	\end{align}
Likewise, we can establish that $\bP[X] \preceq \|\bP[X]\|_{1} \cdot \one\one^T$.

\section{Proof of Lemma \ref{Lemma:VarianceRelations}}
\label{Appendix:Proof_VarianceRelations}

{\bf (\emph{Property 1: Linear transformation})}
Let $Q_{kn}$ be the $(k,n)$-th $M\times M$ block of 
$Q$. Then
	\begin{align}
		\label{Equ:Appendix:P_Qx}
		P[Qx]	
			=	\col\left\{
					\left\|
						\sum_{n=1}^N 
						Q_{1n} x_n
					\right\|^2
					,
					\ldots,
					\left\|
						\sum_{n=1}^N 
						Q_{Kn} x_n
					\right\|^2
				\right\}
	\end{align}
Using the convexity of $\|\cdot\|^2$, we have the following
bound on each $n$-th entry:
	\begin{align}
		\bigg\|&
			\sum_{n=1}^N 
			Q_{kn} x_n
		\bigg\|^2
					\nn\\
			&\overset{(a)}{=}		
					\left[\sum_{n=1}^N \|Q_{kn}\|\right]^2
					\cdot
					\left\|
						\sum_{n=1}^N
						\frac{\|Q_{kn}\|}{\sum_{l=1}^N\|Q_{kl}\|}
						\cdot
						\frac{Q_{kn}}{\|Q_{kn}\|} x_n
					\right\|^2
					\nonumber\\
			&\overset{(b)}{\le} 
					\left[\sum_{n=1}^N \|Q_{kn}\|\right]^2
					\cdot
						\sum_{n=1}^N
						\frac{\|Q_{kn}\|}{\sum_{l=1}^N\|Q_{kl}\|}
						\cdot
						\frac{\|Q_{kn}\|^2}{\|Q_{kn}\|^2} \|x_n\|^2
					\nonumber\\
			&=		\left[\sum_{n=1}^N \|Q_{kn}\|\right]
					\cdot
						\sum_{n=1}^N
						\|Q_{kn}\| \cdot \|x_n\|^2
					\nonumber\\
			&\le 	\max_{k}\left[\sum_{n=1}^N \|Q_{kn}\|\right]
					\cdot
						\sum_{n=1}^N
						\|Q_{kn}\| \cdot \|x_n\|^2
					\nonumber\\
			&=		\|\bP[Q]\|_{\infty}
					\cdot
					\sum_{n=1}^N
					\|Q_{kn}\| \cdot \|x_n\|^2
		\label{Equ:Appendix:VarPropt_LinearT}
	\end{align}
where in step (b) we applied Jensen's inequlity to $\|\cdot\|^2$. 
{ 
Note
that if some $\|Q_{kn}\|$ in step (a) is zero, we eliminate the corresponding term
from the sum and it can be verified that the final result
still holds.
}
Substituting into \eqref{Equ:Appendix:P_Qx}, we establish
\eqref{Equ:VarPropt:Linear}. The special case \eqref{Equ:VarPropt:A}
can be obtained by using $\bP[A^T \otimes I_M]=A^T$
and that $\|A^T\|_{\infty}=1$ (left-stochastic) in \eqref{Equ:VarPropt:Linear}.
Finally, the upper bound \eqref{Equ:VarPropt:Linear_ub} can be proved
by applying \eqref{Equ:Properties:UB_P_bar_inf} to $\bP[Q]$.

{\bf (\emph{Property 2: Update operation})}
By the definition of $P[\cdot]$ and the Lipschitz Assumption
\ref{Assumption:UpdateVectorLipschitz},
we have
	\begin{align}
		P&[s(x)-s(y)]			
							\nn\\
						&=	\col\{
								\|s_1(x_1) - s_1(y_1)\|^2,
								\ldots
								\|s_N(x_N) - s_N(y_N)\|^2
							\}
							\nonumber\\
						&\preceq
							\col\{
								\lambda_U^2 \cdot \|x_1 - y_1\|^2,
								\ldots
								\lambda_U^2 \cdot \|x_N - y_N\|^2
							\}
							\nonumber\\
						&=	\lambda_U^2 \cdot P[x-y]
	\end{align}

{\bf (\emph{Property 3: Centralized opertion})}
Since $T_c: \mathbb{R}^M \rightarrow \mathbb{R}^M$, the output
of $P[T_c(x)-T_c(y)]$ becomes a scalar. From the 
definition, we get 
	\begin{align}
		P&\left[
			T_c(x) - T_c(y)
		\right]
					\nn\\
			&=		\left\|
						x - y
						-
						\mu_{\max}
						\cdot
						(p^T \otimes I_M)
						\left[
							s( \mathds{1} \otimes x)
							-
							s( \mathds{1} \otimes y)
						\right]
					\right\|^2
					\nonumber\\
			&=		\big\|
						x - y
						-
						\mu_{\max}
						\cdot
						\sum_{k=1}^N p_k
						\left[
							s_k( x)
							-
							s_k( y)
						\right]
					\big\|^2
					\nonumber\\
			&=		\left\| x - y\right\|^2
					-2\mu_{\max} \cdot
					(x-y)^T
					\sum_{k=1}^N p_k
					\left[
						s_k( x)
						-
						s_k( y)
					\right]
					\nn\\
					&\quad
					+
					\mu_{\max}^2	
					\left\|
						(p^T \otimes I_M)
						\left[
							s( \mathds{1} \otimes x)
							-
							s( \mathds{1} \otimes y)
						\right]
					\right\|^2				
					\nonumber\\
			&=		\left\| x - y\right\|^2
					-2\mu_{\max} \cdot
					(x-y)^T
					\sum_{k=1}^N p_k
					\left[
						s_k( x)
						-
						s_k( y)
					\right]
					\nn\\
					&\quad
					+
					\mu_{\max}^2	
					P\left[
						(p^T \otimes I_M)
						\left[
							s( \mathds{1} \otimes x)
							-
							s( \mathds{1} \otimes y)
						\right]
					\right]	
		\label{Equ:Appendix:P_Tc_interm1}
	\end{align}
We first prove the upper bound \eqref{Equ:VarPropt:Tc} as follows:
	\begin{align}
		P&\left[
			T_c(x) - T_c(y)
		\right]
					\nn\\
			&=
					\big\|
						x - y
						-
						\mu_{\max}
						\cdot
						\sum_{k=1}^N
						p_k [ s_k(x) - s_k(y) ]
					\big\|^2
					\nn\\
			&=
					\| x - y \|^2
					-
					2 \mu_{\max} \cdot (x-y)^T \sum_{k=1}^N p_k [ s_k(x) - s_k(y) ]
					\nn\\
					&\quad
					+
					\mu_{\max}^2
					\cdot
					\big\|
						\sum_{k=1}^N
						p_k
						[ s_k(x) - s_k(y) ]
					\big\|^2
					\nn\\
			&\overset{\eqref{Equ:Assumption:StrongMonotone}}{\le}
					\| x - y \|^2
					-
					2 \mu_{\max} \cdot \lambda_L \cdot \| x-y \|^2
					\nn\\
					&\quad
					+
					\mu_{\max}^2
					\cdot
					\big\|
						\sum_{k=1}^N
						p_k
						[ s_k(x) - s_k(y) ]
					\big\|^2
					\nn\\
			&\le 
					\| x - y \|^2
					-
					2 \mu_{\max} \cdot \lambda_L \cdot \| x-y \|^2
					\nn\\
					&\quad
					+
					\mu_{\max}^2
					\cdot
					\big[
						\sum_{k=1}^N
						p_k
						\| s_k(x) - s_k(y) \|
					\big]^2
					\nn\\
			&\overset{\eqref{Equ:Assumption:Lipschitz}}{\le}
					\| x - y \|^2
					-
					2 \mu_{\max} \cdot \lambda_L \cdot \| x-y \|^2
					\nn\\
					&\quad
					+
					\mu_{\max}^2
					\cdot
					\left[
						\sum_{k=1}^N
						p_k
						\cdot
						\lambda_U \cdot \| x - y \|
					\right]^2
					\nn\\
			&=	
					\| x - y \|^2
					-
					2 \mu_{\max} \cdot \lambda_L \cdot \| x-y \|^2
					\nn\\
					&\quad
					+
					\mu_{\max}^2
					\cdot
					\| p \|_1^2
					\lambda_U^2 
					\cdot 
					\| x - y \|^2
					\nn\\
			&=		
					(1-2\mu_{\max}\lambda_L + \mu_{\max}^2 \lambda_U^2 \|p\|_1^2)
					\cdot \|x-y\|^2
					\nonumber\\
			&\le
					\left(
						1-\mu_{\max}\lambda_L + \frac{1}{2}\mu_{\max}^2 \lambda_U^2 \|p\|_1^2
					\right)^2
					\cdot \|x-y\|^2
					\label{Equ:Appendix:P_Tc_interm1_UB}
	\end{align}
where in the last step we used
the relation $(1-x) \le (1-\frac{1}{2}x)^2$.

Next, we prove the lower bound \eqref{Equ:VarPropt:Tc_LB}.
From \eqref{Equ:Appendix:P_Tc_interm1}, we notice that the
last term in \eqref{Equ:Appendix:P_Tc_interm1} is always nonnegative
so that 
	\begin{align}
		P&\left[
			T_c(x) - T_c(y)
		\right]
					\nn\\
			&\succeq	
					\left\| x - y\right\|^2
					-2\mu_{\max} \cdot
					(x-y)^T
					\sum_{k=1}^N p_k
					\left[
						s_k( x )
						-
						s_k( y )
					\right]
					\nonumber\\
			&\overset{(a)}{\succeq}
					\left\| x - y\right\|^2
					-2\mu_{\max} \cdot
					\|x-y\|
					\cdot
					\big\|
					\sum_{k=1}^N p_k
					\left[
						s_k( x )
						-
						s_k( y)
					\right]
					\big\|
					\nonumber\\
			&\succeq
					\left\| x - y\right\|^2
					-2\mu_{\max} \cdot
					\|x-y\|
					\cdot
					\sum_{k=1}^N p_k
					\left\|
						s_k( x )
						-
						s_k( y )
					\right\|
					\nonumber\\
			&\overset{(b)}{\succeq}
					\left\| x - y\right\|^2
					-2\mu_{\max} \cdot
					\|x-y\|
					\cdot
					\sum_{k=1}^N p_k
					\lambda_U \|x-y\|
					\nonumber\\
			&=		(1-2\mu_{\max} \lambda_U \|p\|_1) 
					\cdot
					\|x-y\|^2
	\end{align}
where in step (a), we used the Cauchy-Schwartz inequality 
$x^Ty \le |x^T y| \le \|x\|\cdot\|y\|$, and
in step (b) we used \eqref{Equ:Assumption:Lipschitz}.

{\bf (\emph{Property 4: Stable Jordan operator})}
First, we notice that matrix $D_{L,n}$ can be written as
	\begin{align}
		\label{Equ:Appendix:D_Ln_Decomp}
		D_{L,n}	=	d_n \cdot I_{L_n} + \Theta_{L_n}
	\end{align}
where $\Theta_{L_n}$ is an $L_n \times L_n$ strictly upper triangular matrix of the following form:
	\begin{align}
		\label{Equ:Appendix_Theta_Ln}
		\Theta_{L_n}	\defeq	\begin{bmatrix}
           						0 & 1 & &  	\\
           					    		& \ddots & \ddots &	\\
           							& & \ddots & 1	\\
           							& & & 0
           					\end{bmatrix}
	\end{align}
Define the following matrices:
	\begin{align}
		\Lambda_L	&\defeq	\diag\{d_2 I_{L_2}, \ldots, d_{n_0} I_{L_{n_0}}\}
		\label{Equ:Appendix:Lambda_L_def}
							\\
		\Theta_L'	&\defeq	\diag\{\Theta_{L_2},\ldots,\Theta_{L_{n_0}}\}
	\end{align}
Then, the original Jordan matrix $D_L$ can be expressed as
	\begin{align}
		D_L		=		\Lambda_L + \Theta_L'
		\label{Equ:Appendx:DL_expression}
	\end{align}
so that
	\begin{align}
		P_1&[D_L x' + y']	
					\nn\\
				&=	P_1\left[
						\Lambda_L x' + \Theta_L' x' + y'
					\right]
					\nonumber\\
				&=	P_1\bigg[
						|d_2| \cdot\frac{1}{|d_2|}\Lambda_L x' 
						+ 
						\frac{1-|d_2|}{2}
						\cdot
						\frac{2}{1-|d_2|}
						\Theta_L' x'
						\nn\\
						&\qquad
						+ 
						\frac{1-|d_2|}{2}
						\cdot
						\frac{2}{1-|d_2|}
						y'
					\bigg]
					\nonumber\\
				&\overset{(a)}{\preceq}
					|d_2| \cdot
					P_1\left[
						\frac{1}{|d_2|}\Lambda_L x' 
					\right]
					+ 
					\frac{1-|d_2|}{2}
					\cdot
					P_1\left[						
						\frac{2}{1-|d_2|}
						\Theta_L' x' 
					\right]
					\nn\\
					&\quad
					+ 
					\frac{1-|d_2|}{2}
					\cdot
					P_1\left[
						\frac{2}{1-|d_2|}
						y'
					\right]
					\nonumber\\
				&\overset{(b)}{=}
					\frac{1}{|d_2|} \!\cdot\!
					P_1\left[
						\Lambda_L x' 
					\right]
					\!+\!						
					\frac{2}{1-|d_2|}
					\!\cdot\!
					P_1\left[
						\Theta_L' x' 
					\right]
					\!+ \!
					\frac{2}{1-|d_2|}
					\!\cdot\!
					P_1\left[						
						y'
					\right]
					\nonumber\\
				&\overset{(c)}{\preceq}
					\frac{\|\bP_1[\Lambda_L]\|_{\infty}}{|d_2|} 					
					\cdot
					\bP_1[\Lambda_L]
					\cdot
					P_1\left[ 
						x' 
					\right]
					\nn\\
					&\quad
					+						
					\frac{2\|\bP_1[\Theta_L']\|_{\infty}}{1-|d_2|}
					\cdot
					\bP_1[\Theta_L']
					\cdot
					P_1\left[
						x' 
					\right]
					+ 
					\frac{2}{1-|d_2|}
					\cdot
					P_1\left[						
						y'
					\right]
					\nonumber\\
				&\overset{(d)}{\preceq}
					\bP_1[\Lambda_L]
					\!\cdot\!
					P_1\left[ 
						x' 
					\right]
					\!+\!						
					\frac{2}{1 \!-\! |d_2|}
					\!\cdot\!
					\Theta_L'
					\cdot
					P_1\left[
						x' 
					\right]
					\!+\!
					\frac{2}{1 \!-\! |d_2|}
					\!\cdot\!
					P_1\left[						
						y'
					\right]
					\nonumber\\
				&\overset{(e)}{\preceq}
					|d_2| 
					\!\cdot\! 
					I_L
					\!\cdot\!
					P_1\left[ 
						x' 
					\right]
					\!+\!						
					\frac{2}{1 \!-\! |d_2|}
					\!\cdot\!
					\Theta_L
					\cdot
					P_1\left[
						x' 
					\right]
					\!+\!
					\frac{2}{1 \!-\! |d_2|}
					\!\cdot\!
					P_1\left[						
						y'
					\right]
					\nonumber\\
				&\overset{(f)}{=}	
					\Gamma_e
					\cdot
					P_1\left[
						x' 
					\right]
					+ 
					\frac{2}{1-|d_2|}
					\cdot
					P_1\left[						
						y'
					\right]
		\label{Equ:Appendix:DLxe_Plus_ye}
	\end{align}
where step (a) uses the convexity property \eqref{Equ:Properties:PX_CvxComb},
step (b) uses the scaling property, step (c) uses variance relation
\eqref{Equ:VarPropt:Linear}, step (d) uses 
$\|\bP_1[\Lambda_L]\|_{\infty} = |d_2|$, $\bP_1[\Theta_L']=\Theta_L'$ and
$\|\bP_1[\Theta_L']\|_{\infty}=\|\Theta_L'\|_{\infty}=1$, step (e) uses
$\bP_1[\Lambda_L] \preceq |d_2| \cdot I_L$ and $\Theta_L' \preceq \Theta_L$, 
where $\Theta_L$ denotes a matrix of the same form as \eqref{Equ:Appendix_Theta_Ln} 
but of size $L\times L$, step (f) uses
the definition of the matrix $\Gamma_e$ in \eqref{Equ:Propert:Gamma_e}.
{ 
The above derivation assumes $|d_2| \neq 0$. When $|d_2|=0$, we can
verify that the above inequality still holds. To see this, we first
notice that when $|d_2|=0$, the relation $0 \le |d_{n_0}| \le \cdots \le |d_2|$
implies that $d_{n_0}=\cdots = d_2 =0$ so that $\Lambda_L = 0$ and 
$D_L = \Theta_L'$ --- see \eqref{Equ:Appendix:Lambda_L_def} and
\eqref{Equ:Appendx:DL_expression}. Therefore, similar to the 
steps (a)--(f) in \eqref{Equ:Appendix:DLxe_Plus_ye}, we get
	\begin{align}
		P_1[D_L x' + y']	&=	P_1[\Theta_L' x' + y']				\nonumber\\
							&=	P_1\big[\frac{1}{2}\cdot 2\Theta_L' x' + \frac{1}{2}\cdot 2y'\big]	
								\nonumber\\
							&\preceq
								\frac{1}{2} \cdot P_1[ 2\Theta_L' x']
								+ 
								\frac{1}{2}\cdot P_1[ 2y']	
								\nonumber\\
							&=	\frac{1}{2} \cdot 2^2 \cdot P_1[ \Theta_L' x']
								+ 
								\frac{1}{2}\cdot 2^2 \cdot P_1[ y']	
								\nonumber\\
							&=	2 P_1[ \Theta_L' x']
								+ 
								2 P_1[ y']	
								\nonumber\\
							&\preceq
								2 \| \bP_1[\Theta_L']\|_{\infty}
								\cdot
								\bP_1[\Theta_L'] P_1[x']
								+
								2
								P_1[y']
								\nonumber\\
							&=	2\Theta_L' P_1[x'] + 2P_1[y']
								\nonumber\\
							&\preceq
								2\Theta_L P_1[x'] + 2 P_1[y']
	\end{align}
By \eqref{Equ:Propert:Gamma_e}, we have $\Gamma_e = 2\Theta_L$ when $|d_2|=0$.
Therefore, the above expression is the same as the one on the right-hand side
of \eqref{Equ:Appendix:DLxe_Plus_ye}.
}

{\bf (\emph{Property 5: Stable Kronecker Jordan operator})}
Using \eqref{Equ:Appendx:DL_expression} we have
	\begin{align}
		P&\left[\mc{D}_L x_e + y_e\right]
				\nonumber\\
			&=	P\left[
					(\Lambda_L \otimes I_M) x_e
					+
					(\Theta_L' \otimes I_M) x_e
					+
					y_e
				\right]
				\nonumber\\
			&=	P\Big[
					|d_2| \!\cdot\!
					\frac{1}{|d_2|} \!\cdot\!
					(\Lambda_L \!\otimes\! I_M) x_e
					\!+\!
					\frac{1 \!-\! |d_2|}{2}
					\!\cdot\!
					\frac{2}{1 \!-\! |d_2|}
					\!\cdot\!
					(\Theta_L' \!\otimes\! I_M) x_e
					\nn\\
					&\qquad
					+
					\frac{1-|d_2|}{2}\cdot
					\frac{2}{1-|d_2|}
					\cdot
					y_e
				\Big]
				\nonumber\\
			&\overset{(a)}{\preceq}
				|d_2| \!\cdot\!
				P\Big[					
					\frac{1}{|d_2|} 
					\!\cdot\!
					(\Lambda_L \!\otimes\! I_M) x_e
				\Big]
					\nn\\
					&\quad
					+
				\frac{1 \!-\! |d_2|}{2}
				\!\cdot\!
				P\Big[
					\frac{2}{1 \!-\! |d_2|}
					\!\cdot\!
					(\Theta_L' \otimes I_M) x_e
				\Big]
					\nn\\
					&\quad
					+
				\frac{1 \!-\! |d_2|}{2}\cdot
				P\Big[
					\frac{2}{1 \!-\! |d_2|}
					\!\cdot\!
					y_e
				\Big]
				\nonumber\\
			&\overset{(b)}{=}
				\frac{1}{|d_2|} \cdot
				P[						
					(\Lambda_L \otimes I_M) x_e
				]
					+
				\frac{2}{1-|d_2|}
				\cdot
				P[
					(\Theta_L' \otimes I_M) x_e
				]
				\nn\\
				&\quad
				+
				\frac{2}{1-|d_2|}
				\cdot
				P[
					y_e
				]
				\nonumber\\
			&\overset{(c)}{\preceq}
				\frac{\| \bP[(\Lambda_L \!\otimes\! I_M)] \|_{\infty}}{|d_2|} 
				\!\cdot\!
				\bP[(\Lambda_L \otimes I_M)]
				\!\cdot\!
				P[
					x_e
				]
				\nn\\
				&\quad
				+
				\frac{2\|\bP[\Theta_L' \!\otimes\! I_M]\|_{\infty}}{1-|d_2|}
				\!\cdot\!
				\bP[\Theta_L' \!\otimes\! I_M]
				\!\cdot\!
				P[
					x_e
				]
				+
				\frac{2}{1 \!-\! |d_2|}
				\!\cdot\!
				P[
					y_e
				]
				\nonumber\\
			&\overset{(d)}{\preceq}
				\bP[(\Lambda_L \otimes I_M)]
				\cdot
				P[
					x_e
				]
				+
				\frac{2}{1-|d_2|}
				\cdot
				\Theta_L'
				\cdot
				P[
					x_e
				]
				\nn\\
				&\quad
				+
				\frac{2}{1-|d_2|}
				\cdot
				P[
					y_e
				]
				\nonumber\\
			&\overset{(e)}{\preceq}
				|d_2| \!\cdot\! I_L
				\!\cdot\!
				P[
					x_e
				]
				\!+\!
				\frac{2}{1\!-\!|d_2|}
				\!\cdot\!
				\Theta_L
				\cdot
				P[
					x_e
				]
				\!+\!
				\frac{2}{1\!-\!|d_2|}
				\!\cdot\!
				P[
					y_e
				]
				\nonumber\\
			&\overset{(f)}{=}
				\Gamma_e \cdot
				P[
					x_e
				]
				+
				\frac{2}{1-|d_2|}
				\cdot
				P[
					y_e
				]
		\label{Equ:Appendix:P1_DL_x_y}
	\end{align}
where step (a) uses the convexity property \eqref{Equ:Properties:PX_CvxComb},
step (b) uses the scaling property, step (c) uses variance relation
\eqref{Equ:VarPropt:Linear}, step (d) uses 
$\|\bP[\Lambda_L\otimes I_M]\|_{\infty} = |d_2|$
and $\bP[\Theta_L' \otimes I_M]=\Theta_L'$, step (e) uses
$\bP[\Lambda_L \otimes I_M] \preceq |d_2| \cdot I_L$ 
and $\Theta_L' \preceq \Theta_L$, 
and step (f) uses
the definition of the matrix $\Gamma_e$ in \eqref{Equ:Propert:Gamma_e}.
{ 
Likewise, we can also verify that the above inequalty holds for the case $|d_2|=0$.
}

\section{Proof of Theorem \ref{Thm:LimitPoint}}
\label{Appendix:Proof_Thm_LimitPoint}
Consider the following operator:
			\begin{align}
				T_0(w)		\defeq		w 
										- 
										\frac{\lambda_L}{\|p\|_1^2 \lambda_U^2}
										\sum_{k=1}^N p_k s_k(w)
				\label{Equ:ThmProof:LimitPoint:T0_def}
			\end{align}
		As long as we are able to show that $T_0(w)$ is a strict contraction
		mapping, i.e., $\forall \; x,y$, $\|T_0(x)-T_0(y)\| \le \gamma_0 \|x-y\|$
		 with $\gamma_o <1$, then we can invoke the Banach fixed point theorem
		\cite[pp.299-300]{kreyszig1989introductory} to conclude
		that there exists a unique $w^o$ such that $w^o=T_0(w^o)$, i.e.,
			\begin{align}
				&w^o		=		w^o
								-
								\frac{\lambda_L}{\|p\|_1^2 \lambda_U^2}
										\sum_{k=1}^N p_k s_k(w^o)
								\;\Leftrightarrow\;
										\sum_{k=1}^N p_k s_k(w^o) = 0
			\end{align}
		as desired. Now, to show that $T_0(\cdot)$ defined in \eqref{Equ:ThmProof:LimitPoint:T0_def}
		is indeed a contraction, we compare $T_0(\cdot)$ with $T_c(\cdot)$ in 
		\eqref{Equ:Def:Tc2} and observe that $T_0(w)$ has the same form
		as $T_c(\cdot)$ if we set $\mu_{\max}= \frac{\lambda_L}{\|p\|_1^2 \lambda_U^2}$
		in \eqref{Equ:Def:Tc2}. Therefore, calling upon property
		\eqref{Equ:VarPropt:Tc} and using $\mu_{\max}= \frac{\lambda_L}{\|p\|_1^2 \lambda_U^2}$
		in the expression for $\gamma_c$ in \eqref{Equ:VarPropt:gamma_c}, we obtain
			\begin{align}
				P&[T_0(x)-T_0(y)]	
												\nn\\
									&\preceq	
												\left(
													1 \!- \!
													\frac{\lambda_L}{\|p\|_1^2 \lambda_U^2}
													\lambda_L
													\!+\!
													\frac{1}{2}
													\left(
														\frac{\lambda_L}{\|p\|_1^2 \lambda_U^2}
													\right)^2
													\|p\|_1^2 \lambda_U^2
												\right)^2
												\!\cdot\!
												P[x-y]
												\nonumber\\
									&=			\left(
													1 - \frac{1}{2}
													\frac{\lambda_L^2}{\|p\|_1^2 \lambda_U^2}
												\right)^2
												\cdot
												P[x-y]
			\end{align}
		By the definition of $P[\cdot]$ in \eqref{Equ:Def:PowerVectorOperator},
		the above inequality is equivalent to
			\begin{align}
				\|T_0(x)-T_0(y)\|^2	\le 	\left(
												1 - \frac{1}{2}
												\frac{\lambda_L^2}{\|p\|_1^2 \lambda_U^2}
											\right)^2
											\cdot	
											\|x-y\|^2
			\end{align}
		It remains to show that $|1-\lambda_L^2/(2\|p\|_1^2 \lambda_U^2)|<1$.
		By \eqref{Equ:Remark:lambdaU_lambda_L_relation} and the fact
		that $\lambda_L$, $\|p\|_1^2$ and $\lambda_U^2$ are positive, we have
			\begin{align}
					\frac{1}{2}		<	1 - \frac{1}{2}
										\frac{\lambda_L^2}{\|p\|_1^2 \lambda_U^2}
									<	1
			\end{align}
		Therefore, $T_0(w)$ is a strict contraction mapping.

\section{Proof of Theorem \ref{Thm:ConvergenceRefRec:DeterministcCent}}
\label{Appendix:Proof_Thm_ConvergenceRefRec}

By Theorem \ref{Thm:LimitPoint}, $w^o$ is the unique solution
to equation \eqref{Equ:Lemma:LimitPoint_def}. Subtracting both sides of
\eqref{Equ:Lemma:LimitPoint_def} from $w^o$, we recognize that $w^o$
is also the unique solution to the following equation:
	\begin{align}
		w^o	=	w^o - \mu_{\max} \sum_{k=1}^N p_k s_k(w^o)
		\label{Equ:Appendix:ConvergenceRate:FixedEqu}
	\end{align}
so that $w^o=T_c(w^o)$. Applying property \eqref{Equ:VarPropt:Tc}, we obtain
			\begin{align}
				\|\tilde{w}_{c,i}\|^2
						&=			P[w^o - \bar{w}_{c,i}]	
									\nonumber\\
						&=			P[T_c(w^o) - T_c(\bar{w}_{c,i-1})]		
									\nn\\
						&\preceq		\gamma_c^2 \cdot P[w^o-\bar{w}_{c,i-1}]
									\nonumber\\
						&\preceq		\gamma_c^{2i} \cdot P[w^o-\bar{w}_{c,0}]
									\nonumber\\
						&=			\gamma_c^{2i} \cdot \|\tilde{w}_{c,0}\|^2
			\end{align}
		Since $\gamma_c>0$, the upper bound on the right-hand side will converge to zero if
		$\gamma_c<1$. From its definition \eqref{Equ:VarPropt:gamma_c}, this condition
		is equivalent to requiring
			\begin{align}
				1 - \mu_{\max}\lambda_L + \frac{1}{2}\mu_{\max}^2 \|p\|_1^2 \lambda_U^2 < 1
			\end{align}
		Solving the above quadratic inequality in $\mu_{\max}$, we obtain
		\eqref{Equ:Thm:ConvergenceRefRec:StepSize}. On the other hand, to
		prove the lower bound in \eqref{Equ:Thm:ConvergenceRefRec:StepSize},
		we apply \eqref{Equ:VarPropt:Tc_LB} and obtain
			\begin{align}
				\|\tilde{w}_{c,i}\|^2
									&=	P[w^o-\bar{w}_{c,i}]	
										\nonumber\\
									&=	P[T_c(w^o) - T_c(\bar{w}_{c,i-1})]
										\nonumber\\
									&\succeq
										(1-2\mu_{\max} \|p\|_1 \lambda_U)
										\cdot
										P[w^o-\bar{w}_{c,i-1}]
										\nonumber\\
									&\succeq
										(1-2\mu_{\max} \|p\|_1 \lambda_U)^{i}
										\cdot
										P[w^o-\bar{w}_{c,0}]
										\nonumber\\
									&=	(1-2\mu_{\max} \|p\|_1 \lambda_U)^{i}
										\cdot
										\|\tilde{w}_{c,0}\|^2
			\end{align}

\section{Proof of Theorem \ref{Thm:ConvergenceRateRefRec}}
\label{Appendix:Proof_Thm_ConvergenceRefRec_approx}

Since \eqref{Equ:Thm:ConvergenceRefRec:NonAsympBound} already establishes that $\bar{w}_{c,i}$ approaches $w^o$ asymptotically (so that $\tilde{w}_{c,i} \rightarrow 0$), and since from Assumption \ref{Assumption:JacobianUpdatVectorLipschitz} we know that $s_k(w)$ is differentiable when $\|\tilde{w}_{c,i}\| \le r_H$ for large enough $i$, we are justified to use the mean-value theorem \cite[p.24]{poliak1987introduction} to obtain the following useful relation:
	\begin{align}
		s_k&(\bar{w}_{c,i-1}) - s_k(w^o)	
					\nn\\
				&=	-\left[\int_0^1 \nabla_{w^T} s_k(w^o - t\tilde{w}_{c,i-1}) dt\right]
					\tilde{w}_{c,i-1}
					\nonumber\\
				&=
					-\nabla_{w^T} s_k(w^o)
					\cdot
					\tilde{w}_{c,i-1} 
					\nn\\
					&\quad
					-
					\int_0^{1} 
					\big[ 
						\nabla_{w^T} s_k(w^o - t\tilde{w}_{c,i-1}) - \nabla_{w^T} s_k(w^o) 
					\big] dt
					\cdot
					\tilde{w}_{c,i-1}
	\end{align}
Therefore, subtracting $w^o$ from both sides of \eqref{Equ:LearnBehav:RefRec} 
and using \eqref{Equ:Lemma:LimitPoint_def} we get,
	\begin{align}
		\tilde{w}_{c,i}		&=	 \tilde{w}_{c,i-1}
									+
									\mu_{\max}\sum_{k=1}^N 
									p_k
									(s_k(\bar{w}_{c,i-1})-s_k(w^o))
								\nonumber\\
							&=	
								[ I - \mu_{\max} H_c ] \tilde{w}_{c,i-1}
								-
								\mu_{\max}
								\cdot
								e_{i-1}
		\label{Equ:Appendix:RefRec:Recursion_perturbed}							
	\end{align}
where
	\begin{align}
		H_c		&\defeq		\sum_{k=1}^N p_k \nabla_{w^T} s_k(w^o)
							\\
		e_{i-1}	&\defeq		\sum_{k=1}^N  p_k \!\!
							\int_0^{1} \!\!
							\big[ 
								\nabla_{w^T} s_k(w^o \!-\! t\tilde{w}_{c,i-1}) \!-\! \nabla_{w^T} s_k(w^o) 
							\big] dt
							\!\cdot\!
							\tilde{w}_{c,i-1}
	\end{align}
Furthermore, the perturbation term $e_{i-1}$ satisfies the following bound:
	\begin{align}
		\| e_{i-1} \|		
				&\le 			
							\sum_{k=1}^N 
							p_k
							\int_0^{1} 
							\big\| 
								\nabla_{w^T} s_k(w^o \!-\! t\tilde{w}_{c,i-1}) \!-\! \nabla_{w^T} s_k(w^o) 
							\big\| dt
							\nn\\
							&\quad
							\cdot
							\| \tilde{w}_{c,i-1} \|
							\nn\\
				&\le 			
							\sum_{k=1}^N 
							p_k
							\int_0^{1} 
							\lambda_H
							\cdot
							t
							\cdot
							\| \tilde{w}_{c,i-1} \|
							dt
							\cdot
							\| \tilde{w}_{c,i-1} \|
							\nn\\
				&=			\frac{1}{2} \|p\|_1 \lambda_H
							\cdot
							\| \tilde{w}_{c,i-1} \|^2
		\label{Equ:Appendix:RefRec:e_bound}
	\end{align}
Evaluating the weighted Euclidean norm of both sides of \eqref{Equ:Appendix:RefRec:Recursion_perturbed}, we get
	\begin{align}
		\| \tilde{w}_{c,i} \|_{\Sigma}^2	
				&=			\| \tilde{w}_{c,i-1} \|_{B_c^T \Sigma B_c}^2
							-
							2\mu_{\max}
							\cdot
							\tilde{w}_{c,i-1}^T B_c^T \Sigma e_{i-1}
							\nn\\
							&\quad
							+
							\mu_{\max}^2
							\cdot
							\| e_{i-1} \|_{\Sigma}^2
		\label{Equ:Appendix:RefRec:wc_varRelation_interm1}
	\end{align}
where 
	\begin{align}
		B_c		=			I - \mu_{\max} H_c
		\label{Equ:Appendix:RefRec:Bc_def}
	\end{align}
Moreover, $\| x \|_{\Sigma}^2 = x^T \Sigma x$, and $\Sigma$ is an arbitrary positive semi-definite weighting matrix. The second and third terms on the right-hand side of \eqref{Equ:Appendix:RefRec:wc_varRelation_interm1} satisfy the following bounds:
	\begin{align}
		\big| &\tilde{w}_{c,i-1}^T B_c^T \Sigma e_{i-1} \big|
							\nn\\
				&\le
							\| \tilde{w}_{c,i-1} \| \cdot
							\| B_c^T \|
							\cdot
							\| \Sigma \|
							\cdot
							\| e_{i-1} \|
							\nn\\
				&\overset{(a)}{\le}
							\| \tilde{w}_{c,i-1} \| \cdot
							\| B_c^T \|
							\cdot
							\Tr( \Sigma )
							\cdot
							\| e_{i-1} \|
							\nn\\
				&\le 			
							\| \tilde{w}_{c,i-1} \| \cdot
							\| B_c^T \|
							\cdot
							\Tr( \Sigma )
							\cdot
							\frac{\lambda_H \|p\|_1 }{2}
							\cdot
							\| \tilde{w}_{c,i-1} \|^2
	\end{align}
and
	\begin{align}
		\| e_{i-1} \|_{\Sigma}^2
				&\le
							\| \Sigma \| 
							\cdot
							\|e_{i-1}\|^2
							\nn\\
				&\overset{(b)}{\le}
				 			\Tr ( \Sigma )
							\cdot
							\| e_{i-1} \|^2
							\nn\\
				&\le 			\Tr ( \Sigma )
							\cdot
							\frac{\lambda_H^2 \|p\|_1^2 }{4}
							\cdot
							\| \tilde{w}_{c,i-1} \|^4
	\end{align}
where for steps (a) and (b) of the above inequalities we used the property $\| \Sigma \| \le \Tr(\Sigma)$. This is because we consider here the spectral norm, $\|\Sigma\| = \sigma_{\max}(\Sigma)$, where $\sigma_{\max}(\cdot)$ denotes the maximum singular value. Since $\Sigma$ is symmetric and positive semidefinite, its singular values are the same as its eigenvalues so that $\|\Sigma\| = \lambda_{\max}(\Sigma) \le \sum_{m} \lambda_m(\Sigma) = \Tr(\Sigma)$. Now, for any given small $\epsilon > 0$, there exists $i_0$ such that, for $i \ge i_0$, we have $\| \tilde{w}_{c,i-1} \| \le \epsilon$ so that
	\begin{align}
		\big| \tilde{w}_{c,i-1}^T B_c^T \Sigma e_{i-1} \big|
				&\le 			
							\epsilon \cdot
							\| B_c^T \|
							\!\cdot\!
							\Tr( \Sigma )
							\!\cdot\!
							\frac{\lambda_H \|p\|_1 }{2}
							\!\cdot\!
							\| \tilde{w}_{c,i-1} \|^2
		\label{Equ:Appendix:RefRec:VarRelation_Bound2nd}
							\\
		\| e_{i-1} \|_{\Sigma}^2
				&\le 			\epsilon^2
							\cdot
							\Tr ( \Sigma )
							\cdot
							\frac{\lambda_H^2 \|p\|_1^2 }{4}
							\cdot
							\| \tilde{w}_{c,i-1} \|^2
		\label{Equ:Appendix:RefRec:VarRelation_Bound3rd}
	\end{align}
Substituting \eqref{Equ:Appendix:RefRec:VarRelation_Bound2nd}--\eqref{Equ:Appendix:RefRec:VarRelation_Bound3rd} into \eqref{Equ:Appendix:RefRec:wc_varRelation_interm1}, we obtain
	\begin{align}
		\| \tilde{w}_{c,i-1} \|_{B_c^T \Sigma B_c - \Delta}^2
				&\le 			\| \tilde{w}_{c,i} \|_{\Sigma}^2
				\le			\| \tilde{w}_{c,i-1} \|_{B_c^T \Sigma B_c + \Delta}^2
		\label{Equ:Appendix:RefRec:VarRelation_UB_LB}
	\end{align}
where
	\begin{align}
		\Delta	&\defeq		
							\mu_{\max} 
							\epsilon
							\cdot
							\lambda_H 
							\|p\|_1
							\cdot
							\big[		
								\|B_c^T\|						
								+
								\mu_{\max}
								\epsilon
								\frac{\lambda_H \|p\|_1}{4}								
							\big]
							\cdot
							\Tr(\Sigma)
							\cdot
							I_M
							\nn\\
				&=			O(\mu_{\max} \epsilon) \cdot \Tr(\Sigma) \cdot I_M
	\end{align}
Let $\sigma = \mathrm{vec}(\Sigma)$ denote the vectorization operation that stacks the columns of a matrix $\Sigma$ on top of each other. We shall use the notation $\| x \|_{\sigma}^2$ and $\|x\|_{\Sigma}^2$ interchangeably to denote the weighted squared Euclidean norm of a vector. Using the Kronecker product property \cite[p.147]{laub2005matrix}: $\mathrm{vec}(U\Sigma V) = (V^T \otimes U) \mathrm{vec}(\Sigma)$, we can vectorize the matrices $B_c^T \Sigma B_c + \Delta$ and $B_c^T \Sigma B_c - \Delta$ in \eqref{Equ:Appendix:RefRec:VarRelation_UB_LB} as $\mc{F}_{+} \sigma$ and $\mc{F}_{-} \sigma$, respectively, where
	\begin{align}
		\mc{F}_{+}	
				&\defeq		
							B_c^T \!\otimes\! B_c^T
							\!+\!
							\mu_{\max} 
							\epsilon
							\!\cdot\!
							\lambda_H 
							\|p\|_1
							\!\cdot\!
							\big[		
								\|B_c^T\|						
								\!+\!
								\mu_{\max}
								\epsilon
								\frac{\lambda_H \|p\|_1}{4}								
							\big]
							q q^T
							\nn\\
				&=			
							B_c^T \otimes B_c^T
							+
							O(
							\mu_{\max} 
							\epsilon)
	\label{Equ:Appendix:RefRec:Fplus_def}
							\\
		\mc{F}_{-}	
				&\defeq		B_c^T \!\otimes\! B_c^T
							\!-\!
							\mu_{\max} 
							\epsilon
							\!\cdot\!
							\lambda_H 
							\|p\|_1
							\!\cdot\!
							\big[		
								\|B_c^T\|						
								\!+\!
								\mu_{\max}
								\epsilon
								\frac{\lambda_H \|p\|_1}{4}								
							\big]
							q q^T
							\nn\\
				&=			
							B_c^T \otimes B_c^T
							-
							O(
							\mu_{\max} 
							\epsilon)
		\label{Equ:Appendix:RefRec:Fminus_def}
	\end{align}
where $q \defeq \mathrm{vec}(I_M)$, and we have used the fact that $\Tr(\Sigma) = \Tr(\Sigma I_M) = \mathrm{vec}(I_M)^T \mathrm{vec}(\Sigma) = q^T \sigma$. In this way, we can write relation \eqref{Equ:Appendix:RefRec:VarRelation_UB_LB} as
	\begin{align}
		\| \tilde{w}_{c,i-1} \|_{\mc{F}_{-} \sigma}^2
				&\le 			\| \tilde{w}_{c,i} \|_{\sigma}^2
				\le			\| \tilde{w}_{c,i-1} \|_{\mc{F}_{+} \sigma}^2
		\label{Equ:Appendix:RefRec:VarRelation_UB_LB_new}
	\end{align}
Using a state-space technique from \cite[pp.344-346]{Sayed08}, we conclude that $\| \tilde{w}_{c,i} \|_{\Sigma}^2$ converges at a rate that is between $\rho(\mc{F}_{-})$ and $\rho(\mc{F}_{+})$. Recalling from \eqref{Equ:Appendix:RefRec:Fplus_def}--\eqref{Equ:Appendix:RefRec:Fminus_def} that $\mc{F}_{+}$ and $\mc{F}_{-}$ are perturbed matrices of $B_c^T \otimes B_c^T$, and since the perturbation term is $O(\epsilon \mu_{\max})$ which is small for small $\epsilon$, we would expect the spectral radii of $\mc{F}_{+}$ and $\mc{F}_{-}$ to be small perturbations of $\rho( B_c^T \otimes B_c^T)$. This claim is justified below. 
	\begin{lemma}[Perturbation of spectral radius]
		\label{Lemma:PerturbationSpectralRadius}
		Let $\epsilon \ll 1$ be a sufficiently small positive number. 
		For any $M \times M$ matrix $X$, the spectral
		radius of the perturbed matrix $X + E$ for $E=O(\epsilon)$ is 
			\begin{align}
				\rho(X+E)	=	\rho(X) + O\big(\epsilon^{\frac{1}{2(M-1)}}\big)
				\label{Equ:Appendix:rho_perturbation}
			\end{align}
	\end{lemma}
	\begin{proof}
	Let $X=T J T^{-1}$ be the Jordan canonical form of the matrix $X$. Without loss of generality, we  consider the case where there are two Jordan blocks:
	\begin{align}
	J	=	\diag\{J_1, J_2\}
	\end{align}
	where $J_1 \in \mathbb{R}^{L \times L}$ and $J_2 \in \mathbb{R}^{(M-L)\times(M-L)}$ are Jordan blocks of the form
	\begin{align}
	J_k 	=	\begin{bmatrix}
	           			\lambda_k & 1 & &  	\\
	           		    & \ddots & \ddots &	\\
	           		    & & \ddots & 1	\\
	           		    & & & \lambda_k
	           		\end{bmatrix}			
	\end{align}		
	with $|\lambda_1| > |\lambda_2|$.
	Since $X+E$ is similar to $T^{-1}(X+E)T$, the matrix $X+E$ has the same set of eigenvalues as $J + E_0$ where 
	\begin{align}
	E_0 \defeq T^{-1} E T = O(\epsilon)
	\end{align}
	Let
	\begin{align}
	\epsilon_0 		&\defeq \epsilon^{\frac{1}{2(M-1)}}
	\label{Equ:Appendix:epsilon0_def}
							\\
	D_{\epsilon_0}	&\defeq	\diag\left\{
								1, \epsilon_0, \ldots, \epsilon_0^{M-1}
							\right\}
	\label{Equ:Appendix:D_epsilon0_def}
	\end{align}
	Then, by similarity again, the matrix $J + E_0$ has the same set of eigenvalues as 
	\begin{align}
	D_{\epsilon_0}^{-1} (J + E_0) D_{\epsilon_0}	=	D_{\epsilon_0}^{-1} J D_{\epsilon_0} + E_1
	\end{align}
	where $E_1 \defeq D_{\epsilon_0}^{-1} E_0 D_{\epsilon_0}$. Note that the $\infty$-induced norm (the maximum absolute row sum) of $E_1$ is bounded by
	\begin{align}
	\|E_1\|_{\infty} 		&\le  		\|D_{\epsilon_0}^{-1}\|_{\infty} 
							\cdot
							\|E_0\|_{\infty}
							\cdot 
							\|D_{\epsilon_0}\|_{\infty}
							\nn\\
					&=	 	
							\frac{1}{\epsilon_0^{M-1}}  \cdot O(\epsilon) \cdot 1
					=		\frac{1}{\epsilon^{\frac{1}{2}}} \cdot O(\epsilon)
					=		O(\epsilon^{\frac{1}{2}})
	\label{Equ:Appendix:E1_infnorm_bound}
	\end{align}
	and that
	\begin{align}
	D_{\epsilon_0}^{-1} J D_{\epsilon_0} = \diag\{J_1', J_2'\}
	\end{align}
	where
	\begin{align}
	J_k'
					&=		\begin{bmatrix}
			           			\lambda_k & \epsilon_0 & &  	\\
			           		    & \ddots & \ddots &	\\
			           		    & & \ddots & \epsilon_0	\\
			           		    & & & \lambda_k
			           		\end{bmatrix}
	\end{align}
	Then, by appealing to Ger\v{s}gorin Theorem\cite[p.344]{horn1990matrix}, we conclude that the eigenvalues of the matrix $D_{\epsilon_0}^{-1} J D_{\epsilon_0} + E_1$, which are also the  eigenvalues of the matrices $J+E_0$ and $X+E$, lie inside the union of the Ger\v{s}gorin discs, namely, 
	\begin{align}
	\bigcup_{m=1}^M \mc{G}_m
	\end{align}
	where $\mc{G}_m$ is the $m$th Ger\v{s}gorin disc defined as
	\begin{align}
	\mc{G}_m		&\defeq		\begin{cases}
								\displaystyle
								\left\{
									\lambda:\;
									|\lambda - \lambda_1| \le \epsilon_0 + \sum_{\ell=1}^M |E_{1,m \ell}|
								\right\},
												& 1 \le m \le L			\\
								\displaystyle
								\left\{
									\lambda:\;
									|\lambda - \lambda_2| \le \epsilon_0 + \sum_{\ell=1}^M |E_{1,m \ell}|
								\right\},
												& L < m \le M
							\end{cases}
							\nonumber\\
				&=			\begin{cases}
								\displaystyle
								\left\{
									\lambda:\;
									|\lambda \!-\! \lambda_1| 
										\le 
											O\big(\epsilon^{\frac{1}{2(M-1)}}\big)
								\right\},
												& 1 \!\le\! m \le L			\\
								\displaystyle
								\left\{
									\lambda:\;
									|\lambda \!-\! \lambda_2| 
										\le 
											O\big(\epsilon^{\frac{1}{2(M-1)}}\big)
								\right\},
												& L \!<\! m \!\le\! M
							\end{cases} \!\!
	\label{Equ:Appendix:Gm_def}
	\end{align}
	and where $E_{1, m \ell}$ denotes the $(m,\ell)$-th entry of the matrix $E_1$. In the last step we used \eqref{Equ:Appendix:epsilon0_def} and \eqref{Equ:Appendix:E1_infnorm_bound}. Observe from \eqref{Equ:Appendix:Gm_def} that there are two clusters of Ger\v{s}gorin discs that are centered around $\lambda_1$ and $\lambda_2$, respectively, and have radii on the order of $O(\epsilon^{\frac{1}{2(M-1)}})$. A further statement from Ger\v{s}gorin theorem shows that if the these two clusters of discs happen to be disjoint, which is true in our case since $|\lambda_1|>|\lambda_2|$ and we can select $\epsilon$ to be sufficiently small to ensure this property. Then there are exactly $L$ eigenvalues of $X+E$ in 
	$\cup_{m=1}^L \mc{G}_m$ while the remaining $M-L$ eigenvalues are in $\cup_{m=M-L}^M \mc{G}_m$. From $|\lambda_1|>|\lambda_2|$, we conclude that the largest eigenvalue of $D_{\epsilon_0}^{-1} J D_{\epsilon_0} + E_1$ is $\lambda_1 + O(\epsilon^{\frac{1}{2(M-1)}})$, which establishes \eqref{Equ:Appendix:rho_perturbation}. \\
	\end{proof}

Using \eqref{Equ:Appendix:rho_perturbation} for $\mc{F}_{+}$ and $\mc{F}_{-}$ in \eqref{Equ:Appendix:RefRec:Fplus_def}--\eqref{Equ:Appendix:RefRec:Fminus_def}, we conclude that
	\begin{align}
		\rho(\mc{F}_{+})	&=		\big[\rho(I_M - \mu_{\max} H_c)]^2 
									+  
									O\big( (\mu_{\max} \epsilon )^{\frac{1}{2(M-1)}} \big)
								\\
		\rho(\mc{F}_{-})	&=		\big[\rho(I_M - \mu_{\max} H_c)]^2 
									+  
									O\big( (\mu_{\max} \epsilon )^{\frac{1}{2(M-1)}} \big)
	\end{align}
which holds for arbitrarily small $\epsilon$. Since the convergence rate of $\|\tilde{w}_{c,i}\|^2$ is between $\rho(\mc{F}_{+})$ and $\rho(\mc{F}_{-})$, we arrive at \eqref{Equ:DistProc:r_RefRec}.

\section{Proof of Lemma \ref{Lemma:IneqRecur_W_check_prime}}
\label{Appendix:Proof_Lemma_W_check_prime_recursion}

From the definition in \eqref{Equ:DistProc:mW_check_prime_def}, 
it suffices to establish a joint inequality recursion
for both $\E P[\check{\bm{w}}_{c,i}]$ and $\E P[\bm{w}_{e,i}]$. 
To begin with, we state the following bounds on the perturbation terms in
\eqref{Equ:DistProc:s_hat_phi_decomp_zdef_vdef}.
	\begin{lemma}[Bounds on the perturbation terms]
		\label{Lemma:BoundsPerturbation}
		The following bounds hold for any $i \ge 0$.
			\begin{align}
				\label{Equ:Lemma:BoundsPerturbation:P_z}
				P[\bm{z}_{i\!-\!1}]	&\preceq		\lambda_U^2
											\cdot
											\left\|
												\bP_1[A_1^T U_L]
											\right\|_{\infty}^2
											\!\cdot\!
											\mathds{1}\mathds{1}^T
											\!\cdot\!
											P[\bm{w}_{e,i\!-\!1}]
											\\
				\label{Equ:Lemma:BoundsPerturbaton:P_s}
				P[s(\mathds{1}\!\otimes\! \bm{w}_{c,i\!-\!1})]
								&\preceq		3\lambda_U^2
											\!\cdot\!
											P[\check{\bm{w}}_{c,i - 1}]
											\!\cdot\!
											\mathds{1}
											\!+\!
											3\lambda_U^2 \|\tilde{w}_{c,0}\|^2 \!\cdot\! \mathds{1}
											\!+\!
											3g^o
											\\
				\E \{P[\bm{v}_i] | \mF_{i-1} \}	
								&\preceq		
											4\alpha \cdot \one
											\cdot
											P[ \check{\w}_{c,i-1} ]
											\nn\\
											&\quad
											+
											4 \alpha
											\cdot
											\| \bP[ \mA_1^T \mU_L ] \|_{\infty}^2
											\cdot
											\one \one^T
											P[ \w_{e,i-1} ]
											\nn\\
											&\quad
											\!+\!
											\left[
												4\alpha
												\cdot
												( \|\tilde{w}_{c,0}\|^2 \!+\! \|w^o\|^2 )
												\!+\!
												\sigma_v^2
											\right]
											\cdot
											\one
				\label{Equ:Lemma:BoundsPerturbation:P_v_E_Fiminus1}
											\\
				\E P[\bm{v}_i]
								&\preceq		
											4\alpha \cdot \one
											\cdot
											\Expt P[ \check{\w}_{c,i-1} ]
											\nn\\
											&\quad
											+
											4 \alpha
											\cdot
											\| \bP[ \mA_1^T \mU_L ] \|_{\infty}^2
											\cdot
											\one \one^T
											\Expt P[ \w_{e,i-1} ]
											\nn\\
											&\quad
											+\!
											\left[
												4\alpha
												\cdot
												( \|\tilde{w}_{c,0}\|^2 \!+\! \|w^o\|^2 )
												\!+\!
												\sigma_v^2
											\right]
											\!\cdot\!
											\one
				\label{Equ:Lemma:BoundsPerturbation:P_v}				
			\end{align}
		where $P[\check{\w}_{c,i-1}]=\|\check{\w}_{c,i-1}\|^2$,
		and		$g^o		\defeq	P[s(\mathds{1} \otimes w^o)]$.
	\end{lemma}
	\begin{IEEEproof}
		See Appendix \ref{Appendix:Proof_BoundsPerturbation}.
	\end{IEEEproof}

Now, we derive an inequality recursion for $\E P[\check{\bm{w}}_{c,i}]$
from \eqref{Equ:Lemma:ErrorDynamics:JointRec_wc_check}. Note that
	\begin{align}
		\E &P[\check{\bm{w}}_{c,i}] = \E\|\check{\bm{w}}_{c,i}\|^2
						\nonumber\\
				&=		\E P\big[
							T_c(\bm{w}_{c,i-1})
							-
							T_c(\bar{w}_{c,i-1})
							-
							\mu_{\max}\cdot
							(p^T \otimes I_M)
							\bm{z}_{i-1} 
							\nn\\
							&\qquad
							-
							\mu_{\max}\cdot
							(p^T \otimes I_M)
							\bm{v}_i)
						\big]
						\nonumber\\						
				&\overset{(a)}{=}
						\E P\left[
							T_c(\bm{w}_{c,i-1})
							-
							T_c(\bar{w}_{c,i-1})
							-
							\mu_{\max}\cdot
							(p^T \otimes I_M)
							\bm{z}_{i-1} 
						\right]
							\nn\\
							&\quad
						+
						\mu_{\max}^2 \cdot
						\E P\left[
							(p^T \otimes I_M)
							\bm{v}_i)
						\right]
						\nonumber\\
				&=		
						\E P\bigg[
							\gamma_c \cdot
							\frac{1}{\gamma_c}
							\left(
								T_c(\bm{w}_{c,i-1}) - T_c(\bar{w}_{c,i-1})
							\right)
							\nn\\
							&\qquad
							+
							(1-\gamma_c)
							\cdot
							\frac{-\mu_{\max}}{1-\gamma_c}
							(p^T \otimes I_M) 
							\bm{z}_{i-1}
						\bigg]
						\nonumber\\
						&\quad+
						\mu_{\max}^2 \cdot
						\E P\left[
							(p^T \otimes I_M)
							\bm{v}_i)
						\right]
						\nonumber\\
				&\overset{(b)}{\preceq}
						\gamma_c \cdot
						\frac{1}{\gamma_c^2}
						\E P\left[
							T_c(\bm{w}_{c,i-1})
							-
							T_c(\bar{w}_{c,i-1})
						\right]
						\nn\\
						&\quad
						+
						(1-\gamma_c)
						\cdot
						\frac{\mu_{\max}^2}{(1-\gamma_c)^2}
						\E P\left[
							(p^T \otimes I_M) \bm{z}_{i-1}
						\right]
						\nonumber\\
						&\quad+
						\mu_{\max}^2 \E P\left[
							(p^T \otimes I_M) \bm{v}_{i}
						\right]
						\nonumber\\
				&\overset{(c)}{\preceq}
						\gamma_c \cdot 
						\E P\left[
							\check{\bm{w}}_{c,i-1}
						\right]
						+
						\frac{\mu_{\max}^2}{1-\gamma_c}
						\E P\left[
							(p^T \otimes I_M) \bm{z}_{i-1}
						\right]
						\nn\\
						&\quad
						+
						\mu_{\max}^2 \E P\left[
							(p^T \otimes I_M) \bm{v}_{i}
						\right]
						\nonumber\\
				&=
						\gamma_c \cdot 
						\E P\left[
							\check{\bm{w}}_{c,i-1}
						\right]
						+
						\frac{\mu_{\max}^2}{1-\gamma_c}
						\E \left\|
							(p^T \otimes I_M) \bm{z}_{i-1}
						\right\|^2
						\nn\\
						&\quad
						+
						\mu_{\max}^2 
						\E \left\|
							(p^T \otimes I_M) \bm{v}_{i}
						\right\|^2
						\nn\\
				&\overset{(d)}{=}
						\gamma_c \cdot 
						\E P\left[
							\check{\bm{w}}_{c,i-1}
						\right]
						+
						\frac{\mu_{\max}^2}{1-\gamma_c}
						\E \left\|
							\sum_{k=1}^N p_k \bm{z}_{k,i-1}
						\right\|^2
						\nn\\
						&\quad
						+
						\mu_{\max}^2 
						\E \left\|
							\sum_{k=1}^N p_k \bm{v}_{k,i}
						\right\|^2
						\nn\\
				&=
						\gamma_c \cdot 
						\E P\left[
							\check{\bm{w}}_{c,i-1}
						\right]
						\nn\\
						&\quad
						+
						\frac{\mu_{\max}^2}{1-\gamma_c}
						\!\!\cdot\!\!
						\left(
							\sum_{l=1}^N p_l
						\right)^2
						\!\!\!\cdot\!
						\E \left\|
							\sum_{k=1}^N \frac{p_k}{\sum_{l=1}^N p_l} \bm{z}_{k,i-1}
						\right\|^2
						\nn\\
						&\quad
						+
						\mu_{\max}^2 
						\!\!\cdot\!\!
						\left(
							\sum_{l=1}^N p_l
						\right)^2
						\!\!\!\cdot\!
						\E \left\|
							\sum_{k=1}^N \frac{p_k}{\sum_{l=1}^N p_l} \bm{v}_{k,i}
						\right\|^2
						\nn\\
				&\overset{(e)}{\le}
						\gamma_c \cdot 
						\E P\left[
							\check{\bm{w}}_{c,i-1}
						\right]
						\nn\\
						&\quad
						+
						\frac{\mu_{\max}^2}{1-\gamma_c}
						\!\!\cdot\!\!
						\left(
							\sum_{l=1}^N p_l
						\right)^2
						\!\!\!\cdot\!
						\sum_{k=1}^N \frac{p_k}{\sum_{l=1}^N p_l}
						\E \left\|
							\bm{z}_{k,i-1}
						\right\|^2
						\nn\\
						&\quad
						+
						\mu_{\max}^2 
						\!\!\cdot\!\!
						\left(
							\sum_{l=1}^N p_l
						\right)^2
						\!\!\!\cdot\!
						\sum_{k=1}^N \frac{p_k}{\sum_{l=1}^N p_l} 
						\E \left\|
							\bm{v}_{k,i}
						\right\|^2
						\nn\\
				&=
						\gamma_c \cdot 
						\E P\left[
							\check{\bm{w}}_{c,i-1}
						\right]
						+
						\frac{\mu_{\max}^2}{1-\gamma_c}
						\cdot
						\left(
							\sum_{l=1}^N p_l
						\right)
						\cdot
						\sum_{k=1}^N p_k
						\E \left\|
							\bm{z}_{k,i-1}
						\right\|^2
						\nn\\
						&\quad
						+
						\mu_{\max}^2 
						\cdot
						\left(
							\sum_{l=1}^N p_l
						\right)
						\cdot
						\sum_{k=1}^N p_k
						\E \left\|
							\bm{v}_{k,i}
						\right\|^2
						\nn\\
				&=
						\gamma_c \cdot 
						\E P\left[
							\check{\bm{w}}_{c,i-1}
						\right]
						+
						\frac{\mu_{\max}^2}{1-\gamma_c}
						\cdot
						\|p\|_1
						\cdot
						p^T
						\E P[\bm{z}_{i-1}]
						\nn\\
						&\quad
						+
						\mu_{\max}^2 
						\cdot
						\|p\|_1
						\cdot
						p^T
						\E P[ \bm{v}_{i} ]
						\nn\\
				&\overset{(f)}{=}
						\gamma_c \cdot 
						\E P\left[
							\check{\bm{w}}_{c,i-1}
						\right]
						+
						\frac{\mu_{\max}\cdot
						\|p\|_1}
						{\lambda_L - \frac{1}{2}\mu_{\max}\|p\|_1^2 \lambda_U^2}
						\cdot
						p^T
						\E P[\bm{z}_{i-1}]
						\nn\\
						&\quad
						+
						\mu_{\max}^2 \cdot
						\|p\|_1
						\cdot
						p^T
						\E P[\bm{v}_i]
						\nonumber\\
				&\overset{(g)}{\preceq}
						\gamma_c \cdot 
						\E P\left[
							\check{\bm{w}}_{c,i-1}
						\right]
						\nn\\
						&\quad
						+
						\frac{\mu_{\max}\|p\|_1}
						{\lambda_L - \mu_{\max}\frac{1}{2}\|p\|_1^2 \lambda_U^2}
						\nn\\
						&\quad
						\cdot
						p^T 
						\left\{
							\lambda_U^2
							\cdot
							\left\|
								\bP[\mc{A}_1^T \mc{U}_L]
							\right\|_{\infty}^2
							\cdot
							\mathds{1}\mathds{1}^T
							\cdot
							\E
							P[\bm{w}_{e,i-1}]
						\right\}
						\nonumber\\
						&\quad+
						\mu_{\max}^2 \cdot \|p\|_1
						\cdot p^T
						\Big\{
								4\alpha \cdot \one
								\cdot
								\Expt P[ \check{\w}_{c,i-1} ]
								\nn\\
								&\qquad
								+
								4 \alpha
								\cdot
								\| \bP[ \mA_1^T \mU_L ] \|_{\infty}^2
								\cdot
								\one \one^T
								\Expt P[ \w_{e,i-1} ]
								\nn\\
								&\qquad
								+
								\left[
									4\alpha
									\cdot
									( \|\tilde{w}_{c,0}\|^2 + \|w^o\|^2 )
									+
									\sigma_v^2
								\right]
								\cdot
								\one
						\Big\}
						\nonumber\\
				&\overset{(h)}{=}
						\left[
							\gamma_c + \mu_{\max}^2 \cdot 4\alpha \|p\|_1^2
						\right]
						\cdot
						\E P[\check{\bm{w}}_{c,i-1}]
						\nonumber\\
						&\quad+
						\|p\|_1^2 \cdot
						\left\|
								\bP[\mc{A}_1^T \mc{U}_L]
						\right\|_{\infty}^2
						\cdot
						\lambda_U^2
						\nn\\
						&\qquad
						\cdot
						\Big[							
							\frac{\mu_{\max}}
							{\lambda_L - \frac{1}{2}\mu_{\max}\|p\|_1^2 \lambda_U^2}
							+
							4
							\mu_{\max}^2	 \frac{\alpha}{\lambda_U^2}
						\Big]
						\cdot
						\mathds{1}^T 
						\E P[\bm{w}_{e,i-1}]
						\nonumber\\
						&\quad+
						\mu_{\max}^2 \cdot						
						\|p\|_1^2
						\cdot
						\left[
							4\alpha 
							\left(
								\|\tilde{w}_{c,0}\|^2 
								+
								\| w^o \|^2
							\right)
							+
							\sigma_v^2
						\right]	
		\label{Equ:Appendix:EP_wci_recursion_interm1}				
	\end{align}
where step (a) uses the additivity property in Lemma
\ref{Lemma:BasicPropertiesOperator} 
since the definition of $\bz_{i-1}$ and $\bv_i$ in \eqref{Equ:DistProc:s_hat_phi_decomp_zdef_vdef} and the definition of $\w_{c,i-1}$ in \eqref{Equ:DistProc:w_i_prime_w_ci_w_ei} imply that $\bz_{i-1}$ and $\w_{c,i-1}$  depend on all  $\{\bw_j\}$ for $j \le i-1$, meaning that the cross terms are zero:
		\begin{align*}
			\E [\bv_i \bz_{i-1}^T]	&=	\E \left\{											
											\E \left[
												\bv_i
												| \mc{F}_{i-1}
											\right]
											\bz_{i-1}^T
										\right\}
									=	0
										\\
			\E \big\{
				\bv_i  
				[ T_c( & \w_{c,i-1}) - T_c(\bar{w}_{c,i-1}) ]^T
			\big\}			
								\nn\\
							&=	
								\E	\left\{
										\E\left[
											\bv_i
											|
											\mF_{i-1}
										\right]
										[ T_c(\w_{c,i-1}) \!- \!T_c(\bar{w}_{c,i-1}) ]^T
									\right\}
							=		0
		\end{align*}
Step (b) uses the convexity property in Lemma \ref{Lemma:BasicPropertiesOperator}, step (c) uses the variance property \eqref{Equ:VarPropt:Tc}, step (d) uses the notation $\z_{k,i-1}$ and $\bv_{k,i}$ to denote the $k$th $M \times 1$ block of the $NM \times 1$ vector $\z_{i-1}$ and $\bv_{i}$, respectively, step (e) applies Jensen's inequality to the convex function $\| \cdot\|^2$, step (f) substitutes expression \eqref{Equ:VarPropt:gamma_c} for $\gamma_c$, step (g) substitutes the bounds for the perturbation terms from \eqref{Equ:Lemma:BoundsPerturbation:P_z}, \eqref{Equ:Lemma:BoundsPerturbaton:P_s}, and \eqref{Equ:Lemma:BoundsPerturbation:P_v}, step (h) uses the fact that $p^T \mathds{1} = \|p\|_1$.

Next, we derive the bound for $\E P[\bm{w}_{e,i}]$ from
the recursion for $\bm{w}_{e,i}$ in \eqref{Equ:DistProc:JointRec_TF2_we}:
	\begin{align}
		\E &P[\bm{w}_{e,i}]
						\nn\\
				&=		\E P\big[
							\mc{D}_{N-1} \bm{w}_{e,i-1}
							-
							\mc{U}_R 
							\mc{A}_2^T 
							\mc{M}
							\;
							s(\mathds{1} \otimes \bm{w}_{c,i-1})
							\nn\\
							&\qquad
							-
							\mc{U}_R \mc{A}_2^T \mc{M} \bm{z}_{i-1}
							-
							\mc{U}_R \mc{A}_2^T \mc{M} \bm{v}_i
						\big]
						\nonumber\\
				&\overset{(a)}{=}
						\E P\left[
							\mc{D}_{N-1} \bm{w}_{e,i-1}
							-
							\mc{U}_R 
							\mc{A}_2^T 
							\mc{M}
							\;
							\left(
								s(\mathds{1} \otimes \bm{w}_{c,i-1})
								+
								\bm{z}_{i-1}
							\right)
						\right]
						\nn\\
						&\quad
						+
						\E P\left[
							\mc{U}_R \mc{A}_2^T \mc{M} \bm{v}_i
						\right]
						\nonumber\\
				&\overset{(b)}{\preceq}
						\Gamma_e
						\cdot
						\E P\left[
							\bm{w}_{e,i-1}
						\right]
						\nn\\
						&\quad
						+\!
						\frac{2}{1-|\lambda_2(A)|}
						\!\cdot\!
						\E P\left[
							\mc{U}_R \mc{A}_2^T \mc{M}
							\big(
									s(\mathds{1} \otimes \bm{w}_{c,i-1})
									\!+\!
									\bm{z}_{i-1}
							\big)
						\right]
						\nn\\
						&\quad
						+\!
						\E P\left[
							\mc{U}_R \mc{A}_2^T \mc{M} \bm{v}_i
						\right]
						\nonumber\\
				&\overset{(c)}{\preceq}
						\Gamma_e
						\cdot
						\E P\left[
							\bm{w}_{e,i-1}
						\right]
						\nn\\
						&\quad
						+
						\frac{2}{1-|\lambda_2(A)|}
						\cdot
						\left\|
							\bP[
								\mc{U}_R \mc{A}_2^T \mc{M}
							]
						\right\|_{\infty}^2
						\nn\\
						&\qquad
						\cdot
						\mathds{1}\mathds{1}^T
						\cdot
						\E P\left[
									s(\mathds{1} \otimes \bm{w}_{c,i-1})
									+
									\bm{z}_{i-1}
						\right]
						\nonumber\\
						&\quad+
						\left\|
							\bP[
								\mc{U}_R \mc{A}_2^T \mc{M}
							]
						\right\|_{\infty}^2
						\cdot
						\mathds{1}\mathds{1}^T
						\cdot
						\E P\left[
							\bm{v}_i
						\right]
						\nonumber\\
				&\overset{(d)}{\preceq}
						\Gamma_e
						\cdot
						\E P\left[
							\bm{w}_{e,i-1}
						\right]
						\nn\\
						&\quad
						+\!
						\mu_{\max}^2 \cdot
						\frac{4\left\|
							\bP[
								\mc{U}_R \mc{A}_2^T
							]
						\right\|_{\infty}^2}
						{1\!-\!|\lambda_2(A)|}
						\!\cdot\!
						\mathds{1}\mathds{1}^T
						\nn\\
						&\qquad
						\!\cdot\!
						\big\{
							\E P\left[
									s(\mathds{1} \otimes \bm{w}_{c,i-1})
							\right]
							\!+\!
							\E P\left[
										\bm{z}_{i-1}
							\right]
						\big\}
						\nonumber\\
						&\quad+
						\mu_{\max}^2 \cdot
						\left\|
							\bP[
								\mc{U}_R \mc{A}_2^T
							]
						\right\|_{\infty}^2
						\cdot
						\mathds{1}\mathds{1}^T
						\cdot
						\E P\left[
							\bm{v}_i
						\right]
						\nonumber\\
				&\overset{(e)}{\preceq}
						\bigg[
							\Gamma_e
							\!+\!
							4\mu_{\max}^2
							\!\cdot\!
							\left\|
								\bP[
									\mc{U}_R \mc{A}_2^T
								]
							\right\|_{\infty}^2
							\!\cdot\!
							\left\|
								\bP[
									\mc{A}_1^T \mc{U}_L
								]
							\right\|_{\infty}^2
							\lambda_U^2
							N
							\nn\\
							&\qquad
							\times
							\left(
								\frac{1}{1 \!-\! |\lambda_2(A)|} 
								\!+\!
								\frac{\alpha}{\lambda_U^2}
							\right)
							\mathds{1}\mathds{1}^T							
						\bigg]
						\!\cdot\!
						\E P[\bm{w}_{e,i-1}]
						\nonumber\\
						&\quad+
						4\mu_{\max}^2
						\cdot
						\left\|
							\bP[
								\mc{U}_R \mc{A}_2^T
							]
						\right\|_{\infty}^2
						\lambda_U^2
						N
						\left(
							\frac{3}{1-|\lambda_2(A)|}
							+
							\frac{\alpha}{\lambda_U^2}
						\right)
						\nn\\
						&\qquad
						\cdot
						\mathds{1} \cdot
						\E \|\check{\bm{w}}_{c,i-1}\|^2
						\nonumber\\
						&\quad+
						\mu_{\max}^2
						\!\cdot\!
						\left\|
							\bP[
								\mc{U}_R \mc{A}_2^T
							]
						\right\|_{\infty}^2
						\!\cdot\!
						\bigg[
								12
								\frac{\lambda_U^2 \|\tilde{w}_{c,0}\|^2 N \!+\! \mathds{1}^Tg^o}
								{1-|\lambda_2(A)|}
								\nn\\
								&\qquad\qquad
								+\!
								N[
									4 \alpha
									( \| \tilde{w}_{c,0} \|^2 \!+\! \|w^o\|^2 )
									\!+\!
									\sigma_v^2
								]
						\bigg]
						\!\cdot\!
						\mathds{1}
						\nn\\
			&\overset{(f)}{\preceq}
						\bigg[
							\Gamma_e
							\!+\!
							4\mu_{\max}^2
							\!\cdot\!
							\left\|
								\bP[
									\mc{U}_R \mc{A}_2^T
								]
							\right\|_{\infty}^2
							\!\cdot\!
							\left\|
								\bP[
									\mc{A}_1^T \mc{U}_L
								]
							\right\|_{\infty}^2
							\lambda_U^2
							N
							\nn\\
							&\qquad
							\times
							\left(
								\frac{1}{1 \!-\! |\lambda_2(A)|} 
								\!+\!
								\frac{\alpha}{\lambda_U^2}
							\right)
							\mathds{1}\mathds{1}^T							
						\bigg]
						\!\cdot\!
						\E P[\bm{w}_{e,i-1}]
						\nonumber\\
						&\quad+
						4\mu_{\max}^2
						\cdot
						\left\|
							\bP[
								\mc{U}_R \mc{A}_2^T
							]
						\right\|_{\infty}^2
						\lambda_U^2
						N
						\left(
							\frac{3}{1-|\lambda_2(A)|}
							+
							\frac{\alpha}{\lambda_U^2}
						\right)
						\nn\\
						&\qquad
						\cdot
						\mathds{1} \cdot
						\E \|\check{\bm{w}}_{c,i-1}\|^2
						\nonumber\\
						&\quad+
						\mu_{\max}^2
						\!\cdot\!
						N
						\left\|
							\bP[
								\mc{U}_R \mc{A}_2^T
							]
						\right\|_{\infty}^2
						\!\cdot\!
						\bigg[
								12
								\frac{\lambda_U^2 \|\tilde{w}_{c,0}\|^2 \!+\! \|g^o\|_{\infty}}
								{1-|\lambda_2(A)|}
								\nn\\
								&\qquad
								+\!
								4 \alpha
								( \| \tilde{w}_{c,0} \|^2 \!+\! \|w^o\|^2 )
								\!+\!
								\sigma_v^2
						\bigg]
						\!\cdot\!
						\mathds{1}
		\label{Equ:Appendix:EP_wei_recursion_interm1}
	\end{align}
where step (a) uses the additivity property in Lemma \ref{Lemma:BasicPropertiesOperator} 
since the definition of $\bz_{i-1}$ and $\bv_i$ in \eqref{Equ:DistProc:s_hat_phi_decomp_zdef_vdef} and the definitions of $\w_{c,i-1}$ and $\w_{e,i-1}$ in \eqref{Equ:DistProc:w_i_prime_w_ci_w_ei} imply that $\bz_{i-1}$, $\w_{c,i-1}$ and $\w_{e,i-1}$ depend on all  $\{\bw_j\}$ for $j \le i-1$, meaning that the cross terms between $\bv_i$ and all other terms are zero, just as in step (a) of \eqref{Equ:Appendix:EP_wci_recursion_interm1}, step (b) uses the variance relation of stable Kronecker Jordan operators from \eqref{Equ:Propert:StableKronJordanOperator} with $d_2=\lambda_2(A)$, step (c) uses the variance relation of linear operator \eqref{Equ:VarPropt:Linear_ub}, step (d) uses the submultiplictive property \eqref{Equ:Properties:PX_bar_SubMult} and $P[x + y] \preceq 2P[x]+2P[y]$ derived from the convexity property \eqref{Equ:Properties:PX_CvxComb} and the scaling property in  \eqref{Equ:Lemma:BoundsPerturbation:P_z}, \eqref{Equ:Lemma:BoundsPerturbaton:P_s}, and \eqref{Equ:Lemma:BoundsPerturbation:P_v}, step (e) substitutes the bounds on the perturbation terms from \eqref{Equ:Lemma:BoundsPerturbation:P_z}--\eqref{Equ:Lemma:BoundsPerturbation:P_v}, and step (f) uses the inequality $| \one^T g^o | \le N \|g^o\|_{\infty}$.

Finally, using the quantities defined in 
\eqref{Equ:Lemma:IneqRecur:psi0_def}--\eqref{Equ:Lemma:IneqRecur:bve_def},
we can rewrite recursions \eqref{Equ:Appendix:EP_wci_recursion_interm1}
and \eqref{Equ:Appendix:EP_wei_recursion_interm1} as 
	\begin{align}
		\E P[\check{\bm{w}}_{c,i}]
					&\preceq		(\gamma_c \!+\! \mu_{\max}^2 \psi_0) \!\cdot\!
								\E P[\check{\bm{w}}_{c,i-1}]
								\nn\\
								&\quad
								+\!
								(\mu_{\max} h_c(\mu_{\max}) \!+\! \mu_{\max}^2 \psi_0)
								\cdot
								\one^T
								\E P[\bm{w}_{e,i-1}]
								\nn\\
								&\quad
								+\!
								\mu_{\max}^2 b_{v,c}
		\label{Equ:Appendix:EPwc_bound_final}
								\\
		\E P[\bm{w}_{e,i}]
					&\preceq		\mu_{\max}^2 \psi_0
								\one \cdot
								\E P[\check{\bm{w}}_{c,i-1}]
								\nn\\
								&\quad
								+
								(\Gamma_e + \mu_{\max}^2 \psi_0 \mathds{1}\mathds{1}^T)
								\cdot
								\E P[\bm{w}_{e,i-1}]
								\nn\\
								&\quad
								+
								\mu_{\max}^2 b_{v,e} \cdot \one
		\label{Equ:Appendix:EPwe_bound_final}
	\end{align}
where $\E P[\check{\bm{w}}_{c,i}] = \E \|\check{\bm{w}}_{c,i}\|^2$.
Using the matrices and vectors defined in \eqref{Equ:FirstOrderAnal:Gamma_def}--%
\eqref{Equ:FirstOrderAnal:bv_def}, we can write the above 
two recursions in a joint form as in \eqref{Equ:FirstOrderAnal:W_i_prime_ineq_Rec1}.


\section{Proof of Lemma \ref{Lemma:BoundsPerturbation}}
\label{Appendix:Proof_BoundsPerturbation}

First, we establish the bound for $P[\bm{z}_{i-1}]$ in 
\eqref{Equ:Lemma:BoundsPerturbation:P_z}. Substituting \eqref{Equ:DistProc:invTF_wi_wi_prime}
and \eqref{Equ:DistProc:relation_phi_w_wprime} into the definition of $\bm{z}_{i-1}$ in 
\eqref{Equ:DistProc:s_hat_phi_decomp_zdef_vdef} we get:
	\begin{align}
		P&[\bm{z}_{i-1}]				\nn\\
						&\preceq		
									P\big[
										s\left(
											\mathds{1} \!\otimes \! \bm{w}_{c,i-1}
											\!+\!
											(A_1^T U_L \!\otimes\! I_M)
											\bm{w}_{e,i-1}
										\right)
										\!-\!
										s(\mathds{1} \!\otimes\! \bm{w}_{c,i-1})
									\big]
									\nonumber\\
						&\overset{(a)}{\preceq}
									\lambda_U^2
									\cdot
									P\left[
										(A_1^T U_L \otimes I_M) 
										\bm{w}_{e,i-1}
									\right]
									\nonumber\\
						&\overset{(b)}{\preceq}
									\lambda_U^2
									\cdot
									\left\|
										\bP[\mc{A}_1^T \mc{U}_L]
									\right\|_{\infty}^2
									\cdot
									\mathds{1}\mathds{1}^T
									\cdot
									P[\bm{w}_{e,i-1}]
	\end{align}
where step (a) uses the variance relation \eqref{Equ:VarPropt:s},
and step (b) uses property \eqref{Equ:VarPropt:Linear_ub}.

Next, we prove the bound on $P[s(\mathds{1}\otimes \bm{w}_{c,i-1})]$.
It holds that
	\begin{align}
		P&[s(\mathds{1}\otimes \bm{w}_{c,i-1})]	\nonumber\\
			&=		P\Big[
						\frac{1}{3}\cdot 3
						\big(
							s(\mathds{1} \otimes \bm{w}_{c,i-1})
							-
							s(\mathds{1}\otimes \bar{w}_{c,i-1})
						\big)
						\nn\\
						&\qquad
						+
						\frac{1}{3}\cdot 3
						\big(
							s(\mathds{1}\otimes \bar{w}_{c,i-1})
							-
							s(\mathds{1}\otimes w^o)
						\big)
						\nn\\
						&\qquad
						+
						\frac{1}{3}\cdot 3 \cdot
						s(\mathds{1}\otimes w^o)
					\Big]
					\nonumber\\
			&\overset{(a)}{\preceq}
					\frac{1}{3} \cdot P\big[
						 3
						\big(
							s(\mathds{1} \otimes \bm{w}_{c,i-1})
							-
							s(\mathds{1}\otimes \bar{w}_{c,i-1})
						\big)
					\big]
					\nonumber\\
					&\quad
					+
					\frac{1}{3}\cdot P
					\big[
						 3
						\big(
							s(\mathds{1}\otimes \bar{w}_{c,i-1})
							-
							s(\mathds{1}\otimes w^o)
						\big)
					\big]
					\nn\\
					&\quad
					+
					\frac{1}{3}\cdot P
					\big[
						3 \cdot
						s(\mathds{1}\otimes w^o)
					\big]
					\nonumber\\
			&\overset{(b)}{=}
					3P\big[
						s(\mathds{1} \otimes \bm{w}_{c,i-1})
						\!-\!
						s(\mathds{1}\otimes \bar{w}_{c,i-1})
					\big]
					\nn\\
					&\quad
					+
					3P
					\big[
						s(\mathds{1}\otimes \bar{w}_{c,i-1})
						\!-\!
						s(\mathds{1}\otimes w^o)
					\big]
					+
					3P
					\big[
						s(\mathds{1}\otimes w^o)
					\big]
					\nonumber\\
			&\overset{(c)}{\preceq}
					3\lambda_U^2 \cdot
					P\big[
						\mathds{1} \otimes (\bm{w}_{c,i-1} - \bar{w}_{c,i-1})
					\big]
					\nn\\
					&\quad
					+
					3\lambda_U^2 \cdot
					P[\mathds{1}\otimes(\bar{w}_{c,i-1}-w^o)]
					+
					3P[s(\mathds{1}\otimes w^o)]
					\nonumber\\
			&\overset{(d)}{=}
					3\lambda_U^2 \cdot
					\|\check{\bm{w}}_{c,i-1}\|^2
					\cdot \mathds{1}
					+
					3\lambda_U^2 \cdot
					\|\bar{w}_{c,i-1} - w^o\|^2
					\cdot
					\mathds{1}
					\nn\\
					&\quad
					+
					3P[s(\mathds{1}\otimes w^o)]
					\nonumber\\
			&\overset{(e)}{\preceq}
					3\lambda_U^2 \cdot
					\|\check{\bm{w}}_{c,i-1}\|^2
					\!\cdot\! \mathds{1}
					\!+\!
					3\lambda_U^2 \|\tilde{w}_{c,0}\|^2
					\!\cdot\!
					\mathds{1}
					\!+\!
					3P[s(\mathds{1}\otimes w^o)]
	\end{align}
where step (a) uses the convexity property \eqref{Equ:Properties:PX_CvxComb},
step (b) uses the scaling property in Lemma \ref{Lemma:BasicPropertiesOperator},
step (c) uses the variance relation \eqref{Equ:VarPropt:s},
step (d) uses property \eqref{Equ:Properties:PX_Kron},
and step (e) uses the bound \eqref{Equ:Thm:ConvergenceRefRec:NonAsympBound}
and the fact that $\gamma_c<1$.

Finally, we establish the bounds on $P[\bm{v}_i]$ in
\eqref{Equ:Lemma:BoundsPerturbation:P_v_E_Fiminus1}--\eqref{Equ:Lemma:BoundsPerturbation:P_v}. Introduce the $MN \times 1$
vector $\bm{x}$:
	\begin{align}
		\bm{x}	\defeq	\mathds{1} \otimes \bm{w}_{c,i-1}
						+
						\mc{A}_1^T \mc{U}_L \bm{w}_{e,i-1}
				\equiv	\bm{\phi}_{i-1}
	\end{align}
We partition $\bm{x}$ in block form as
$\bm{x} = \col\{\bm{x}_1,\ldots,\bm{x}_N\}$, where each $\bm{x}_k$ is
$M\times 1$. Then, by the definition of $\bm{v}_i$ from 
\eqref{Equ:DistProc:s_hat_phi_decomp_zdef_vdef}, we have
	\begin{align}
		\E \{ P[\bm{v}_i] | \mc{F}_{i-1} \}	
						&=		\E \{ P[\hat{\bm{s}}_i(\bm{x}) - s(\bm{x}) ] | \mc{F}_{i-1} \}
								\nonumber\\
						&=		\col\big\{
									\E \left[ 
										\|
											\hat{\bm{s}}_{1,i}(\bm{x}_1)
											-
											s_1(\bm{x}_1)
										\|^2
										\big|
										\mc{F}_{i-1}
									\right],
									\nn\\
									&\qquad
									\ldots,
									\E
									\left[
										\|
											\hat{\bm{s}}_{N,i}(\bm{x}_N)
											-
											s_N(\bm{x}_N)
										\|^2
										\big|
										\mc{F}_{i-1}
									\right]
								\big\}
								\nonumber\\
						&\overset{(a)}{\preceq}
								\col\big\{
									\alpha 
									\cdot 
									\| \bm{x}_1 \|^2 
									\!+\! 
									\sigma_{v}^2,\;
									\ldots,\;
									\alpha
									\cdot 
									\| \bm{x}_N \|^2 
									\!+\! 
									\sigma_{v}^2
								\big\}
								\nonumber\\
		\label{Equ:Appendix:P_v_bound1}
						&=		\alpha \cdot P[\bm{x}] 
								+ \sigma_v^2 \mathds{1}
	\end{align}
where step (a) uses Assumption \eqref{Equ:Assumption:Randomness:RelAbsNoise}. 
Now we bound $P[\bm{x}]$:
	\begin{align}
		P[ \bm{x} ]
						&=		P\left[
										\mathds{1} \otimes \bm{w}_{c,i-1}
										+
										\mc{A}_1^T \mc{U}_L \bm{w}_{e,i-1}
								\right]
								\nonumber\\
						&=		P\Big[
									\frac{1}{4} \!\cdot\! 4 \!\cdot\! \one \otimes ( \w_{c,i-1} \!-\! \bar{w}_{c,i-1} )
									\!+\!
									\frac{1}{4} \!\cdot\! 4 \!\cdot\! \one \otimes (\bar{w}_{c,i-1} \!-\! w^o)
									\nn\\
									&\qquad
									+\!
									\frac{1}{4} \!\cdot\! 4 \!\cdot\! \mA_1^T \mU_L \w_{e,i-1} 
									\!+\!
									\frac{1}{4} \!\cdot\! 4 \!\cdot\! \one \otimes w^o
								\Big]
								\nonumber\\
						&=
								P\Big[
									\frac{1}{4} \!\cdot\! 4 \!\cdot\! \one \otimes \check{\w}_{c,i-1}
									\!+\!
									\frac{1}{4} \!\cdot\! 4 \!\cdot\! \one \otimes \tilde{w}_{c,i-1}
									\nn\\
									&\qquad
									+\!
									\frac{1}{4} \!\cdot\! 4 \!\cdot\! \mA_1^T \mU_L \w_{e,i-1} 
									\!+\!
									\frac{1}{4} \!\cdot\! 4 \!\cdot\! \one \otimes w^o
								\Big]
								\nn\\
						&\overset{(a)}{\preceq}
								\frac{1}{4} \cdot 4^2
								\!\cdot\! 
								P[ \one \otimes \check{\w}_{c,i-1} ]
								+
								\frac{1}{4} \cdot 4^2
								\!\cdot\! 
								P[ \one \otimes \tilde{w}_{c,i-1} ]
								\nn\\
								&\quad
								+
								\frac{1}{4} \cdot 4^2
								\!\cdot\!  
								P[ \mA_1^T \mU_L \w_{e,i-1} ]
								+
								\frac{1}{4} \cdot 4^2
								\!\cdot\! 
								P[ \one \otimes w^o ]
								\nonumber\\
						&\overset{(b)}{=}
								4 \cdot
								\|
										\check{\bm{w}}_{c,i-1}
								\|^2
								\cdot
								\mathds{1}
								+
								4 \cdot
								\|
									\tilde{w}_{c,i-1}
								\|^2
								\cdot
								\mathds{1}
								\nn\\
								&\quad
									+ 
								4 \cdot
								P[
										\mc{A}_1^T \mc{U}_L \bm{w}_{e,i-1}
								]
								+								
								4 \cdot
								\| w^o \|^2 \cdot \one
								\nonumber\\
						&\overset{(c)}{\preceq}
								4 \cdot
								\|
										\check{\bm{w}}_{c,i-1}
								\|^2
								\!\cdot\!
								\mathds{1}
								\!+\!
								4 \!\cdot\!
								\|
									\bP[ \mc{A}_1^T \mc{U}_L ]
								\|_{\infty}^2
								\!\cdot\!
								\one\one^T
								\!\cdot\!
								P[
										\bm{w}_{e,i-1}
								]
								\nn\\
								&\quad
								+
								4 \cdot
								\|\tilde{w}_{c,0}\|^2
								\cdot
								\mathds{1}
								+ 
								4 \cdot
								\| w^o \|^2 \cdot \one
		\label{Equ:Appendix:P_v_bond2}
	\end{align}
where step (a) uses the convexity property \eqref{Equ:Properties:PX_CvxComb} and the scaling property in Lemma \ref{Lemma:BasicPropertiesOperator}, step (b) uses the Kronecker property \eqref{Equ:Properties:PX_Kron}, step (c) uses the variance relation \eqref{Equ:VarPropt:Linear} and  the bound \eqref{Equ:Thm:ConvergenceRefRec:NonAsympBound}. Substituting \eqref{Equ:Appendix:P_v_bond2} into \eqref{Equ:Appendix:P_v_bound1}, we obtain \eqref{Equ:Lemma:BoundsPerturbation:P_v_E_Fiminus1}, and taking expectation of \eqref{Equ:Lemma:BoundsPerturbation:P_v_E_Fiminus1} with respect to $\mF_{i-1}$ leads to \eqref{Equ:Lemma:BoundsPerturbation:P_v}.

\section{Proof of Theorem \ref{Thm:NonAsymptoticBound}}
\label{Appendix:Proof_Thm_NonAsymptotiBound}

Assume initially that the matrix $\Gamma$ is stable (we show further ahead how the step-size parameter $\mu_{\max}$ can be selected to ensure this property). Then,  we can iterate the inequality recursion \eqref{Equ:FirstOrderAnal:W_i_prime_ineq_Rec1} and obtain
	\begin{align}
		\check{\mW}_i'	&\preceq	\Gamma^i \check{\mW}_0' 
									+ 
									\mu_{\max}^2 \sum_{j=0}^{i-1}
									\Gamma^j b_v
									\nonumber\\
						&\preceq	\sum_{j=0}^{\infty}
									\Gamma^j \check{\mW}_0' 
									+ 
									\mu_{\max}^2 \sum_{j=0}^{\infty}
									\Gamma^j b_v
									\nonumber\\
						&\preceq	(I-\Gamma)^{-1} 
									(\check{\mW}_0' 
									+ 
									\mu_{\max}^2  b_v)
		\label{Equ:Appendix:mWcheckprime_bound_interm1}
	\end{align}
where the first two inequalities use the fact that all entries of $\Gamma$ are nonnegative. Moreover, substituting \eqref{Equ:FirstOrderAnal:Gamma_def} into  \eqref{Equ:FirstOrderAnal:W_i_prime_ineq_Rec1}, we get
	\begin{align}
		\check{\mW}_i'	&\preceq	\Gamma_0 \check{\mW}_{i-1}' 
									+
									\mu_{\max}^2 \psi_0 \one\one^T \check{\mW}_{i-1}' 
									+
									\mu_{\max}^2 b_v
		\label{Equ:Appendix:mWcheckprime_ineqrecur_interm1}
	\end{align}
Substituting \eqref{Equ:Appendix:mWcheckprime_bound_interm1} 
into the second term on the right-hand side of
\eqref{Equ:Appendix:mWcheckprime_ineqrecur_interm1} leads to
	\begin{align}
		\check{\mW}_i'	&\preceq	\Gamma_0 \check{\mW}_{i-1}' + \mu_{\max}^2 \cdot c_v(\mu_{\max})
		\label{Equ:FirstOrderAnal:W_i_prime_IneqRecursion_Gamma0}
	\end{align}
where
	\begin{align}
		c_v(\mu_{\max})	\defeq	\psi_0 \cdot \one^T (I-\Gamma)^{-1}
							(\check{\mW}_0' + \mu_{\max}^2 b_v)
							\cdot \one
							+
							b_v
		\label{Equ:Appendix:bvmu_expr}
	\end{align}
Now iterating \eqref{Equ:FirstOrderAnal:W_i_prime_IneqRecursion_Gamma0} 
leads to the 
following non-asymptotic bound:
	\begin{align}
		\check{\mW_i}'	&\preceq	
									\Gamma_0^i \check{\mW}_0' 
									+
									\sum_{j=0}^{i-1}
									\mu_{\max}^2\Gamma_0^j \cdot c_v(\mu_{\max})
				\preceq
						\Gamma_0^i 
						\check{\mW}_0' 
						+
						\check{\mW}_{\infty}^{\mathrm{ub}'}
		\label{Equ:FirstOrderAnl:W_i_prime_bound_interdiate1}
	\end{align}
where 
	\begin{align}
		\check{\mW}_{\infty}^{\mathrm{ub}'}	
						&\defeq	\mu_{\max}^2 (I-\Gamma_0)^{-1} \cdot c_v(\mu_{\max})
		\label{Equ:FirstOrdarAnal:W_inf_prime_def}
	\end{align}
We now derive the non-asymptotic bounds \eqref{Equ:DistProc:EPwcicheck_bound}--\eqref{Equ:DistProc:EPwei_bound}
from \eqref{Equ:FirstOrderAnl:W_i_prime_bound_interdiate1}.
To this end, we need to study the structure of the term $\Gamma_0^i \check{\mW}_0'$.
Our approach relies on applying the unilateral $z-$transform to 
the causal matrix sequence $\{\Gamma_0^i, i\geq 0\}$ to get
	\begin{align}
		\Gamma_0(z)		\defeq		\mZ\left\{ \Gamma_0^i \right\}	
						=			z(zI - \Gamma_0)^{-1} 
	\end{align}
since $\Gamma_0$ is a stable matrix. Note from \eqref{Equ:FirstOrderAnal:Gamma0_def} that $\Gamma_0$ is a 
$2\times 2$ block upper triangular matrix. Substituting \eqref{Equ:FirstOrderAnal:Gamma0_def} 
into the above expression
and using the formula for inverting $2\times 2$ block upper triangular
matrices (see formula (4) in \cite[p.48]{laub2005matrix}), we obtain
	\begin{align}
		\Gamma_0(z)		&=			\begin{bmatrix}
										\frac{z}{z-\gamma_c}
											&		
													\mu_{\max} h_c(\mu_{\max})\cdot
													\frac{z}{z-\gamma_c}
													\cdot
													\one^T (zI - \Gamma_e)^{-1}
													\\
										0	&		z(zI - \Gamma_e)^{-1}
									\end{bmatrix}
		\label{Equ:Appendix:Gamma0z_interm1}
	\end{align}
Next we compute the inverse $z-$transform to obtain $\Gamma_0^i$. Thus, observe that the inverse $z-$transform of the $(1,1)$ entry, the $(2,1)$ block, and the $(2,2)$ block are the causal sequences  $\gamma_c^i$, $0$, and $\Gamma_e^i$, respectively. For the $(1,2)$ block, it can be expressed in partial fractions as
	\begin{align}
		&\mu_{\max} h_c(\mu_{\max})\cdot
		\frac{z}{z-\gamma_c}
		\cdot
		\one^T (zI - \Gamma_e)^{-1}
						\nn\\
			&=			\mu_{\max} h_c(\mu_{\max})
						\!\cdot\!
						\one^T 
						(\gamma_c I \!-\! \Gamma_e)^{-1}
						\left(
							\frac{z}{z\!-\!\gamma_c} I
							\!-\!
							z(zI \!-\! \Gamma_e)^{-1}
						\right)
						\nonumber
	\end{align}
from which we  conclude that the inverse $z-$transform of the $(1,2)$ block is 
	\begin{align}
		\mu_{\max} h_c(\mu_{\max})
		\!\cdot\!
		\one^T 
		(\gamma_c I \!-\! \Gamma_e)^{-1}
		\left(
			\gamma_c^i I
			\!-\!
			\Gamma_e^i
		\right),
		\quad
		i \ge 0
	\end{align}
It follows that
	\begin{align}
		\Gamma_0^i		&=			\begin{bmatrix}
										\gamma_c^i
											&		\mu_{\max} h_c(\mu_{\max})
													\cdot
													\one^T 
													(\gamma_c I - \Gamma_e)^{-1}
													\left(
														\gamma_c^i I
														-
														\Gamma_e^i
													\right)
													\\
										0	&		\Gamma_e^i
									\end{bmatrix}
		\label{Equ:Appendix:Gamma0_power_i}
	\end{align}
Furthermore, as indicated by \eqref{Equ:LearnBehav:wc0bar_initRefRecur} 
in Sec. \ref{Sec:LearningBehavior:Overview},
the reference recursion \eqref{Equ:LearnBehav:RefRec} is initialized
at the centroid of the network,
i.e., $\bar{w}_{c,0} = \sum_{k=1}^N \theta_k w_{k,0}$. This fact, together
with \eqref{Equ:DistProc:wci_centroid} leads to $\bar{w}_{c,0}=w_{c,0}$,
which means that $\check{\bw}_{c,0}=0$.
As a result, we get the following form for $\mW_0'$:
	\begin{align}
		\label{Equ:FirstOrderAnal:W_0_prime}
		\check{\mW}_{0}'	=	\col\left\{ 0, \; \E P[\bm{w}_{e,0}] \right\}
	\end{align}
Multiplying  \eqref{Equ:Appendix:Gamma0_power_i} to the left of
 \eqref{Equ:FirstOrderAnal:W_0_prime}
gives
	\begin{align}
		\label{Equ:FirstOrderAnl:TransientTerm}
		\Gamma_0^i \check{\mW}_0'	=	\begin{bmatrix}
									\mu_{\max} h_c(\mu_{\max})
									\cdot
									\one^T 
									(\gamma_c I \!-\! \Gamma_e)^{-1}
									\left(
										\gamma_c^i I
										\!-\!
										\Gamma_e^i
									\right)
									\mW_{e,0}\\
									\Gamma_e^{i}
									\mW_{e,0}
								\end{bmatrix}
	\end{align}
where $\mW_{e,0}=\E P[\bw_{e,0}]$. Substituting  \eqref{Equ:FirstOrderAnl:TransientTerm} into
\eqref{Equ:FirstOrderAnl:W_i_prime_bound_interdiate1}, we obtain
	\begin{align}
			\check{\mW}_i'	&\preceq	\begin{bmatrix}
								\mu_{\max} h_c(\mu_{\max})
								\!\cdot\!
								\one^T 
								(\gamma_c I \!-\! \Gamma_e)^{-1}
								\left(
									\gamma_c^i I
									\!-\!
									\Gamma_e^i
								\right)
								\E P[\bw_{e,0}]\\
								\Gamma_e^{i}
								\;
								\E P[\bw_{e,0}]
							\end{bmatrix}
							\nn\\
							&\quad
							+
							\check{\mW}_{\infty}^{\mathrm{ub}'}
		\label{Equ:Thm_NonAsympBound:W_i_prime_bound_final}
	\end{align}
	
Finally, we study the behavior of the asymptotic bound 
$\check{\mW}_{\infty}^{\mathrm{ub}'}$ by calling upon the following lemma.
	\begin{lemma}[Useful matrix expressions]
		\label{Lemma:UsefulMatrixExpr}
		It holds that
			\begin{align}
				&\one^T (I - \Gamma)^{-1}
						\nn\\
					&=	\zeta(\mu_{\max})						
						\nn\\
						&\cdot
						\begin{bmatrix}
							\frac{\mu_{\max}^{-1}}{\lambda_L - \frac{\mu_{\max}}{2}\|p\|_1^2 \lambda_U^2}
							&
							\left(
								1 +
								\frac{h_c(\mu_{\max})}{\lambda_L - \frac{\mu_{\max}}{2}\|p\|_1^2\lambda_U^2}
							\right)
								\one^T (I - \Gamma_e)^{-1}
						\end{bmatrix}
				\label{Equ:Lemma:UsefulMatrixExpr:one_I_Gamma}
						\\
				&(I - \Gamma_0)^{-1}
						\nn\\
					&=	\begin{bmatrix}
							\frac{\mu_{\max}^{-1}}{\lambda_L - \frac{\mu_{\max}}{2}\|p\|_1^2 \lambda_U^2}
							&
							\frac{h_c(\mu_{\max})}
							{\lambda_L \!-\! \frac{\mu_{\max}}{2}\|p\|_1^2\lambda_U^2}
							\one^T (I \!-\! \Gamma_e)^{-1}
							\\
							0
							&
							(I-\Gamma_e)^{-1}
						\end{bmatrix}\!\!
				\label{Equ:Lemma:UsefulMatrixExpr:I_Gamma0}
			\end{align}
		where
			\begin{align}
				\label{Equ:Lemma:UsefulMatrixExpr:zeta}
				\zeta&(\mu_{\max})	
					=		\bigg\{
								1
								-
								\psi_0
								\cdot
								\Big[
									\frac{\mu_{\max}}
									{\lambda_L \!-\! \frac{\mu_{\max}}{2}\|p\|_1^2 \lambda_U^2}
									\nn\\
									&
									+\!
									\mu_{\max}^2\!
									\left(\!
										1 \!+\!
										\frac{h_c(\mu_{\max})}
										{\lambda_L \!-\! \frac{\mu_{\max}}{2}\|p\|_1^2\lambda_U^2}
									\right)\!
										\one^T (I \!-\! \Gamma_e)^{-1} \one
								\Big]\!
							\bigg\}^{\!-1} \!\!
			\end{align}		
	\end{lemma}
	\begin{IEEEproof}
		See Appendix \ref{Appendix:Proof_UsefulMatrixExpr}.
	\end{IEEEproof}
Substituting \eqref{Equ:Appendix:bvmu_expr}, \eqref{Equ:FirstOrderAnal:W_0_prime},
\eqref{Equ:Lemma:UsefulMatrixExpr:one_I_Gamma} and
\eqref{Equ:Lemma:UsefulMatrixExpr:I_Gamma0} into
\eqref{Equ:FirstOrdarAnal:W_inf_prime_def} and after some
algebra, we obtain
	\begin{align}
			\check{\mW}_{\infty}^{\mathrm{ub}'}
						&=	\psi_0 \! \cdot\! \zeta(\mu_{\max})
							f(\mu_{\max})
							\nn\\
							&\quad
							\cdot
							\begin{bmatrix}
								\mu_{\max}
								\frac{1+\mu_{\max} h_c(\mu_{\max}) \one^T (I-\Gamma_e)^{-1}\one}
								{\lambda_L - \frac{\mu_{\max}}{2}\|p\|_1^2\lambda_U^2}
								\\
								\mu_{\max}^2 \cdot
								(I-\Gamma_e)^{-1}
								\one
							\end{bmatrix}	
							\nn\\
							&\quad
							+\!
							\begin{bmatrix}
								\mu_{\max}
								\frac{
										b_{v,c}
										+
										\mu_{\max} h_c(\mu_{\max}) \one^T (I - \Gamma_e)^{-1}\one
										b_{v,e}
								}
								{\lambda_L - \frac{\mu_{\max}}{2}\|p\|_1^2\lambda_U^2}
								\\
								\mu_{\max}^2 b_{v,e}\cdot
								(I-\Gamma_e)^{-1} \one 
							\end{bmatrix}
		\label{Equ:Appendix:W_ub_inf_prime_expr}
	\end{align}
where
	\begin{align}
		f(\mu_{\max})
						&\defeq
									\frac{\mu_{\max} b_{v,c}}
									{\lambda_L \!-\! \frac{\mu_{\max}}{2}\|p\|_1^2\lambda_U^2}
									\nn\\
									&\quad
									+\!
									\left( \!
										1\!+\!
										\frac{h_c(\mu_{\max})}
										{\lambda_L \!-\! \frac{\mu_{\max}}{2}\|p\|_1^2\lambda_U^2}
										\!
									\right)
									\!\cdot\!
									\one^T (I \!-\! \Gamma_e)^{-1} 
									\nn\\
									&\qquad
									\times
									\left(
										\E P[\bw_{e,0}]
										\!+\!
										\mu_{\max}^2\one b_{v,e}
									\right)
		\label{Equ:Appendx:f_mu_def}
	\end{align}
Introduce 
	\begin{align}
		\check{\mW}_{c,\infty}^{\mathrm{ub}'}
						&\defeq
									\psi_0 \! \cdot\! \zeta(\mu_{\max})
									f(\mu_{\max})
									\nn\\
									&\qquad
									\cdot
									\mu_{\max}
									\frac{
										1+\mu_{\max} h_c(\mu_{\max}) 
										\one^T (I-\Gamma_e)^{-1}\one
									}
									{\lambda_L - \frac{\mu_{\max}}{2}\|p\|_1^2\lambda_U^2}
									\nn\\
									&\quad
									+
									\mu_{\max}
									\frac{
											b_{v,c}
											+
											\mu_{\max} h_c(\mu_{\max}) \one^T (I - \Gamma_e)^{-1}\one
											b_{v,e}
									}
									{\lambda_L - \frac{\mu_{\max}}{2}\|p\|_1^2\lambda_U^2}
		\label{Equ:Appendix:mW_c_ub_prime_interm1}
									\\
		\check{\mW}_{e,\infty}^{\mathrm{ub}'}
						&\defeq
									\psi_0 \! \cdot\! \zeta(\mu_{\max})
									f(\mu_{\max})
									\cdot
									\mu_{\max}^2 \cdot
									(I-\Gamma_e)^{-1}
									\one
									\nn\\
									&\quad
									+
									\mu_{\max}^2 b_{v,e}\cdot
									(I-\Gamma_e)^{-1} \one 
		\label{Equ:Appendix:mW_e_ub_prime_interm1}
	\end{align}
Then, we have 
	\begin{align}
		\check{\mW}_{\infty}^{\mathrm{ub}'} 
			= 	\col\{
						\check{\mW}_{c,\infty}^{\mathrm{ub}'}, \; 
						\check{\mW}_{e,\infty}^{\mathrm{ub}'} 
				\}
		\label{Equ:Appendix:mW_ub_prime_blockform}
	\end{align}
Substituting \eqref{Equ:Appendix:mW_ub_prime_blockform} into \eqref{Equ:Thm_NonAsympBound:W_i_prime_bound_final}, we conclude  \eqref{Equ:DistProc:EPwcicheck_bound}--\eqref{Equ:DistProc:EPwei_bound}. Now, to prove \eqref{Equ:DistProc:EPwci_Omu}--\eqref{Equ:DistProc:EPwei_Omu}, it suffices to prove
	\begin{align}
		\lim_{\mu_{\max} \rightarrow 0}
		\frac{\check{\mW}_{c,\infty}^{\mathrm{ub}'}}{\mu_{\max}}
					&=
								\frac{
									\psi_0 \big(\lambda_L \!+\! h_c(0)\big) 
									\one^T (I \!-\! \Gamma_e)^{-1} \mW_{e,0} 
									\!+\! b_{v,c} \lambda_L
								}
							    {\lambda_L^2}
		\label{Equ:Appendix:limRatio_mWc_ub_interm1}
								\\
		\lim_{\mu_{\max} \rightarrow 0}
		\frac{\check{\mW}_{e,\infty}^{\mathrm{ub}'}}{\mu_{\max}^2}
					&=
								\frac{	
										\psi_0 \big( \lambda_L \!+\! h_c(0) \big) 
										\one^T (I-\Gamma_e)^{-1} \mW_{e,0} \!+\! b_{v,e}\lambda_L
									}
									{\lambda_L}
								\nn\\
								&\quad
									\cdot
								(I\!-\!\Gamma_e)^{-\!1} \one
		\label{Equ:Appendix:limRatio_mWe_ub_interm1}
	\end{align}
Substituting \eqref{Equ:Appendix:mW_c_ub_prime_interm1} and \eqref{Equ:Appendix:mW_e_ub_prime_interm1} into the left-hand side of \eqref{Equ:Appendix:limRatio_mWc_ub_interm1} and \eqref{Equ:Appendix:limRatio_mWe_ub_interm1}, respectively, we get
	\begin{align}
		&\lim_{\mu_{\max} \rightarrow 0}
		\frac{\check{\mW}_{c,\infty}^{\mathrm{ub}'}}{\mu_{\max}}
							\nn\\
					&=
							\lim_{\mu_{\max} \rightarrow 0}
							\bigg\{
								\nn\\
								&\qquad\quad
								\psi_0 \! \cdot\! \zeta(\mu_{\max})
								f(\mu_{\max})
								\!\cdot\!
								\frac{1 \!+\! \mu_{\max} h_c(\mu_{\max}) 
								\one^T (I \!-\!\Gamma_e)^{-1}\one}
								{\lambda_L - \frac{\mu_{\max}}{2}\|p\|_1^2\lambda_U^2}
								\nn\\
								&\qquad\quad
								+
								\frac{
										b_{v,c}
										+
										\mu_{\max} h_c(\mu_{\max}) \one^T (I - \Gamma_e)^{-1}\one
										b_{v,e}
								}
								{\lambda_L - \frac{\mu_{\max}}{2}\|p\|_1^2\lambda_U^2}	
							\bigg\}
							\nn\\
					&=
							\psi_0 \! \cdot\! \zeta(0)
							f(0)
							\cdot
							\frac{1}
							{\lambda_L}
							+
							\frac{
									b_{v,c}
							}
							{\lambda_L}	
							\nn\\
					&\overset{(a)}{=}
							\psi_0 \! \cdot\! 
							1
							\cdot
							\left[
								\left( \!
									1\!+\!
									\frac{h_c(0)}
									{\lambda_L}
									\!
								\right)
								\!\cdot\!
								\one^T (I \!-\! \Gamma_e)^{-1}
								\E P[\bw_{e,0}]
							\right]
							\cdot
							\frac{1}
							{\lambda_L}
							+
							\frac{
									b_{v,c}
							}
							{\lambda_L}	
							\nn\\
					&=
							\frac{
								\psi_0 \big(\lambda_L+h_c(0)\big) 
								\one^T (I-\Gamma_e)^{-1} \mW_{e,0} 
								+ b_{v,c} \lambda_L
							}
						    {\lambda_L^2}
						    \\
		&\lim_{\mu_{\max} \rightarrow 0}
		\frac{\check{\mW}_{e,\infty}^{\mathrm{ub}'}}{\mu_{\max}^2}
							\nn\\
					&=
							\lim_{\mu_{\max} \rightarrow 0}
							\big\{
								\psi_0 \! \cdot\! \zeta(\mu_{\max})
								f(\mu_{\max})
								\cdot
								(I-\Gamma_e)^{-1}
								\one
								\nn\\
								&\qquad\quad
								+
								b_{v,e}
								\cdot
								(I-\Gamma_e)^{-1} \one 
							\big\}
							\nn\\
					&=
							\psi_0 \! \cdot\! \zeta(0)
							f(0)
							\cdot
							(I-\Gamma_e)^{-1}
							\one
							+
							b_{v,e}
							\cdot
							(I-\Gamma_e)^{-1} \one 
							\nn\\
					&\overset{(b)}{=}
							\psi_0 
							\! \cdot\! 
							\left[
								\left( \!
									1\!+\!
									\frac{h_c(0)}
									{\lambda_L}
									\!
								\right)
								\!\cdot\!
								\one^T (I \!-\! \Gamma_e)^{-1}
								\E P[\bw_{e,0}]
							\right]
							\cdot
							(I-\Gamma_e)^{-1}
							\one
							\nn\\
							&\quad
							+
							b_{v,e}
							\cdot
							(I-\Gamma_e)^{-1} \one 
							\nn\\
					&=
							\frac{	
									\psi_0 \big( \lambda_L + h_c(0) \big) 
									\one^T (I-\Gamma_e)^{-1} \mW_{e,0} + b_{v,e}\lambda_L
								}
								{\lambda_L}
								\!\cdot\!
							(I\!-\!\Gamma_e)^{-\!1} \one
	\end{align}
where steps (a) and (b) use the expressions for $\zeta(\mu_{\max})$ and $f(\mu_{\max})$ from \eqref{Equ:Lemma:UsefulMatrixExpr:zeta} and \eqref{Equ:Appendx:f_mu_def}.

Now we proceed to prove \eqref{Equ:Thm:NonasymptoticBound:EPwci_check_Omu_bound}. We already know that the second term on the right-hand side of \eqref{Equ:DistProc:EPwcicheck_bound}, $\check{\mW}_{c,\infty}^{\mathrm{ub}'}$, is $O(\mu_{\max})$ because of \eqref{Equ:DistProc:EPwci_Omu}. Therefore, we only need to show that the first term on the right-hand side of \eqref{Equ:DistProc:EPwcicheck_bound} is $O(\mu_{\max})$ for all $i \ge 0$. To this end, it suffices to prove that
	\begin{align}
		\lim_{ \mu_{\max} \rightarrow 0}&
		\frac{
				\left\|
					\mu_{\max} h_c(\mu_{\max})
					\!\cdot\!
					\one^T 
					(\gamma_c I \!-\! \Gamma_e)^{-1}
					\left(
						\gamma_c^i I
						\!-\!
						\Gamma_e^i
					\right)
					\mW_{e,0}
				\right\|
		}{\mu_{\max}}
					\nn\\					
				&\qquad\qquad\le 	\mathrm{constant}
	\end{align}
where the constant on the right-hand side should be independent of $i$. This can be proved as below:
	\begin{align}		
		&\lim_{ \mu_{\max} \rightarrow 0}
		\frac{
				\left\|
					\mu_{\max} h_c(\mu_{\max})
					\!\cdot\!
					\one^T 
					(\gamma_c I \!-\! \Gamma_e)^{-1}
					\left(
						\gamma_c^i I
						\!-\!
						\Gamma_e^i
					\right)
					\mW_{e,0}
				\right\|
		}{\mu_{\max}}
						\nn\\
						&=
							\lim_{ \mu_{\max} \rightarrow 0}
							\left\|							
								h_c(\mu_{\max})
								\!\cdot\!
								\one^T 
								(\gamma_c I \!-\! \Gamma_e)^{-1}
								\left(
									\gamma_c^i I
									\!-\!
									\Gamma_e^i
								\right)
								\mW_{e,0}
							\right\|
							\nn\\
						&\le 
							\lim_{ \mu_{\max} \rightarrow 0}
							|h_c(\mu_{\max})|
								\cdot
								\| \one \|
								\cdot
								\|(\gamma_c I \!-\! \Gamma_e)^{-1}\|
								\nn\\
								&\qquad
								\cdot
								\left(
									\left\|
										\gamma_c^i I
									\right\|
									+
									\left\|
										\Gamma_e^i
									\right\|
								\right)
								\cdot
								\|
								\mW_{e,0}
								\|
								\nn\\
						&\overset{(a)}{\le }
							\lim_{ \mu_{\max} \rightarrow 0}
							|h_c(\mu_{\max})|
								\cdot
								N
								\cdot
								\|(\gamma_c I \!-\! \Gamma_e)^{-1}\|
								\nn\\
								&\qquad
								\cdot
								\left(
									1
									+
									C_e \cdot \big(\rho (\Gamma_e)+\epsilon\big)^i
								\right)
								\cdot
								\|
								\mW_{e,0}
								\|
								\nn\\
						&\overset{(b)}{\le }
							\lim_{ \mu_{\max} \rightarrow 0}
							|h_c(\mu_{\max})|
								\!\cdot\!
								N
								\!\cdot\!
								\|(\gamma_c I \!-\! \Gamma_e)^{-1}\|
								\!\cdot\!
								\left(
									1
									\!+\!
									C_e
								\right)
								\!\cdot\!
								\|
								\mW_{e,0}
								\|
								\nn\\
						&\overset{(c)}{=}
							|h_c(0)|
								\cdot
								N
								\cdot
								\|(I \!-\! \Gamma_e)^{-1}\|
								\cdot
								\left[
									1
									+
									C_e
								\right]
								\cdot
								\|
								\mW_{e,0}
								\|	
								\nn\\
						&=
								\mathrm{constant}
	\end{align}
where step (a) uses $\gamma_c = 1 - \mu_{\max} \lambda_L + \frac{1}{2}\mu_{\max}^2 \lambda_L^2 < 1$ for sufficiently small step-sizes, and uses the property that for any small $\epsilon>0$ there exists a constant $C$ such that $\| X^i \| \le C \cdot [\rho(X)+\epsilon]^i$ for all $i \ge 0$ \cite[p.38]{poliak1987introduction}, step (b) uses the fact that $\rho(\Gamma_e) = |\lambda_2(A)| < 1$ so that $\rho(\Gamma_e) + \epsilon < 1$ for small $\epsilon$ (e.g., $\epsilon = (1-\rho(\Gamma_e))/2$), and step (c) uses $\gamma_c = 1 - \mu_{\max} \lambda_L + \frac{1}{2}\mu_{\max}^2 \lambda_L^2 \rightarrow 1$ when $\mu_{\max} \rightarrow 0$.

It remains to prove that condition \eqref{Equ:Thm_NonAsympBound:StepSize} 
guarantees the stability of the matrix $\Gamma$, i.e., $\rho(\Gamma)<1$.
First, we introduce the diagonal matrices $D_{\epsilon,0} \defeq
\diag\{\epsilon,\cdots,\epsilon^{N-1}\}$ and $D_{\epsilon} = \diag\{1, D_{\epsilon,0}\}$,
where $\epsilon$ is chosen to be 
	\begin{align}
		\epsilon \defeq \frac{1}{4}(1-|\lambda_2(A)|)^2 \le \frac{1}{4}
		\label{Equ:Appendix:epsilon_def}
	\end{align}
It holds that $\rho(\Gamma)=\rho(D_{\epsilon}^{-1}\Gamma D_{\epsilon})$
since similarity transformations do not alter eigenvalues. 
By the definition of $\Gamma$ in \eqref{Equ:FirstOrderAnal:Gamma_def},
we have
	\begin{align}
		D_{\epsilon}^{-1}\Gamma D_{\epsilon}	
				=		D_{\epsilon}^{-1} \Gamma_0 D_{\epsilon}
						+
						\mu_{\max}^2 \psi_0 \cdot D_{\epsilon}^{-1} \one \one^T D_{\epsilon}
	\end{align}
We now recall that the spectral radius of a matrix is upper bounded by any of its
matrix norms. Thus, taking the $1-$norm (the maximum absolute column sum
of the matrix) of both sides of the above 
expression and using the triangle inequality and
the fact that $0 < \epsilon \le 1/4$, we get
	\begin{align}
		\rho(\Gamma)		&=		\rho(D_{\epsilon}^{-1}\Gamma D_{\epsilon})
								\nonumber\\
						&\le 	\|D_{\epsilon}^{-1} \Gamma_0 D_{\epsilon}\|_{1}
								+
								\|\mu_{\max}^2 \psi_0 \cdot D_{\epsilon}^{-1} \one \one^T D_{\epsilon}\|_1
								\nonumber\\
						&\le 	\|D_{\epsilon}^{-1} \Gamma_0 D_{\epsilon}\|_{1}
								+
								\mu_{\max}^2 \psi_0 \cdot
								\|D_{\epsilon}^{-1} \one \one^T D_{\epsilon}\|_1
								\nonumber\\
						&=		\|D_{\epsilon}^{-1} \Gamma_0 D_{\epsilon}\|_{1}
								+
								\mu_{\max}^2 \psi_0 \cdot
								\left(
									1+\epsilon^{-1}+\cdots + \epsilon^{-(N-1)}
								\right)
								\nonumber\\
						&=		\|D_{\epsilon}^{-1} \Gamma_0 D_{\epsilon}\|_{1}
								+
								\mu_{\max}^2 \psi_0 \cdot
								\frac{1 - \epsilon^{-N}}{ 1 - \epsilon^{-1} }
								\nonumber\\
						&=		\|D_{\epsilon}^{-1} \Gamma_0 D_{\epsilon}\|_{1}
								+
								\mu_{\max}^2 \psi_0 \cdot
								\frac{\epsilon(\epsilon^{-N}-1)}{1-\epsilon}
								\nonumber\\
						&\le 	\|D_{\epsilon}^{-1} \Gamma_0 D_{\epsilon}\|_{1}
								+
								\mu_{\max}^2 \psi_0 \cdot
								\frac{\frac{1}{4}(\epsilon^{-N}-1)}{1-\frac{1}{4}}
								\nonumber\\
						&\le 	\|D_{\epsilon}^{-1} \Gamma_0 D_{\epsilon}\|_{1}
								+
								\frac{1}{3}\mu_{\max}^2 \psi_0  
								\epsilon^{-N}
		\label{Equ:Appendix:rhoGamma_UB_interm1}
	\end{align}
Moreover, we can use \eqref{Equ:FirstOrderAnal:Gamma0_def} to write:
	\begin{align}
		D_{\epsilon}^{-1} \Gamma_0 D_{\epsilon}
						=		\begin{bmatrix}
									\gamma_c		&		\mu_{\max} h_c(\mu_{\max}) \one^T D_{\epsilon,0}	\\
									0			&		D_{\epsilon,0}^{-1} \Gamma_e D_{\epsilon,0}			
								\end{bmatrix}
	\end{align}
where (recall the expression for $\Gamma_e$
from \eqref{Equ:Propert:Gamma_e}  where we replace $d_2$ by $\lambda_2(A)$):
	\begin{align}
		&D_{\epsilon,0}^{-1} \Gamma_e D_{\epsilon,0}	
						=				\begin{bmatrix}
           									|\lambda_2(A)| & \frac{1-|\lambda_2(A)|}{2} & &
           														  	\\
           									    & \ddots & \ddots &	\\
           									    & & \ddots & 
           									    	\frac{1-|\lambda_2(A)|}{2}	\\
           									    & & & |\lambda_2(A)|
           								\end{bmatrix}	\!\!
           								\\
         &\mu_{\max} h_c(\mu_{\max}) \one^T D_{\epsilon,0}
         				=				\mu_{\max} h_c(\mu_{\max})
         								\begin{bmatrix}
         									\epsilon
										&\!\!\!
         									\cdots\!\!\!
										&
         									\epsilon^{N-1}
										\end{bmatrix}   \!\!  
	\end{align}
Therefore, the $1$-norm of $D_{\epsilon}^{-1} \Gamma_0 D_{\epsilon}$
can be evaluated as
	\begin{align}
		\|D_{\epsilon,0}^{-1} \Gamma_0 D_{\epsilon,0}\|_1
						&=				\max\Big\{
											\gamma_c,\;
											|\lambda_2(A)| + \mu_{\max} h_c(\mu_{\max})\epsilon,
											\nonumber\\
											&\quad\quad
											\frac{1+|\lambda_2(A)|}{2}
											+					
											\mu_{\max} h_c(\mu_{\max}) \epsilon^2,
											\cdots,
											\nonumber\\
											&\quad\quad
											\frac{1+|\lambda_2(A)|}{2}
											+					
											\mu_{\max} h_c(\mu_{\max}) \epsilon^{N-1}					
										\Big\}		
	\end{align}
Since $0<\epsilon \le 1/4$, we have $\epsilon> \epsilon^2 > \cdots>\epsilon^{N-1}>0$.
Therefore, 
	\begin{align}
		\|D_{\epsilon,0}^{-1} \Gamma_0 D_{\epsilon,0}\|_1
						&=	 			\max\Big\{
											\gamma_c,\;
											|\lambda_2(A)| + \mu_{\max} h_c(\mu)\epsilon,
											\nonumber\\
											&\qquad
											\frac{1+|\lambda_2(A)|}{2}
											+					
											\mu_{\max} h_c(\mu_{\max}) \epsilon^2		
										\Big\}		\!\!
	\end{align}
Substituting the above expression for
$\|D_{\epsilon,0}^{-1} \Gamma_0 D_{\epsilon,0}\|_1$
into \eqref{Equ:Appendix:rhoGamma_UB_interm1} leads to
	\begin{align}
		\rho(\Gamma)		&\le 			\max\Big\{
											\gamma_c,\;
											|\lambda_2(A)| 
											+ \mu_{\max} h_c(\mu_{\max})\epsilon,
											\nonumber\\
											&\quad
											\frac{1 \!+\! |\lambda_2(A)|}{2}
											\!+\!					
											\mu_{\max} h_c(\mu_{\max}) \epsilon^2		
										\Big\}
										\!+ \!
										\frac{1}{3}\mu_{\max}^2 \psi_0  
										\epsilon^{-N} 
	\end{align}
We recall from \eqref{Equ:VarPropt:gamma_c_positive} that $\gamma_c>0$.
To ensure $\rho(\Gamma)<1$, it suffices to require that $\mu_{\max}$ is such that the following conditions are satisfied: 
	\begin{align}
		&\gamma_c + \frac{1}{3}\mu_{\max}^2 \psi_0 \epsilon^{-N}	<	1		
		\label{Equ:Appendix:rhoGamma_ineq1}
		\\
		&|\lambda_2(A)| + \mu_{\max} h_c(\mu_{\max})\epsilon
		+ \frac{1}{3}\mu_{\max}^2 \psi_0  \epsilon^{-N}	<	1
		\label{Equ:Appen:rhoGamma_ineq3}
		\\
		&\frac{1+|\lambda_2(A)|}{2}
		+\mu_{\max} h_c(\mu_{\max}) \epsilon^2 + \frac{1}{3}\mu_{\max}^2 \psi_0 \epsilon^{-N} < 1
		\label{Equ:Appendix:rhoGamma_ineq2}
	\end{align}
We now solve these three inequalities to get a condition on $\mu_{\max}$.
Substituting the expression for $\gamma_c$ from
\eqref{Equ:VarPropt:gamma_c} into \eqref{Equ:Appendix:rhoGamma_ineq1},
we get
	\begin{align}
		1 - \mu_{\max}\lambda_L + 
		\mu_{\max}^2 \Big(\frac{1}{3}\psi_0 \epsilon^{-N} + \frac{1}{2}\|p\|_1^2 \lambda_U^2\Big)
		<1
		\label{Equ:Appendix:rhoGama_stepsizecond1}
	\end{align}
the solution of which is given by
	\begin{align}
		0		<	\mu_{\max}		<	\frac{\lambda_L}
								{\frac{1}{3}\psi_0 \epsilon^{-N} + \frac{1}{2}\|p\|_1^2 \lambda_U^2}
		\label{Equ:Appendix:StepSizeCondition_Sufficient_interm1}
	\end{align}
For \eqref{Equ:Appen:rhoGamma_ineq3}--\eqref{Equ:Appendix:rhoGamma_ineq2}, if we substitute the expression
for $h_c(\mu_{\max})$ from \eqref{Equ:Lemma:IneqRecur:hc_def} into 
\eqref{Equ:Appen:rhoGamma_ineq3}--\eqref{Equ:Appendix:rhoGamma_ineq2}, we get a third-order
inequality in $\mu_{\max}$, which is difficult to solve in closed-form.
However, inequalities
\eqref{Equ:Appen:rhoGamma_ineq3}--\eqref{Equ:Appendix:rhoGamma_ineq2} can be guaranteed
by the following conditions:
	\begin{align}
		\mu_{\max} h_c(\mu_{\max}) \epsilon 	&< \frac{(1\!-\!|\lambda_2(A)|)^2}{4},
		\;
		\frac{\mu_{\max}^2 \psi_0 \epsilon^{-N}}{3} < \frac{1\!-\!|\lambda_2(A)|}{4}
		\label{Equ:Appendix:StepSizeCondition_Sufficient_interm2}
	\end{align}
This is because we would then have:
	\begin{align}
		|\lambda_2(A)| &+ \mu_{\max} h_c(\mu_{\max})\epsilon
		+ \frac{1}{3}\mu_{\max}^2 \psi_0  \epsilon^{-N}
							\nn\\
				&< 
							|\lambda_2(A)| 
							+ 
							\frac{(1\!-\!|\lambda_2(A)|)^2}{4}
							+
							\frac{1\!-\!|\lambda_2(A)|}{4}
							\nonumber\\
				&\le  
							|\lambda_2(A)| 
							+ 
							\frac{1\!-\!|\lambda_2(A)|}{4}
							+
							\frac{1\!-\!|\lambda_2(A)|}{4}
							\nonumber\\
				&=			\frac{1\!+\!|\lambda_2(A)|}{2}
							< 1
	\end{align}
Likewise, by the fact that $0 < \epsilon \le 1/4<1$,
	\begin{align}
		&\frac{1+|\lambda_2(A)|}{2}
		+\mu_{\max} h_c(\mu_{\max}) \epsilon^2 + \frac{1}{3}\mu_{\max}^2 \psi_0 \epsilon^{-N}
							\nn\\
				&\quad< 
							\frac{1+|\lambda_2(A)|}{2}
							+
							\mu_{\max} h_c(\mu_{\max}) \epsilon 
							+ 
							\frac{1}{3}\mu_{\max}^2 \psi_0 \epsilon^{-N}
							\nonumber\\
				&\quad<	 		
							\frac{1+|\lambda_2(A)|}{2}
							+		
							\frac{(1\!-\!|\lambda_2(A)|)^2}{4}
							+
							\frac{1\!-\!|\lambda_2(A)|}{4}
							\nonumber\\
				&\quad\le 			
							\frac{1+|\lambda_2(A)|}{2}
							+		
							\frac{1\!-\!|\lambda_2(A)|}{4}
							+
							\frac{1\!-\!|\lambda_2(A)|}{4}			
							\nonumber\\
				&\quad=			
							1			
	\end{align}
Substituting \eqref{Equ:Lemma:IneqRecur:hc_def} and \eqref{Equ:Appendix:epsilon_def} into \eqref{Equ:Appendix:StepSizeCondition_Sufficient_interm2}, we find that the latter conditions are satisfied for
	\begin{align}
		\begin{cases}
			\displaystyle
			0	<	\mu_{\max}		<	\frac{\lambda_L}
								{\|p\|_1^2 \lambda_U^2 
								\big(\|\bP[\mA^T \mU_2]\|_{\infty}^2\!+\!\frac{1}{2}
								\big)}
								\\
								\\
			\displaystyle
			0	<	\mu_{\max}		<	\sqrt{\frac{3(1-|\lambda_2(A)|)^{2N+1}}{2^{2N+2}\psi_0}}
		\end{cases}
		\label{Equ:Appendix:rhoGama_stepsizecond3}
	\end{align}
Combining \eqref{Equ:Appendix:rhoGamma_ineq1}, 
\eqref{Equ:Appendix:rhoGamma_ineq2} and \eqref{Equ:Appendix:rhoGama_stepsizecond3}
, we arrive at condition \eqref{Equ:Thm_NonAsympBound:StepSize}.


\section{Proof of Lemma \ref{Lemma:UsefulMatrixExpr}}
\label{Appendix:Proof_UsefulMatrixExpr}
Applying the matrix inversion lemma \cite{horn1990matrix} to \eqref{Equ:FirstOrderAnal:Gamma_def}, we 
get
	\begin{align}
		(I-\Gamma)^{-1}	
				&=		(I - \Gamma_0 - \mu_{\max}^2 \psi_0 \cdot \one\one^T)^{-1}
						\nonumber\\
				&=		(I \!-\! \Gamma_0)^{-1}
						\!+\!						
						\frac{\mu_{\max}^2\psi_0
						\cdot
						(I \!-\! \Gamma_0)^{-1} \one
						\one^T (I \!-\! \Gamma_0)^{-1}
						}
						{1 \!-\! \mu_{\max}^2\psi_0\cdot\one^T(I-\Gamma_0)^{-1}\one}
	\end{align}
so that
	\begin{align}
		\label{Equ:Appendix:I_Gamma_temp1}
		\one^T (I-\Gamma)^{-1}
				&=		\frac{1}
						{1-\mu_{\max}^2\psi_0\cdot\one^T(I-\Gamma_0)^{-1}\one}
						\cdot
						\one^T (I-\Gamma_0)^{-1}
	\end{align}
By \eqref{Equ:FirstOrderAnal:Gamma0_def}, the matrix
$\Gamma_0$ is a $2\times 2$ block upper triangular matrix whose inverse is given by
	\begin{align}
		(I-\Gamma_0)^{-1}
				&=		\begin{bmatrix}
							(1-\gamma_c)^{-1}	
								&	\frac{\mu_{\max} h_c(\mu_{\max})}{1-\gamma_c}		
									\one^T (I-\Gamma_e)^{-1})
									\\
							0	&	(I-\Gamma_e)^{-1}		
						\end{bmatrix}
		\label{Equ:Appendix:I_minus_Gamma0_inv}
	\end{align}
Substituting \eqref{Equ:VarPropt:gamma_c} into the above expression
leads to \eqref{Equ:Lemma:UsefulMatrixExpr:I_Gamma0}.
Furthermore, from \eqref{Equ:Lemma:UsefulMatrixExpr:I_Gamma0},
we have
	\begin{align}
		\one^T (I-\Gamma_0)^{-1}\one
				&=			\frac{\mu_{\max}^{-1}}
							{
								\lambda_L
								\!-\! 
								\frac{\mu_{\max}}{2}\|p\|_1^2 \lambda_U^2
							}
							\nn\\
							&\quad
							+\!
							\left(
								1 \!+\!
								\frac{h_c(\mu_{\max})}
								{
									\lambda_L 
									\!-\! 
									\frac{\mu}{2}\|p\|_1^2\lambda_U^2
								}
							\right)
								\one^T (I \!-\! \Gamma_e)^{-1} \one		
		\label{Equ:Appendix:I_minus_Gamma0_expr_final}
	\end{align}
Substituting \eqref{Equ:Appendix:I_minus_Gamma0_expr_final} into
\eqref{Equ:Appendix:I_Gamma_temp1}, we obtain
\eqref{Equ:Lemma:UsefulMatrixExpr:one_I_Gamma}.

\section{Proof of Theorem \ref{Thm:LearnBehav_wki}}
\label{Appendix:Proof_LearnBehav_wki}

Taking the squared Euclidean norm of both sides of \eqref{Equ:DistProc:wki_tilde_decomposition_final} and applying the expectation operator, we obtain
	\begin{align}
		\Expt \| \tilde{\w}_{k,i} \|^2
				&=	
						\| \tilde{w}_{c,i} \|^2
						+
						\Expt \left\|
							\check{\w}_{c,i} + (u_{L,k} \otimes I_M) \w_{e,i}
						\right\|^2
						\nn\\
						&\quad
						-
						2 \tilde{w}_{c,i}^T 
						\left[
							\Expt \check{\w}_{c,i}
							+
							(u_{L,k} \otimes I_M) \w_{e,i}
						\right]
	\end{align}
which means that, for all $i \ge 0$,
	\begin{align}
		\big|
			&\Expt \| \tilde{\w}_{k,i} \|^2
			-
			\| \tilde{w}_{c,i} \|^2
		\big|				
							\nn\\
				&= 
							\Big|
								\Expt \left\|
									\check{\w}_{c,i} + (u_{L,k} \otimes I_M) \w_{e,i}
								\right\|^2
								\nn\\
								&\quad
								-
								2 \tilde{w}_{c,i}^T 
								\big[
									\Expt \check{\w}_{c,i}
									+
									(u_{L,k} \otimes I_M) \Expt\w_{e,i}
								\big]
							\Big|
							\nn\\
				&\overset{(a)}{\le }
							\Expt \left\|
									\check{\w}_{c,i} + (u_{L,k} \otimes I_M) \w_{e,i}
							\right\|^2
							\nn\\
							&\quad
							+
							2 
							\|\tilde{w}_{c,i}\|
							\cdot
							\big[
								\left\|
									\Expt \check{\w}_{c,i}
								\right\|
								+
								\|u_{L,k} \otimes I_M\|
								\cdot
								\| \Expt\w_{e,i} \|
							\big]
							\nn\\
				&\overset{(b)}{\le}
							2\Expt \|
								\check{\w}_{c,i}
							\|^2 
							+ 
							2\|u_{L,k} \otimes I_M\|^2
							\cdot 
							\Expt
							\| \w_{e,i} \|^2
							\nn\\
							&\quad
							+
							2 
							\|\tilde{w}_{c,i}\|
							\cdot
							\big[
								\Expt \|\check{\w}_{c,i}\|
								+
								\|u_{L,k} \otimes I_M\|
								\cdot
								\Expt \| \w_{e,i} \|
							\big]
							\nn\\
				&\overset{(c)}{\le}
							2\Expt \|
								\check{\w}_{c,i}
							\|^2 
							+ 
							2\|u_{L,k} \otimes I_M\|^2
							\cdot 
							\Expt
							\| \w_{e,i} \|^2
							\nn\\
							&\quad
							+
							2 
							\|\tilde{w}_{c,i}\|
							\cdot
							\big[
								\sqrt{
									\Expt \|\check{\w}_{c,i}\|^2
								}
								+
								\|u_{L,k} \otimes I_M\|
								\cdot
								\sqrt{
									\Expt \| \w_{e,i} \|^2
								}
							\big]
							\nn\\
				&\overset{(d)}{=}
							2\Expt P[
								\check{\w}_{c,i}
							]
							+ 
							2\|u_{L,k} \otimes I_M\|^2
							\cdot 
							\one^T
							\Expt
							P[ \w_{e,i} ]
							\nn\\
							&\quad
							+
							2 
							\|\tilde{w}_{c,i}\|
							\cdot
							\Big[
								\sqrt{
									\Expt P[\check{\w}_{c,i}]
								}
								+
								\|u_{L,k} \otimes I_M\|
								\cdot
								\sqrt{
									\one^T \Expt P[ \w_{e,i} ]
								}
							\Big]
							\nn\\
				&\overset{(e)}{\le }
							O(\mu_{\max}) 
							+ 
							2 \| u_{L,k} \otimes I_M \|^2
							\cdot							 
							\big(
								\one^T\Gamma_e^i \mW_{e,0}
								+
								\one^T\check{\mW}_{e,\infty}^{\mathrm{ub}'}
							\big)
							\nn\\
							&\quad
							+
							2 \gamma_c^i \|\tilde{w}_{c,0} \|
							\cdot
							\Big[
								O(\mu_{\max}^{\frac{1}{2}})
								\nn\\
								&\qquad\qquad
								+
								\| u_{L,k} \otimes I_M \|
								\cdot
								\sqrt{
									\one^T\Gamma_e^i \mW_{e,0}
									+
									\one^T\check{\mW}_{e,\infty}^{\mathrm{ub}'}
								}
							\Big]
							\nn\\
				&\overset{(f)}{\le }
							O(\mu_{\max}) 
							+ 
							2 \| u_{L,k} \otimes I_M \|^2
							\cdot							 
							\big(
								\one^T\Gamma_e^i \mW_{e,0}
								+
								O(\mu_{\max}^2)
							\big)
							\nn\\
							&\quad
							+
							2 \gamma_c^i \|\tilde{w}_{c,0} \|
							\cdot
							\bigg[
								O(\mu_{\max}^{\frac{1}{2}})
								\nn\\
								&\qquad\qquad
								+
								\| u_{L,k} \otimes I_M \|
								\cdot
								\left(
									\sqrt{
										\one^T\Gamma_e^i \mW_{e,0}
									}
									+
									\sqrt{
										O(\mu_{\max}^2)
									}
								\right)
							\bigg]
							\nn\\
				&= 
							2 \| u_{L,k} \otimes I_M \|^2
							\cdot	
							\one^T\Gamma_e^i \mW_{e,0}
							\nn\\
							&\qquad\qquad
							+
							2\gamma_c^i
							\cdot
							\|\tilde{w}_{c,0}\|
							\cdot
							\| u_{L,k} \otimes I_M \|
							\cdot
							\sqrt{
									\one^T\Gamma_e^i \mW_{e,0}
							}
							\nn\\
							&\quad
							+
							2 \gamma_c^i \|\tilde{w}_{c,0} \|
							\cdot
							\left[
								O(\mu_{\max}^{\frac{1}{2}})
								+
								\| u_{L,k} \otimes I_M \|
								\cdot
								O(\mu_{\max})
							\right]
							\nn\\
							&\quad
							+
							O(\mu_{\max})
							+
							O(\mu_{\max}^2)
							\nn\\
				&\overset{(g)}{\le}
							2 \| u_{L,k} \otimes I_M \|^2
							\! \cdot \!
							\one^T\Gamma_e^i \mW_{e,0}
							\nn\\
							&\quad
							+\!
							2
							\|\tilde{w}_{c,0}\|
							\!\cdot\!
							\| u_{L,k} \otimes I_M \|
							\!\cdot\!
							\sqrt{
									\one^T\Gamma_e^i \mW_{e,0}
							}
							\nn\\
							&\quad
							+\!
							\gamma_c^i							
							\!\cdot\!
							O(\mu_{\max}^{\frac{1}{2}})
							\!+\!
							O(\mu_{\max})
		\label{Equ:Appendix:DifferenceLearningCurve_interm1}
	\end{align}
where step (a) used Cauchy-Schwartz inequality, step (b) used $\|x+y\|^2\leq 2\|x\|^2+2\|y\|^2$, step (c) applied Jensen's inequality to the concave function $\sqrt{\cdot}$, step (d) used property \eqref{Equ:Properties:PX_EuclNorm}, step (e) substituted the non-asymptotic bounds \eqref{Equ:DistProc:EPwcicheck_bound} and \eqref{Equ:DistProc:EPwei_bound} and the fact that $\E P[\check{\w}_{c,i}] \le O(\mu_{\max})$ for \emph{all} $i \ge 0$ from \eqref{Equ:Thm:NonasymptoticBound:EPwci_check_Omu_bound}, step (f) used \eqref{Equ:DistProc:EPwei_Omu} and the fact that $\sqrt{x+y} \le \sqrt{x}+\sqrt{y}$ for $x,y \ge 0$,
and step (g) used $\gamma_c<1$ for sufficiently small step-sizes (guaranteed by \eqref{Equ:Thm:ConvergenceRefRec:StepSize}).


\bibliographystyle{IEEEbib}
\bibliography{DistOpt,IT,DictLearn}

\end{document}